\documentclass[10pt,reqno,a4paper]{amsart}
\usepackage[british]{babel}
\usepackage{amssymb,amstext,amsthm,eucal,bbm,bm, mathrsfs,amscd,caption}
\usepackage{bbold}%,relsize}
\usepackage{stmaryrd}
\usepackage[hmargin=1in,bmargin=1in,tmargin=1in]{geometry}
\usepackage{array,slashed}
\usepackage[svgnames,table]{xcolor}
\usepackage[utf8]{inputenc}
\usepackage[small]{eulervm}
\usepackage{enumitem,multicol,cite}
\usepackage[unicode]{hyperref}
\hypersetup{%
  pdftitle = {},
  pdfkeywords = {Lie superalgebras, R-symmetry, central charges, triple systems, 3-algebras},
  pdfauthor  = {Paul de Medeiros},
  pdfcreator = {\LaTeX\ with package \flqq hyperref\frqq},
  linkcolor=NavyBlue,
  citecolor=ForestGreen,
  urlcolor=OrangeRed,
  anchorcolor=OrangeRed,
  colorlinks=true
}
\newcommand{\half}{\tfrac12}

\newcommand{\cS}{S}
\newcommand{\cV}{U}

\newcommand{\fg}{\mathfrak{g}}

\newcommand{\fa}{\mathfrak{a}}

\newcommand{\fG}{\mathfrak{G}}

\newcommand{\fM}{\mathfrak{M}}
\newcommand{\fP}{\mathfrak{P}}

\newcommand{\fU}{\mathfrak{U}}
\newcommand{\fV}{\mathfrak{V}}
\newcommand{\fW}{\mathfrak{W}}

\newcommand{\fgl}{\mathfrak{gl}}
\newcommand{\fh}{\mathfrak{h}}

\newcommand{\fk}{\mathfrak{k}}

\newcommand{\fp}{\mathfrak{p}}
\newcommand{\fr}{\mathfrak{r}}
\newcommand{\fs}{\mathfrak{s}}

\newcommand{\fz}{\mathfrak{z}}
\newcommand{\fso}{\mathfrak{so}}
\newcommand{\fosp}{\mathfrak{osp}}

\newcommand{\fsl}{\mathfrak{sl}}
\newcommand{\fsp}{\mathfrak{sp}}
\newcommand{\fsu}{\mathfrak{su}}
\newcommand{\fpsu}{\mathfrak{psu}}
\newcommand{\fu}{\mathfrak{u}}

\newcommand{\Cl}{\mathrm{C}\ell}

\newcommand{\SL}{\mathrm{SL}}

\newcommand{\RR}{\mathbb{R}}
\newcommand{\CC}{\mathbb{C}}

\renewcommand{\dim}{\mathrm{dim}}

\DeclareMathOperator{\Hom}{Hom}
\DeclareMathOperator{\End}{End}
\DeclareMathOperator{\Der}{Der}
\DeclareMathOperator{\Mat}{Mat}

\DeclareMathOperator{\ad}{ad}

\DeclareMathOperator{\tr}{tr}

\allowdisplaybreaks[1]
\numberwithin{equation}{section}
\newtheorem{thm}{Theorem}[section]
\newtheorem{prop}[thm]{Proposition}
\newtheorem{lem}[thm]{Lemma}
\newtheorem{cor}[thm]{Corollary}
\theoremstyle{definition}

\newtheorem{example}[thm]{Example}

\theoremstyle{remark}
\newtheorem{remark}[thm]{Remark}
\begin{document}

\title{Poincar\'{e} superalgebras and triple systems}
\author{Paul de Medeiros}
\address{Department of Mathematics and Physics, University of Stavanger, 4036 Stavanger, Norway}
\email{\href{mailto:paul.demedeiros@uis.no}{paul.demedeiros[at]uis.no}, ORCID: \href{https://orcid.org/0000-0002-2730-7251}{0000-0002-2730-7251}}
\date{\today}
\begin{abstract}
We consider a class of Poincar\'{e} superalgebras for which the nested bracket of three supercharges is necessarily zero only in dimensions greater than three. In lower dimensions, we give a precise characterisation of the data which encodes any such Poincar\'{e} superalgebra in terms of a more elementary embedding superalgebra. Up to isomorphism, we classify every classical embedding superalgebra that defines a Poincar\'{e} superalgebra. More generally, we show how to construct an embedding superalgebra in dimensions one, two and three from a certain type of triple system whose product structure encodes the nested bracket of three supercharges in the associated Poincar\'{e} superalgebra. 
\end{abstract}

\maketitle
\vspace*{-1cm}
\tableofcontents

%%%%%%%%%%%%%%%%%%%%%%

\section{Introduction and summary}
\label{sec:IntroductionAndSummary} 

The Coleman-Mandula theorem \cite{Coleman:1967ad} famously implies that the Lie group of symmetries of a non-trivial S-matrix for an interacting relativistic quantum field theory with massive states in four-dimensional Minkowski space must (under very mild assumptions) have a Lie algebra that is isomorphic to the direct sum of the Poincar\'{e} algebra and the Lie algebra of a compact internal symmetry group. 

For supersymmetries described by a Lie superalgebra whose even part is of the form above, the Haag-\L{}opusza\'{n}ski-Sohnius theorem \cite{Haag:1974qh} (see also section 2.6 in \cite{Sohnius:1985qm}) further implies that the odd part of the Lie superalgebra must be a direct sum of tensor products of a spinor representation of the Lorentz algebra and a representation of the internal algebra. Translations in the Poincar\'{e} algebra must commute with supercharges in the odd part of the Lie superalgebra. Moreover, the image of the odd-odd bracket for the Lie superalgebra must contain only translations in the Poincar\'{e} algebra and central charges in the abelian part of the internal algebra, meaning that the nested bracket of three supercharges is always zero. The R-symmetry is described by the quotient of the internal algebra by the ideal spanned by the central charges. Any such Lie superalgebra is generically referred to as a {\emph{Poincar\'{e} superalgebra}}.

The Coleman-Mandula theorem remains valid in any dimension $n>2$, but not in lower dimensions.
\!\footnote{This is because the theorem assumes that the S-matrix is an analytic function of the scattering angles, and there is only forward and backward scattering in dimensions one and two.}
The generalisation of the Haag-\L{}opusza\'{n}ski-Sohnius theorem for $n>2$ was addressed by Nahm in \cite{Nahm:1977tg} (see also \cite{Strathdee:1986jr}). Making similar assumptions to those employed in \cite{Haag:1974qh}, Nahm proved that any Poincar\'{e} superalgebra in $n>3$ has the same structural properties as those outlined in the paragraph above (see Propositions~2.4 and 2.5  in \cite{Nahm:1977tg}). The same is true in $n=3$, but with one curious exception where the odd-odd bracket contains terms in the internal algebra that are not central charges, and so the nested bracket of three supercharges is not zero!
\!\footnote{This property is novel for Poincar\'{e} superalgebras, but necessarily ubiquitous for conformal superalgebras \cite{Nahm:1977tg}.}
In this case, the Poincar\'{e} superalgebra is encoded by (the direct sum of copies of) the simple Lie superalgebra $\fpsu (2|2)$. The even part of $\fpsu (2|2)$ describes the internal algebra $\fsu (2) \oplus \fsu (2)$ which acts non-trivially on the supercharges that are valued in the odd part $({\bm 2},{\bm 2}) \oplus ({\bm 2},{\bm 2})$. Despite being somewhat exotic, this Poincar\'{e} superalgebra is known to have a number of interesting realisations in string and M-theory \cite{Blau:2001ne,Itzhaki:2005tu,Lin:2005nh,Hosomichi:2008qk,Bergshoeff:2008ta,Cordova:2016xhm}. 
\!\footnote{It is important to emphasise that there exist even more exotic Lie superalgebras in diverse dimensions which contain the Poincar\'{e} algebra and so-called \lq brane charges' which transform in non-trivial representations of the Lorentz algebra \cite{vanHolten:1982mx,Lugo:1997ac}. Despite flagrantly violating the Coleman-Mandula theorem, these Lie superalgebras are important since they describe supersymmetries of extended objects carrying topological charges in string and M-theory \cite{deAzcarraga:1989mza}. We will not refer to these Lie superalgebras as Poincar\'{e} superalgebras.}

In dimensions one and two, supersymmetry is typically realised using a Poincar\'{e} superalgebra with the properties noted above, where the nested bracket of three supercharges is zero (see \cite{Strathdee:1986jr}). Unlike in higher dimensions where physical considerations can force this to be the case, here it is just a convenient choice. An exception to this paradigm in $n=1$ was pioneered by Smilga in \cite{Smilga:2003gu} who constructed a model of quantum mechanics with so-called \lq weak supersymmetry', where the Poincar\'{e} superalgebra is a one-dimensional central extension of the simple Lie superalgebra $\fsu (2|1)$. The extending element is identified with the Hamiltonian which spans the one-dimensional Poincar\'{e} algebra in $n=1$. The even part of $\fsu (2|1)$ describes the internal algebra $\fsu (2) \oplus \fu (1)$ which acts non-trivially on the supercharges that are valued in the odd part ${\bm 2}_{+1} \oplus {\bm 2}_{-1}$. Similar models have subsequently been constructed based on (central extensions of) the simple Lie superalgebras $\fpsu (2|2)$ \cite{Ivanov:2016hoe}, $\fsu (4|1)$ \cite{Ivanov:2019gxo} and $\fsu(N|1)$ \cite{Sidorov:2019rvo,Smilga:2022dcd}. 

These more exotic Poincar\'{e} superalgebras in low dimensions can also have interesting avatars in higher dimensions. It is well-known that the rigid supersymmetries of a supergravity background generate a Lie superalgebra called a {\emph{Killing superalgebra}}. For example, new minimal supergravity in $n=4$ has a maximally supersymmetric background $\RR \times S^3$ whose Killing superalgebra is the direct sum of $\fsu (2)$ and the  one-dimensional central extension of $\fsu (2|1)$ mentioned above (see section 6 in \cite{Festuccia:2011ws}). The even part of this Killing superalgebra is the direct sum of the Killing algebra $\RR \oplus \fsu (2) \oplus \fsu(2)$ of isometries of $\RR \times S^3$ and the R-symmetry algebra $\fu(1)$. Remarkably, this Killing superalgebra can therefore be thought of as a Poincar\'{e} superalgebra in $n=1$. 

An algebraic method based on Spencer cohomology was proposed by Figueroa-O'Farrill and Santi in \cite{Figueroa-OFarrill:2015rfh,Figueroa-OFarrill:2015tan} for systematically constructing Killing superalgebras in $n=11$. It has since been applied in a number of different contexts in lower dimensions \cite{deMedeiros:2016srz,deMedeiros:2018ooy,Beckett:2021cwx,Beckett:2024loh,Beckett:2024msc,Beckett:2024qxu,Pdm3dNext}. In all cases, the method recovers the Killing superalgebra associated with every supergravity background as a filtered subdeformation of a Poincar\'{e} superalgebra. In some cases \cite{deMedeiros:2018ooy,Beckett:2021cwx,Beckett:2024loh,Pdm3dNext}, it also produces Killing superalgebras of backgrounds whose supergravity origin (assuming it exists) is unclear. This typically occurs when the Poincar\'{e} superalgebra being deformed has a non-trivial R-symmetry. For example, by deforming the Poincar\'{e} superalgebra  in $n=6$ with $\fsp (1)$ R-symmetry, one can obtain a Killing superalgebra of $\RR \times S^5$ with maximal supersymmetry \cite{deMedeiros:2018ooy}, but $\RR \times S^5$ is not a maximally supersymmetric background of the expected minimal $(1,0)$ supergravity theory. 

In low dimensions, one can also obtain the more exotic Poincar\'{e} superalgebras mentioned above as filtered subdeformations of a conventional Poincar\'{e} superalgebra. For example, by deforming the conventional Poincar\'{e} superalgebra in $n=3$ with $\fso (4)$ R-symmetry, one can recover the exotic Poincar\'{e} superalgebra based on $\fpsu (2|2)$ as a Killing superalgebra of Minkowski space \cite{Pdm3dNext}. 

The goal of this paper is to provide a framework for describing Poincar\'{e} superalgebras that is rich enough to accommodate the examples and potential applications highlighted above. To this end, we shall consider Poincar\'{e} superalgebras which differ from their conventional counterparts only in that their internal symmetries will not be assumed to be compact and the image of their odd-odd brackets may contain terms in the internal algebra which need not be central charges. If the internal symmetry is compact then any such Poincar\'{e} superalgebra corresponds to a filtered deformation of a conventional Poincar\'{e} superalgebra. For convenience, we will work exclusively over the field of complex numbers. This is simply because a number of structural results that we will utilise are either invalid or, at best, significantly more complicated to state over the reals. Having said that, it is straightforward to take real forms of all the complex Poincar\'{e} superalgebras we shall consider. 

The proof of Proposition~2.5 in \cite{Nahm:1977tg} for $n \geq 4$ relies on the assumption that the internal symmetry group is compact, meaning that its unitary representations are completely reducible. However, by explicitly solving the Jacobi identity for Poincar\'{e} superalgebras in dimension $n$, we will arrive at the same conclusion: that the nested bracket of three supercharges is necessarily zero if $n \geq 4$. For $n<4$, we identify precisely which conditions must be met in order to solve the Jacobi identity.

In order to characterise the solutions of these conditions, we first show that every Poincar\'{e} superalgebra can be expressed as a certain abelian extension (by the translation ideal) of a semidirect sum of the Lorentz algebra and a more elementary {\emph{embedding superalgebra}} which contains the internal algebra and the supercharges. Any embedding superalgebra which defines a Poincar\'{e} superalgebra in $n$ dimensions will be called {\emph{$n$-admissible}} and we give a precise characterisation of all $n$-admissible Lie superalgebras with $n<4$. We will also classify (up to isomorphism) every classical $n$-admissible Lie superalgebra.

More generally, we will show how to construct an $n$-admissible Lie superalgebra with $n<4$ from a certain type of triple system that is equipped with a Lie algebra of derivations and a nondegenerate derivation-invariant symmetric bilinear form. Provided the even part of the Lie superalgebra acts faithfully on the odd part, we prove that this construction can be strengthened to a bijective correspondence. For $n=1$, the construction involves an anti-Lie triple system \cite{FaulknerFerrar} and we give examples wherein the anti-Lie triple system is itself encoded by a more elementary anti-Jordan triple system \cite{Kamiya}. For $n>1$, the anti-Lie triple system must be of a particular type which we prove is equivalent to either a polarised anti-Jordan triple system \cite{Kamiya} (if $n=2$) or a Filippov triple system \cite{Filippov} (if $n=3$). Explicit examples are given in these cases too.

The organisation of this paper is as follows. In Section~\ref{sec:NotationAndConventions} we outline our notation and conventions. In Section~\ref{sec:TheCliffordAlgebraAndSpinorRepresentations} we summarise those aspects of Clifford algebras, spinor representations and spinorial bilinear forms that are needed to construct a Poincar\'{e} superalgebra. In Section~\ref{sec:ExtendedPoincareSuperalgebras} we define the Poincar\'{e} superalgebra and prove that it is a Lie superalgebra if and only if the conditions in Theorem~\ref{thm:111Jacobi} are met. In Section~\ref{sec:CentralChargesAndRSymmetry} we briefly recap some details about central charges and R-symmetry in conventional Poincar\'{e} superalgebras in order to contextualise the unconventional examples we shall go on to construct. In Section~\ref{sec:DeconstructionAndReconstructionViaLieSuperalgebraExtensions} we show how to encode a Poincar\'{e} superalgebra in terms of an embedding superalgebra and define a notion of equivalence between Poincar\'{e} superalgebras based on this construction. We then characterise $n$-admissible Lie superalgebras with $n<4$ and classify (up to isomorphism) all classical $n$-admissible Lie superalgebras which define inequivalent Poincar\'{e} superalgebras. In Section~\ref{sec:EmbeddingLSAFromTripleSystems} we describe the construction of $n$-admissible Lie superalgebras with $n<4$ in terms of triple systems. Appendix~\ref{sec:InvariantBilinearFormsOnRepresentationsOFSLW} contains some technical results that are needed in the proof of Theorem~\ref{thm:ClassicalLSA1Admissible}.

\newpage

%%%%%%%%%%%%%%%%%%%%%%

\section{Notation and conventions}
\label{sec:NotationAndConventions}

Unless stated otherwise, it should be assumed that everything is finite-dimensional and defined over $\CC$.

If $V$ is a vector space then $V^* = \Hom ( V , \CC )$ denotes its {\emph{dual space}}. If $v \in V$ and $\alpha \in V^*$ then $v \otimes \alpha$ defines an endomorphism of $V$ such that
\begin{equation}\label{eq:VVStar}
( v \otimes \alpha ) w = \alpha ( w ) v~, 
\end{equation}
for all $w \in V$. If $X \in \End V$ then $X^* \in \End V^*$ is defined such that
\begin{equation}\label{eq:EndStar}
( X^* \alpha ) (v) = \alpha ( Xv)~, 
\end{equation}
for all $\alpha \in V^*$ and $v \in V$. 

If $V$ is equipped with a nondegenerate bilinear form $b : V \times V \rightarrow \CC$ then the {\emph{canonical isomorphisms}} $\flat : V \rightarrow V^*$ and $\sharp : V^* \rightarrow V$ are defined such that
\begin{equation}\label{eq:FlatSharp}
v^\flat (w)= b (v,w)~, \quad \alpha (v) = b ( \alpha^\sharp ,v)~, 
\end{equation}
for all $v,w \in V$ and $\alpha \in V^*$. If $X \in \End V$ then \eqref{eq:EndStar} and \eqref{eq:FlatSharp} imply
\begin{equation}\label{eq:EndStarb}
( X^* v^\flat ) (w) = b (v,Xw)~, 
\end{equation}
for all $v,w \in V$. 

The commutator of endomorphisms of $V$ defines a Lie bracket on $\End V$ and the resulting Lie algebra is denoted by $\fgl (V)$. The Lie subalgebra of traceless endomorphisms of $V$ is denoted by $\fsl (V)$. The Lie algebra of endomorphisms which leave $b$ invariant, i.e. 
\begin{equation}\label{eq:SoSpV}
\{ X \in \fgl (V)~|~b(Xv,w) + b(v,Xw) = 0,~v,w \in V \}~,
\end{equation}
is denoted by $\fso(V)$ if $b$ is symmetric and $\fsp (V)$ if $b$ is skewsymmetric.

As vector spaces, $\fso (V) \cong \Lambda^2 V$ and $\fsp (V) \cong S^2 V$. If $b$ is symmetric then any $v \otimes w - w \otimes v \in \Lambda^2 V$ corresponds to $v \otimes w^\flat - w \otimes v^\flat \in \fso (V)$. If $b$ is skewsymmetric then any $v \otimes w + w \otimes v \in S^2 V$ corresponds to $v \otimes w^\flat + w \otimes v^\flat \in \fsp (V)$.

If $V$ is a representation of a Lie algebra $\fg$ then $V^*$ should be construed as its {\emph{dual representation}}, defined such that 
\begin{equation}\label{eq:DualRep}
( X \cdot \alpha ) (v) = - \alpha ( X \cdot v)~, 
\end{equation}
for all $\alpha \in V^*$ and $v \in V$, where $\cdot$ denotes the action of $\fg$ in a given representation. If $\fg$ is either $\fgl (V)$ or $\fsl (V)$ then \eqref{eq:DualRep} and \eqref{eq:EndStar} imply
\begin{equation}\label{eq:DualRepGL}
X \cdot \alpha = - X^* \alpha~, 
\end{equation}
for all $X \in \fg$ and $\alpha \in V^*$. If $\fg$ is either $\fso (V)$ or $\fsp (V)$ then \eqref{eq:DualRepGL} and \eqref{eq:EndStarb} imply
\begin{equation}\label{eq:DualRepb}
X \cdot v^\flat = ( Xv )^\flat~, 
\end{equation}
for all $X \in \fg$ and $v \in V$. It follows that $V \cong V^*$ as representations of $\fg$ (not just as vector spaces). This may be thought of as a consequence of the following deeper result.

\begin{prop} \label{prop:gVbExistence}
If $V$ is an irreducible representation of a semisimple Lie algebra $\fg$ then $V$ admits a (non-zero) $\fg$-invariant bilinear form $b : V \times V \rightarrow \CC$, i.e.
\begin{equation}\label{eq:gVbExistence}
b( X \cdot v , w ) + b( v , X \cdot w ) = 0~,
\end{equation}
for all $X \in \fg$ and $v,w \in V$, if and only if $V \cong V^*$ as representations of $\fg$. Furthermore, if $b$ exists, it is (up to scaling) unique, nondegenerate and either symmetric or skewsymmetric. 
\end{prop}

\begin{proof}
See, for example, ch.~VIII \S7 no.~5 prop.~12 in \cite{Bourbaki}.
\end{proof}

%%%%%%%%%%%%%%%%%%%%%%

\section{Clifford algebras and spinor representations}
\label{sec:TheCliffordAlgebraAndSpinorRepresentations}

Now let $\cV$ be an $n$-dimensional vector space equipped with a nondegenerate symmetric bilinear form $( -,- ) : \cV \times \cV \rightarrow \CC$. 

The {\emph{tensor algebra}} of $\cV$ is the infinite-dimensional vector space
\begin{equation}\label{eq:TensorAlgebra}
T \cV = \bigoplus_{k=0}^\infty \cV^{\otimes k}
\end{equation}
that is a unital associative algebra (with unit $1$) with respect to the tensor product. $T \cV$ contains a two-sided ideal $I$ whose elements are linear combinations of terms containing a factor of the form $x \otimes x - (x,x) 1$, for any $x \in \cV$.

The {\emph{Clifford algebra}} of $\cV$ is the unital associative quotient algebra
\begin{equation}\label{eq:CliffordAlgebra}
\Cl ( \cV ) = T \cV /I
\end{equation}
whose elements are subject to the relation 
\begin{equation}\label{eq:CliffordRelation}
x^2 = (x,x) 1~,
\end{equation}
for all $x \in \cV$. 

For any $A \in \fso( \cV )$, let $s_A \in \Lambda^2 \cV$ be defined such that 
\begin{equation}\label{eq:omegaA}
s_A ( x^\flat , y^\flat ) = \tfrac{1}{4} ( x , Ay )~,
\end{equation}
for all $x,y \in \cV$. If $x_1 , x_2 \in \cV$ and $A = x_1 \otimes x_2^\flat -  x_2 \otimes x_1^\flat$ then \eqref{eq:omegaA} implies $s_A = \tfrac{1}{4} ( x_1 \otimes x_2 - x_2 \otimes x_1 )$. Thought of as an element in $\Cl ( \cV )$, we write $s_A = \tfrac{1}{4} ( x_1 x_2 - x_2 x_1 )$.

\begin{prop} \label{prop:OmegaABCommutator}
With respect to multiplication in $\Cl ( \cV )$, the commutator 
\begin{equation}\label{eq:OmegaABCommutator}
[ s_A , s_B ] = s_{[A,B]}~, 
\end{equation}
for all $A,B \in \fso ( \cV )$, where $[A,B] = AB-BA$ is the commutator of endomorphisms.
\end{prop}

\begin{proof}
If $x_1 , x_2 \in \cV$ and $A = x_1 \otimes x_2^\flat -  x_2 \otimes x_1^\flat$ then
\begin{equation}\label{eq:OmegaAzCommutator1}
[ s_A , y ] \underset{\eqref{eq:omegaA}}{=} \tfrac{1}{4} ( x_1 x_2 y - x_2 x_1 y - y x_1 x_2 + y x_2 x_1 )\underset{\eqref{eq:CliffordRelation}}{=}  ( x_2 , y ) x_1 - ( x_1 , y ) x_2~, 
\end{equation}
for all $y \in \cV$. If $y_1 , y_2 \in \cV$ and $B = y_1 \otimes y_2^\flat -  y_2 \otimes y_1^\flat$ then
\begin{align}
[ s_A , s_B ] &\underset{\eqref{eq:omegaA}}{=} \tfrac{1}{4} ( [ s_A , y_1 ] y_2 + y_1 [ s_A , y_2 ] - [ s_A , y_2 ] y_1 - y_2 [ s_A , y_1 ]  ) \nonumber \\
&\underset{\eqref{eq:OmegaAzCommutator1}}{=} \tfrac{1}{4} ( ( x_2 , y_1 ) ( x_1 y_2 - y_2 x_1 ) - ( x_1 , y_1 ) ( x_2 y_2 - y_2 x_2 ) + ( x_2 , y_2 ) ( y_1 x_1 - x_1 y_1 ) - ( x_1 , y_2 ) ( y_1 x_2 - x_2 y_1 ) ) \nonumber \\
&\underset{\eqref{eq:omegaA}}{=} s_{[A,B]}  \label{eq:OmegaABCommutator1}
\end{align}
since 
\begin{align}
[A,B] &= ( x_2 , y_1 ) ( x_1 \otimes y_2^\flat - y_2 \otimes x_1^\flat  ) - ( x_2 , y_2 ) ( x_1 \otimes y_1^\flat - y_1 \otimes x_1^\flat ) \nonumber \\
&\hspace*{.4cm} - ( x_1 , y_1 ) ( x_2 \otimes y_2^\flat - y_2 \otimes x_2^\flat ) + ( x_1 , y_2 ) ( x_2 \otimes y_1^\flat - y_1 \otimes x_2^\flat )~. \label{eq:ABCommutator}
\end{align}
The result follows since any $A, B \in \fso ( \cV )$ can be expressed as linear combinations of terms of the form used above. 
\end{proof}

\begin{cor} \label{cor:SpinorRep}
Any representation of $\Cl ( \cV )$ restricts to a representation of $\fso ( \cV )$ where the action of $A \in \fso ( \cV )$ is defined by the action of $s_A \in \Cl ( \cV )$.
\end{cor}

If $n=2p$ then $\Cl ( \cV ) \cong \Mat_{2^p} ( \CC )$. Up to isomorphism, $\Mat_{2^p} ( \CC )$ has a unique irreducible representation defined by the standard left action of matrices on $\CC^{2^p}$. Let $C^{(2p)}$ denote the corresponding irreducible representation of $\Cl ( \cV )$ and let $\cS^{(2p)}$ denote its restriction to $\fso ( \cV )$. With respect to an orthonormal basis $e_1 ,..., e_{2p}$ for $\cV$, the element $\varpi = i^p e_1 ... e_{2p}$ squares to $1$ and commutes with $s_A$, for all $A \in \fso ( \cV )$. Let $\cS_\pm^{(2p)} \subset \cS^{(2p)}$ denote the inequivalent irreducible (spinor) representations of $\fso ( \cV )$ on which $\varpi = \pm 1$.

If $n=2p+1$ then $\Cl ( \cV ) \cong \Mat_{2^p} ( \CC) \oplus \Mat_{2^p} ( \CC)$. Up to isomorphism, $\Mat_{2^p} ( \CC) \oplus \Mat_{2^p} ( \CC)$ has  two inequivalent irreducible representations defined by the standard left action of matrices in each factor of $\Mat_{2^p} ( \CC)$ on $\CC^{2^p}$. Let $C_\pm^{(2p+1)}$ denote the two corresponding inequivalent irreducible representations of $\Cl ( \cV )$. With respect to an orthonormal basis $e_1 ,..., e_{2p+1}$ for $\cV$, the element $\varpi = i^p e_1 ... e_{2p+1}$ squares to $1$ and is central in $\Cl ( \cV )$. Let $C_\pm^{(2p+1)}$ correspond to the inequivalent irreducible representations of $\Cl ( \cV )$ on which $\varpi = \pm 1$. Up to isomorphism, $C_\pm^{(2p+1)}$ both restrict to the same irreducible (spinor) representation of $\fso ( \cV )$ that we shall denote by $\cS^{(2p+1)}$.  

%%%%%%%%%%%%%%%%%%%%%

\subsection{Spinorial bilinear forms}
\label{sec:SpinorialBilinearForms}

\begin{thm} \label{thm:AdmissibleSpinorialBilinearForm}
There exists a nondegenerate bilinear form $\langle -,- \rangle : \cS^{(n)} \times \cS^{(n)} \rightarrow \CC$ obeying
\begin{align}
\langle \psi , \chi \rangle &= \sigma \langle \chi , \psi \rangle~, \label{eq:AdmissibleSpinorialBilinearForm1} \\
\langle x \cdot \psi , \chi \rangle &= \tau \langle \psi , x \cdot \chi \rangle~, \label{eq:AdmissibleSpinorialBilinearForm2}
\end{align}
for all $\psi , \chi \in \cS^{(n)}$ and $x \in \cV$, for some choice of signs $\sigma$ and $\tau$.
\end{thm}

\begin{proof}
See \cite{AlekCort1995math,Alekseevsky:2003vw}. 
\end{proof}

The possible values of these signs are displayed in Table~\ref{tab:SigmaTauSigns}.
\begin{table}
  \centering
  \caption{Signs for spinorial bilinear forms}
  \label{tab:SigmaTauSigns}
  \begin{tabular}{|c||c|c|c|c|c|c|c|c|}
   \hline
   &&&&&&&& \\ [-.35cm]  
    $n$ mod $8$ & 1 & 2 & 3 & 4 & 5 & 6 & 7 & 8 \\ [.05cm]
    \hline
    &&&&&&&& \\ [-.35cm] 
    $\sigma$ & $+$ & $\mp$ & $-$ & $-$ & $-$ & $\pm$ & $+$ & $+$ \\ [.05cm]
    \hline
    &&&&&&&& \\ [-.35cm] 
    $\tau$ & $+$ & $\mp$ & $-$ & $\pm$ & $+$ & $\mp$ & $-$ & $\pm$ \\ [.05cm]
    \hline 
  \end{tabular}
\end{table}
If $n$ is even and $\langle -,- \rangle$ has one of the two possible values of $( \sigma ,\tau )$ then the nondegenerate bilinear form $\langle -, \varpi \cdot - \rangle$ has the other value.

\begin{prop} \label{prop:SpinInvariance}
The bilinear form $\langle -,- \rangle$ is $\fso ( \cV )$-invariant, i.e. 
\begin{equation}\label{eq:SpinInvariance}
\langle s_A \cdot \psi , \chi \rangle + \langle \psi , s_A \cdot \chi \rangle = 0~,
\end{equation}
for all $A \in \fso ( \cV )$ and $\psi , \chi \in \cS^{(n)}$. 
\end{prop}

\begin{proof}
If $x_1 , x_2 \in \cV$ and $A = x_1 \otimes x_2^\flat -  x_2 \otimes x_1^\flat$ then
\begin{equation}\label{eq:SpinInvariance1}
\langle s_A \cdot \psi , \chi \rangle \underset{\eqref{eq:omegaA}}{=} \tfrac{1}{4} \langle ( x_1 x_2 - x_2 x_1 ) \cdot \psi , \chi \rangle \underset{\eqref{eq:AdmissibleSpinorialBilinearForm2}}{=} \tau^2 \tfrac{1}{4} \langle \psi , ( x_2 x_1 - x_1 x_2 ) \cdot \chi \rangle \underset{\eqref{eq:omegaA}}{=} - \langle \psi , s_A \cdot \chi \rangle~,
\end{equation}
for all $\psi , \chi \in \cS^{(n)}$. The result follows since any $A \in \fso ( \cV )$ can be expressed as a linear combination of terms of the form used above. 
\end{proof}

\begin{remark}
In fact, any $\fso ( \cV )$-invariant bilinear form on $\cS^{(n)}$ is a linear combination of those above.
\end{remark}

If $n=2p$ then 
\begin{equation}\label{eq:AdmissibleSpinorialBilinearForme2p}
\langle e_1 ... e_{2p} \cdot \psi , \chi \rangle  \underset{\eqref{eq:AdmissibleSpinorialBilinearForm2}}{=} \tau^{2p} \langle \psi , e_{2p} ... e_1 \cdot \chi \rangle  \underset{\eqref{eq:CliffordRelation}}{=} (-1)^p \langle \psi , e_1 ... e_{2p} \cdot \chi \rangle~,
\end{equation}
for all $\psi , \chi \in \cS^{(2p)}$. Since $\varpi = \pm 1$ on $\cS_\pm^{(2p)}$, it follows that
\begin{equation}\label{eq:AdmissibleSpinorialBilinearFormPlusMinus}
\langle \psi_\pm , \chi_\pm \rangle = (-1)^p \langle \psi_\pm , \chi_\pm \rangle \; , \quad \langle \psi_\pm , \chi_\mp \rangle = (-1)^{p+1} \langle \psi_\pm , \chi_\mp \rangle~,
\end{equation}
for all $\psi_\pm , \chi_\pm \in \cS_\pm^{(2p)}$. Therefore $\langle \psi_\pm , \chi_\pm \rangle = 0$ if $p$ is odd while $\langle \psi_\pm , \chi_\mp \rangle = 0$ if $p$ is even. 

The bilinear map  
\begin{equation}\label{eq:XiMap}
\xi  :\cS^{(n)} \times \cS^{(n)} \longrightarrow  \cV
\end{equation}
is defined such that
\begin{equation}\label{eq:XiDef}
(x, \xi ( \psi , \chi ) ) = \langle \psi , x \cdot \chi \rangle~,
\end{equation}
for all $x \in \cV$ and $\psi , \chi \in \cS^{(n)}$.

\begin{prop} \label{prop:SpinEquivariance}
The bilinear map $\xi$ is $\fso ( \cV )$-equivariant, i.e. 
\begin{equation}\label{eq:SpinEquivariance}
\xi ( s_A \cdot \psi , \chi ) + \xi ( \psi ,  s_A \cdot \chi ) = A \xi ( \psi , \chi )~,
\end{equation}
for all $A \in \fso ( \cV )$ and $\psi , \chi \in \cS^{(n)}$. 
\end{prop}

\begin{proof}
For any $\psi , \chi \in \cS^{(n)}$ and $x,y \in \cV$, we have 
\begin{equation}\label{eq:XiSym}
\langle \psi , x \cdot \chi \rangle \underset{\eqref{eq:AdmissibleSpinorialBilinearForm1}}{=} \sigma \langle x \cdot \chi , \psi \rangle \underset{\eqref{eq:AdmissibleSpinorialBilinearForm2}}{=} \sigma \tau \langle \chi , x \cdot \psi \rangle
\end{equation}
and
\begin{equation}\label{eq:XixSym}
\langle x \cdot \psi , y \cdot \chi \rangle \underset{\eqref{eq:AdmissibleSpinorialBilinearForm2}}{=} \tau \langle \psi , xy \cdot \chi \rangle \underset{\eqref{eq:CliffordRelation}}{=} - \tau \langle \psi , yx \cdot \chi \rangle + 2 \tau (x,y) \langle \psi , \chi \rangle~.
\end{equation}
Therefore
\begin{align}
\xi ( \psi , \chi ) &= \sigma \tau \xi ( \chi , \psi )~, \label{eq:XiSym2} \\
\xi ( x \cdot \psi , \chi ) &= - \tau \xi ( \psi , x \cdot \chi ) + 2\tau \langle \psi , \chi \rangle x \label{eq:XiSym3}~,
\end{align}
for all $\psi , \chi \in \cS^{(n)}$ and $x \in \cV$. 

If $x_1 , x_2 \in \cV$ and $A = x_1 \otimes x_2^\flat -  x_2 \otimes x_1^\flat$ then
\begin{align}\label{eq:SpinEquivariance1}
\xi ( s_A \cdot \psi , \chi ) + \xi ( \psi ,  s_A \cdot \chi ) &\;\underset{\eqref{eq:omegaA}}{=} \tfrac{1}{4} ( \xi ( ( x_1 x_2 - x_2 x_1 ) \cdot \psi , \chi ) + \xi ( \psi , ( x_1 x_2 - x_2 x_1 ) \cdot \chi ) ) \nonumber \\ 
&\underset{\eqref{eq:XiSym3}}{=}  \tfrac{1}{4}  ( - \tau \xi ( x_2 \cdot \psi , x_1 \cdot \chi ) + 2\tau \langle x_2 \cdot \psi , \chi \rangle x_1 + \xi ( \psi , x_1 x_2 \cdot \chi ) ) - ( 1 \leftrightarrow 2 )\nonumber \\
&\underset{\eqref{eq:XiSym3}}{=}  \tfrac{1}{4}  (\xi ( \psi , x_2 x_1  \cdot \chi ) - 2 \langle \psi , x_1 \cdot \chi \rangle x_2 + 2\tau \langle x_2 \cdot \psi , \chi \rangle x_1 + \xi ( \psi , x_1 x_2 \cdot \chi ) ) - ( 1 \leftrightarrow 2 ) \nonumber \\
&\underset{\eqref{eq:AdmissibleSpinorialBilinearForm2}}{=} - \langle \psi , x_1 \cdot \chi \rangle x_2 +  \langle \psi , x_2 \cdot \chi \rangle x_1  \nonumber \\
&\underset{\eqref{eq:XiDef}}{=} - ( x_1 , \xi ( \psi , \chi ) ) x_2 +  ( x_2 , \xi ( \psi , \chi ) ) x_1 \nonumber \\
&\;\; = ( x_1 \otimes x_2^\flat - x_2 \otimes x_1^\flat ) \xi ( \psi , \chi ) \nonumber \\
&\;\; = A \xi ( \psi , \chi )~, 
\end{align}
for all $\psi , \chi \in \cS^{(n)}$. The result follows since any $A \in \fso ( \cV )$ can be expressed as a linear combination of terms of the form used above. 
\end{proof}

If $n=2p$ then 
\begin{equation}\label{eq:Xie2p}
\langle e_1 ... e_{2p} \cdot \psi , x \cdot \chi \rangle \underset{\eqref{eq:AdmissibleSpinorialBilinearForme2p}}{=} (-1)^p \langle \psi , e_1 ... e_{2p} x \cdot \chi \rangle \underset{\eqref{eq:CliffordRelation}}{=} (-1)^{p+1} \langle \psi , x e_1 ... e_{2p} \cdot \chi \rangle ~,
\end{equation}
for all $\psi , \chi \in \cS^{(2p)}$ and $x \in \cV$. Therefore 
\begin{equation}\label{eq:Xie2p2}
\xi ( e_1 ... e_{2p} \cdot \psi , \chi ) = (-1)^{p+1} \xi ( \psi , e_1 ... e_{2p} \cdot \chi )~,
\end{equation}
for all $\psi , \chi \in \cS^{(2p)}$. Since $\varpi = \pm 1$ on $\cS_\pm^{(2p)}$, it follows that
\begin{equation}\label{eq:XiPlusMinus}
\xi ( \psi_\pm , \chi_\pm ) = (-1)^{p+1} \xi ( \psi_\pm , \chi_\pm ) \; , \quad \xi ( \psi_\pm , \chi_\mp ) = (-1)^p \xi ( \psi_\pm , \chi_\mp )~,
\end{equation}
for all $\psi_\pm , \chi_\pm \in \cS_\pm^{(2p)}$. Therefore $\xi ( \psi_\pm , \chi_\pm ) = 0$ if $p$ is even while $\xi ( \psi_\pm , \chi_\mp ) = 0$ if $p$ is odd.

%%%%%%%%%%%%%%%%%%%%%%

\section{Poincar\'{e} superalgebras}
\label{sec:ExtendedPoincareSuperalgebras}

In this section we identify the data needed to define a Poincar\'{e} superalgebra and prove that the existence of this data ensures that a Poincar\'{e} superalgebra is a Lie superalgebra.

%%%%%%%%%%%%%%%%%%%%%%

\subsection{Lie superalgebras}
\label{sec:LieSuperalgebras}

A {\emph{2-grading}} on a vector space $\fV$ is a way of expressing $\fV = \fV_{\bar 0} \oplus \fV_{\bar 1}$ as the direct sum of two subspaces $\fV_{\bar 0}$ and $\fV_{\bar 1}$. We will refer to $\fV_{\bar 0}$ and $\fV_{\bar 1}$ respectively as the {\emph{even}} and {\emph{odd}} parts of $\fV$. If ${\bar \alpha} \in \{ {\bar 0} , {\bar 1} \}$ then any $x \in \fV_{\bar \alpha}$ is said to be {\emph{homogeneous}} and the {\emph{parity}} of $x$ is $|x| = \alpha$. 

A 2-graded vector space is called a {\emph{vector superspace}}. Let $\fV$, $\fW$ and $\fU$ be vector superspaces. A linear map $f : \fV \rightarrow \fW$ will be called {\emph{even}} if $f ( \fV_{\bar \alpha} ) \subset \fW_{\bar \alpha}$, for all ${\bar \alpha} \in \{ {\bar 0} , {\bar 1} \}$. A bilinear map $f : \fV \times \fV \rightarrow \fW$ will be called {\emph{supersymmetric}} if 
\begin{equation}\label{eq:supersymmetry}
f(x,y)= (-1)^{|x||y|} f(y,x)~,
\end{equation}
for all homogeneous $x,y \in \fV$, and {\emph{superskewsymmetric}} if 
\begin{equation}\label{eq:superskewsymmetry}
f(x,y)= - (-1)^{|x||y|} f(y,x)~,
\end{equation}
for all homogeneous $x,y \in \fV$. If ${\bar \alpha} , {\bar \beta} \in \{ {\bar 0} , {\bar 1} \}$ then we define 
\begin{equation}\label{eq:Z2Addition}
{\bar \alpha} + {\bar \beta} = 
\begin{cases}
{\bar 0} \quad {\mbox{if ${\bar \alpha} = {\bar \beta}$}}~, \\
{\bar 1} \quad {\mbox{if ${\bar \alpha} \neq {\bar \beta}$}}~.
\end{cases}
\end{equation}
A bilinear map $f : \fV \times \fW \rightarrow \fU$ will be called {\emph{even}} if $f ( \fV_{\bar \alpha} , \fW_{\bar \beta} ) \subset \fU_{{\bar \alpha} + {\bar \beta}}$, for all ${\bar \alpha} , {\bar \beta} \in \{ {\bar 0} , {\bar 1} \}$.

A vector superspace $\fG$ is called a {\emph{Lie superalgebra}} if it is equipped with an even superskewsymmetric bilinear map (or {\emph{bracket}})
\begin{equation}\label{eq:LSABracket} 
[-,-] : \fG \times \fG \longrightarrow \fG
\end{equation}
which obeys the {\emph{Jacobi identity}}
\begin{equation}\label{eq:LSAJacobiG}
[ a , [ b , c ] ] = [ [ a , b ] , c ] + (-1)^{|a||b|} [ b , [ a , c ] ]~,
\end{equation}
for all homogeneous $a , b \in \fG$ and $c \in \fG$. 

If $\fG$ is a Lie superalgebra then the $[\bar{0}\bar{0}]$ bracket defines a Lie bracket on $\fG_{\bar 0}$. The $[\bar{0}\bar{1}]$ bracket defines an action of the Lie algebra $\fG_{\bar 0}$ on $\fG_{\bar 1}$ that makes $\fG_{\bar 1}$ into a representation of $\fG_{\bar 0}$. The $[\bar{1}\bar{1}]$ bracket defines a $\fG_{\bar 0}$-equivariant bilinear map $\fG_{\bar 1} \times \fG_{\bar 1} \rightarrow \fG_{\bar 0}$. It is easily verified that these three properties together with the condition
\begin{equation}\label{eq:LSAJacobi111}
[ a , [ b , c ] ] + [ b , [ c , a ] ] + [ c , [ a , b ] ] = 0~,
\end{equation}
for all $a,b,c \in \fG_{\bar 1}$, are in fact necessary and sufficient to guarantee that $\fG$ is a Lie superalgebra.

%%%%%%%%%%%%%%%%%%%%%%

\subsection{The basic setup}
\label{sec:TheBasicSetup}

The {\emph{Poincar\'{e} algebra}} $\fp (n)$ in $n$ dimensions is the Lie algebra defined by the vector space $\fso ( \cV ) \oplus \cV$ equipped with the Lie brackets 
\begin{align}
[A,B] &= AB-BA~, \label{eq:PoincareBracketsAB} \\
[A,x] &= Ax~, \label{eq:PoincareBracketsAx} \\
[x,y] &= 0~, \label{eq:PoincareBracketsxy}
\end{align}
for all $A,B \in \fso( \cV )$ and $x,y \in \cV$. 

The {\emph{Poincar\'{e} superalgebra}} $\fP (n)$ in $n$ dimensions is the vector superspace with even part 
\begin{equation}\label{eq:SP0}
\fP (n)_{\bar 0} = \fp (n) \oplus \fg~,
\end{equation}
for some Lie algebra $\fg$, and odd part 
\begin{equation}\label{eq:SP1}
\fP (n)_{\bar 1} = 
\begin{cases}
\cS^{(n)} \otimes V \quad\quad\quad\quad\quad\quad\quad\quad {\mbox{if $n$ is odd}}~, \\
\cS_+^{(n)} \otimes V_+ \oplus \cS_-^{(n)} \otimes V_- \quad\quad\, {\mbox{if $n$ is even}}~,
\end{cases}
\end{equation}
for some representations $V$, $V_\pm$ of $\fg$. If $n$ is even then we define $V = V_+ \oplus V_-$.

We will follow the conventional assumption (motived by \cite{Coleman:1967ad}) that $[ \fp (n) , \fg ] = 0$, i.e.
\begin{align}
[A,X] &=0~, \label{eq:SPAX} \\
[x,X] &=0~, \label{eq:SPxX}
\end{align}
for all $A \in \fso ( \cV )$, $x \in \cV $ and $X \in \fg$. Therefore $\fP (n)_{\bar 0}$ is a Lie algebra since it is the direct sum of the Lie algebras $\fp (n)$ and $\fg$.  

%%%%%%%%%%%%%%%%%%%%%%

\subsection{$[\bar{0}\bar{1}]$ bracket and $[\bar{0}\bar{0}\bar{1}]$ Jacobi identity}
\label{sec:01BracketAnd001JacobiIdentity}

The action of $\fP (n)_{\bar 0}$ on $\fP (n)_{\bar 1}$ is defined by the brackets
\begin{align} 
[ A , \psi \otimes v ] &= s_A \cdot \psi \otimes v~, \label{eq:SP01A} \\
[ x , \psi \otimes v ] &= 0~, \label{eq:SP01x} \\
\quad [ X , \psi \otimes v ] &= \psi \otimes X \cdot v~, \label{eq:SP01X}
\end{align}
for all $A \in \fso ( \cV )$, $x \in \cV$, $X \in \fg$ and $\psi \otimes v \in \fP (n)_{\bar 1}$. 

\begin{prop} \label{prop:001Jacobi}
The brackets above define $\fP (n)_{\bar 1}$ as a representation of $\fP (n)_{\bar 0}$.
\end{prop}

\begin{proof}
For any $A,B \in \fso ( \cV )$, $x,y \in \cV$, $X,Y \in \fg$ and $\psi \otimes v \in \fP (n)_{\bar 1}$,
\begin{align}\label{eq:SP001} 
[ [A,B] , \psi \otimes v ] &\underset{\eqref{eq:SP01A}}{=} s_{[A,B]} \cdot \psi \otimes v \underset{\eqref{eq:OmegaABCommutator}}{=} s_A \cdot ( s_B \cdot \psi ) \otimes v - s_B \cdot ( s_A \cdot \psi ) \otimes v \underset{\eqref{eq:SP01A}}{=}  [ A , [ B , \psi \otimes v ] ] - [ B , [ A , \psi \otimes v ] ]~, \nonumber \\ 
[ [A,x] , \psi \otimes v ] &\underset{\eqref{eq:PoincareBracketsAx}}{=} [ Ax , \psi \otimes v ] \underset{\eqref{eq:SP01x}}{=} 0 \overset{\eqref{eq:SP01A}}{\underset{\eqref{eq:SP01x}}{=}} [ A , [ x , \psi \otimes v ] ] - [ x , [ A , \psi \otimes v ] ]~, \nonumber \\
[ [x,y] , \psi \otimes v ] &\underset{\eqref{eq:PoincareBracketsxy}}{=} 0 \underset{\eqref{eq:SP01x}}{=}  [ x , [ y , \psi \otimes v ] ] - [ y , [ x , \psi \otimes v ] ]~, \nonumber \\
[ [A,X] , \psi \otimes v ] &\underset{\eqref{eq:SPAX}}{=} 0 = s_A \cdot \psi \otimes X \cdot v - s_A \cdot \psi \otimes X \cdot v \overset{\eqref{eq:SP01A}}{\underset{\eqref{eq:SP01X}}{=}} [ A , [ X , \psi \otimes v ] ] - [ X , [ A , \psi \otimes v ] ]~, \nonumber \\ 
[ [x,X] , \psi \otimes v ] &\underset{\eqref{eq:SPxX}}{=} 0 \overset{\eqref{eq:SP01x}}{\underset{\eqref{eq:SP01X}}{=}} [ x , [ X , \psi \otimes v ] ] - [ X , [ x , \psi \otimes v ] ]~, \nonumber \\
[ [X,Y] , \psi \otimes v ] &\underset{\eqref{eq:SP01X}}{=} \psi \otimes [X,Y] \cdot v = \psi \otimes X \cdot ( Y \cdot v ) - \psi \otimes Y \cdot ( X \cdot v ) \underset{\eqref{eq:SP01X}}{=} [ X , [ Y , \psi \otimes v ] ] - [ Y , [ X , \psi \otimes v ] ]~. 
\end{align}
\end{proof}

%%%%%%%%%%%%%%%%%%%%%%

\subsection{$[\bar{1}\bar{1}]$ bracket and $[\bar{0}\bar{1}\bar{1}]$ Jacobi identity}
\label{sec:11BracketAnd011JacobiIdentity}

Now let 
\begin{equation}\label{eq:b}
b : V \times V \longrightarrow  \CC
\end{equation}
be a nondegenerate bilinear form obeying 
\begin{equation}\label{eq:gInvariantb}
b ( X \cdot v , w ) + b ( v , X \cdot w ) = 0~,
\end{equation}
for all $X \in \fg$ and $v,w \in V$, and let
\begin{equation}\label{eq:c}
c : V \times V \longrightarrow  \fg
\end{equation}
be a bilinear map obeying 
\begin{equation}\label{eq:gEquivariantc}
[ X , c(v,w) ] = c ( X \cdot v , w ) + c ( v , X \cdot w )~,
\end{equation}
for all $X \in \fg$ and $v,w \in V$. In other words, $b$ is $\fg$-invariant and $c$ is $\fg$-equivariant.

The remaining bracket for $\fP (n)$ is defined by 
\begin{equation}\label{eq:SP11}
[ \psi \otimes v , \chi \otimes w ] = \xi ( \psi , \chi ) b(v,w) + \langle \psi , \chi \rangle c(v,w)~,
\end{equation}
for all $\psi \otimes v , \chi \otimes w \in \fP (n)_{\bar 1}$.

To ensure that \eqref{eq:SP11} is symmetric, \eqref{eq:XiSym2} and \eqref{eq:AdmissibleSpinorialBilinearForm1} imply that we must have
\begin{align} 
b(v,w) &= \sigma\tau b(w,v)~, \label{eq:bSym} \\
c(v,w) &= \sigma c(w,v)~, \label{eq:cSym}
\end{align}
for all $v,w \in V$.

If $n=2p$ then \eqref{eq:XiPlusMinus} and \eqref{eq:AdmissibleSpinorialBilinearFormPlusMinus} imply that we can assume
\begin{equation}\label{eq:bcPlusMinuspOdd}
b( V_\pm , V_\mp ) = 0~, \quad c( V_\pm , V_\pm ) = 0~, 
\end{equation}
if $p$ is odd, and
\begin{equation}\label{eq:bcPlusMinuspEven}
b( V_\pm , V_\pm ) = 0~, \quad c( V_\pm , V_\mp ) = 0~, 
\end{equation}
if $p$ is even.

\begin{prop} \label{prop:011Jacobi}
The bilinear map $\fP (n)_{\bar 1} \times \fP (n)_{\bar 1} \rightarrow \fP (n)_{\bar 0}$ defined by \eqref{eq:SP11} is $\fP (n)_{\bar 0}$-equivariant.
\end{prop}

\begin{proof}
For any $A \in \fso ( \cV )$, $x \in \cV$, $X \in \fg$ and $\psi \otimes v , \chi \otimes w \in \fP (n)_{\bar 1}$,
\begin{align}\label{eq:SP011} 
[ A , [ \psi \otimes v , \chi \otimes w ] ] &\underset{\eqref{eq:SP11}}{=} [ A , \xi( \psi , \chi ) ] b(v,w) +  \langle \psi , \chi \rangle [ A , c(v,w) ] \nonumber \\
&\overset{\eqref{eq:PoincareBracketsAx}}{\underset{\eqref{eq:SPAX}}{=}} A \xi( \psi , \chi ) b(v,w) \nonumber \\ 
&\underset{\eqref{eq:SpinEquivariance}}{=} ( \xi( s_A \cdot \psi , \chi ) + \xi( \psi , s_A \cdot \chi ) ) b(v,w) \nonumber \\
&\underset{\eqref{eq:SpinInvariance}}{=} \xi( s_A \cdot \psi , \chi ) b(v,w) + \langle s_A \cdot \psi , \chi \rangle c(v,w) + \xi( \psi , s_A \cdot \chi ) b(v,w) + \langle \psi , s_A \cdot \chi \rangle c(v,w) \nonumber \\
&\underset{\eqref{eq:SP11}}{=} [ s_A \cdot \psi \otimes v , \chi \otimes w ] + [ \psi \otimes v , s_A \cdot  \chi \otimes w ] \nonumber \\ 
&\underset{\eqref{eq:SP01A}}{=} [ [ A , \psi \otimes v ] , \chi \otimes w ] + [ \psi \otimes v , [ A , \chi \otimes w ] ]~, \nonumber \\
[ x, [ \psi \otimes v , \chi \otimes w ] ] &\underset{\eqref{eq:SP11}}{=}  [ x , \xi( \psi , \chi ) ] b(v,w) +  \langle \psi , \chi \rangle [ x , c(v,w) ] \nonumber \\
&\overset{\eqref{eq:PoincareBracketsxy}}{\underset{\eqref{eq:SPxX}}{=}} 0 \nonumber \\ 
&\underset{\eqref{eq:SP01x}}{=} [ [ x , \psi \otimes v ] , \chi \otimes w ] + [ \psi \otimes v , [ x , \chi \otimes w ] ]~, \nonumber \\
[ X , [ \psi \otimes v , \chi \otimes w ] ] &\underset{\eqref{eq:SP11}}{=} [ X , \xi( \psi , \chi ) ] b(v,w) +  \langle \psi , \chi \rangle [ X , c(v,w) ] \nonumber \\
&\overset{\eqref{eq:SPxX}}{\underset{\eqref{eq:gEquivariantc}}{=}} \langle \psi , \chi \rangle ( c ( X \cdot v , w ) + c ( v , X \cdot w ) ) \nonumber \\ 
&\underset{\eqref{eq:gInvariantb}}{=} \xi ( \psi , \chi ) b ( X \cdot v , w ) + \langle \psi , \chi \rangle c ( X \cdot v , w ) + \xi ( \psi , \chi )  b ( v , X \cdot w ) + \langle \psi , \chi \rangle c ( v , X \cdot w ) \nonumber \\
&\underset{\eqref{eq:SP11}}{=} [ \psi \otimes X \cdot v , \chi \otimes w ] + [ \psi \otimes v , \chi \otimes X \cdot w ] \nonumber \\ 
&\underset{\eqref{eq:SP01X}}{=} [ [ X , \psi \otimes v ] , \chi \otimes w ] + [ \psi \otimes v , [ X , \chi \otimes w ] ]~.
\end{align}
\end{proof}

%%%%%%%%%%%%%%%%%%%%%%

\subsection{$[\bar{1}\bar{1}\bar{1}]$ Jacobi identity}
\label{sec:111JacobiIdentity}

The final condition \eqref{eq:LSAJacobi111} needed for $\fP (n)$ to be a Lie superalgebra is
\begin{align}\label{eq:SP111} 
&[ [ \psi \otimes v , \chi \otimes w ] , \phi \otimes u ] + [ [ \chi \otimes w , \phi \otimes u ] , \psi \otimes v ] + [ [ \phi \otimes u , \psi \otimes v ] , \chi \otimes w ] \nonumber \\ 
&\overset{\eqref{eq:SP11}}{\underset{\eqref{eq:SP01x}}{=}} \langle \psi , \chi \rangle [ c(v,w) , \phi \otimes u ] + \langle \chi , \phi \rangle [ c(w,u) , \psi \otimes v ]  + \langle \phi , \psi \rangle [ c(u,v) , \chi \otimes w ]  \nonumber \\
&\underset{\eqref{eq:SP01X}}{=} \langle \psi , \chi \rangle \phi \otimes c(v,w) \cdot u + \langle \chi , \phi \rangle \psi \otimes c(w,u) \cdot v + \langle \phi , \psi \rangle \chi \otimes c(u,v) \cdot w \nonumber \\
&\hspace*{.2cm} =0~,
\end{align}
for all $\psi \otimes v , \chi \otimes w , \phi \otimes u \in \fP (n)_{\bar 1}$. 

If $c$ is such that
\begin{equation}\label{eq:cVVV0}
c( v , w ) \cdot u = 0~,
\end{equation}
for all $v,w,u \in V$, then clearly \eqref{eq:SP111} is solved. We shall refer to this as the {\emph{conventional solution}}. 

\begin{thm} \label{thm:111Jacobi}
If $n\geq 4$ then every solution of \eqref{eq:SP111} is conventional. If $n < 4$ then \eqref{eq:SP111} is solved if and only if
\begin{itemize}
\item $c( v , w ) \cdot u + c( u , w ) \cdot v = 0$, for all $v,w,u \in V$, if $n=3$.

\item $c( v_\pm , w_\mp ) \cdot u_\pm + c( u_\pm , w_\mp ) \cdot v_\pm = 0$, for all $v_\pm , w_\pm , u_\pm \in V_\pm$, if $n=2$.

\item $c( v , w ) \cdot u + c( w , u ) \cdot v + c( u , v ) \cdot w = 0$, for all $v,w,u \in V$, if $n=1$.
\end{itemize}
\end{thm}

\begin{proof}
First let $n=2p+1$. Abstracting $\phi$ in \eqref{eq:SP111} gives the equivalent condition
\begin{equation}\label{eq:SP111n2p+1End}
\langle \psi , \chi \rangle \mathbb{1} \otimes c(v,w) \cdot u + \langle \chi , - \rangle \psi \otimes c(w,u) \cdot v + \langle - , \psi \rangle \chi \otimes c(u,v) \cdot w = 0~, 
\end{equation}
that is valued in $\End \cS^{(n)} \otimes V$, where $\mathbb{1}$ denotes the identity map in $\End \cS^{(n)}$. Taking the trace of the $\End \cS^{(n)}$ part of \eqref{eq:SP111n2p+1End} gives
\begin{align}\label{eq:SP111n2p+1EndTrace}
&\langle \psi , \chi \rangle \, \dim \, \cS^{(n)} c( v , w ) \cdot u + \langle \chi ,  \psi \rangle ( c( w , u ) \cdot v + c( u , v ) \cdot w ) \nonumber \\
&= \langle \psi , \chi \rangle ( 2^p c( v , w ) \cdot u + \sigma ( c( w , u ) \cdot v + c( u , v ) \cdot w )  ) \nonumber \\
&=0~,
\end{align}
using \eqref{eq:AdmissibleSpinorialBilinearForm1} and $\dim \, \cS^{(n)} = 2^p$. \eqref{eq:SP111n2p+1EndTrace} is equivalent to
\begin{equation}\label{eq:SP111n2p+1EndTraceg}
2^p c( v , w ) \cdot u + \sigma ( c( w , u ) \cdot v + c( u , v ) \cdot w ) = 0~, 
\end{equation}
for all $v , w , u \in V$.  If we momentarily denote   
\begin{equation}
\Phi ( v , w , u ) = 2^p c( v , w ) \cdot u + \sigma ( c( w , u ) \cdot v + c( u , v ) \cdot w ) ~, 
\end{equation}
then $\Phi ( v , w , u ) + \Phi ( w , u , v ) + \Phi ( u , v , w ) = 0$ (which follows from \eqref{eq:SP111n2p+1EndTraceg}) implies
\begin{equation}\label{eq:SP111n2p+1EndTraceg2}
c( v , w ) \cdot u + c( w , u ) \cdot v + c( u , v ) \cdot w = 0~, 
\end{equation}
for all $v , w , u \in V$, unless $2^p + 2\sigma =0$. But substituting \eqref{eq:SP111n2p+1EndTraceg2} back into \eqref{eq:SP111n2p+1EndTraceg} implies 
\begin{equation}\label{eq:SP111n2p+1EndTracegTrivial}
c( v , w ) \cdot u = 0~, 
\end{equation}
for all $v , w , u \in V$, unless $2^p - \sigma = 0$. \eqref{eq:SP111n2p+1EndTracegTrivial} is the conventional solution of \eqref{eq:SP111}. 

If $2^p - \sigma =0$ then we must have $p=0$ and $\sigma = 1$ (the only option when $n=1$, see Table~\ref{tab:SigmaTauSigns}). In this case, \eqref{eq:SP111n2p+1EndTraceg2} solves \eqref{eq:SP111} since $\dim \, \cS^{(1)} = 1$.

If $2^p + 2\sigma =0$ then we must have $p=1$ and $\sigma = -1$ (the only option when $n=3$, see Table~\ref{tab:SigmaTauSigns}). In this case, \eqref{eq:SP111n2p+1EndTraceg} is
\begin{equation}\label{eq:SP111n3}
2 c( v , w ) \cdot u - c( w , u ) \cdot v - c( u , v ) \cdot w = 0~, 
\end{equation}
for all $v , w , u \in V$. Substituting \eqref{eq:SP111n3} back into \eqref{eq:SP111} gives
\begin{equation}\label{eq:SP111n32}
( \half \langle \psi , \chi \rangle \phi + \langle \chi , \phi \rangle \psi ) \otimes c(w,u) \cdot v + ( \half \langle \psi , \chi \rangle \phi + \langle \phi , \psi \rangle \chi ) \otimes c(u,v) \cdot w = 0~, 
\end{equation}
for all $\psi \otimes v , \chi \otimes w , \phi \otimes u \in \cS^{(3)} \otimes V$. If $\psi = \chi$ and $\langle \psi , \phi \rangle \neq 0$ then \eqref{eq:SP111n32} implies
\begin{equation}\label{eq:SP111n33}
c( w , u ) \cdot v - c( u , v ) \cdot w = 0~, 
\end{equation}
for all $w , u , v \in V$, using $\sigma = -1$. Substituting \eqref{eq:SP111n33} back into \eqref{eq:SP111n32} gives
\begin{equation}\label{eq:SP111n34}
( \langle \psi , \chi \rangle \phi + \langle \chi , \phi \rangle \psi + \langle \phi , \psi \rangle \chi ) \otimes c( u , v ) \cdot w = 0~,
\end{equation}
for all $\psi \otimes v , \chi \otimes w , \phi \otimes u \in \cS^{(3)} \otimes V$. But the sum of terms in $\cS^{(3)}$ in \eqref{eq:SP111n34} is identically zero since $\dim \, \cS^{(3)} = 2$ and $\langle -,- \rangle$ is skewsymmetric. Therefore \eqref{eq:SP111n33} (which implies \eqref{eq:SP111n3}) solves \eqref{eq:SP111} in this case.

If $n=2p$ and $p$ is odd then \eqref{eq:SP111} is equivalent to 
\begin{equation}\label{eq:SP111n2podd}
\langle \psi_\pm , \chi_\mp \rangle \phi_\pm \otimes c( v_\pm , w_\mp ) \cdot u_\pm + \langle \chi_\mp , \phi_\pm \rangle \psi_\pm \otimes c( w_\mp , u_\pm ) \cdot v_\pm = 0~, 
\end{equation}
for all $\psi_\pm \otimes v_\pm , \chi_\pm \otimes w_\pm , \phi_\pm \otimes u_\pm \in \cS_\pm^{(n)} \otimes V_\pm$, since $\langle \cS_\pm^{(n)} , \cS_\pm^{(n)} \rangle = 0$. Abstracting $\phi_\pm$ in \eqref{eq:SP111n2podd} gives the equivalent condition
\begin{equation}\label{eq:SP111n2poddEnd}
\langle \psi_\pm , \chi_\mp \rangle \mathbb{1}_\pm \otimes c( v_\pm , w_\mp ) \cdot u_\pm + \langle \chi_\mp , - \rangle \psi_\pm \otimes c( w_\mp , u_\pm ) \cdot v_\pm = 0~, 
\end{equation}
that is valued in $\End \cS_\pm^{(n)} \otimes V_\pm$, where $\mathbb{1}_\pm$ denotes the identity map in $\End \cS_\pm^{(n)}$. Taking the trace of the $\End \cS_\pm^{(n)}$ part of \eqref{eq:SP111n2poddEnd} gives
\begin{align}\label{eq:SP111n2poddEndTrace}
&\langle \psi_\pm , \chi_\mp \rangle \, \dim \, \cS_\pm^{(n)} c( v_\pm , w_\mp ) \cdot u_\pm + \langle \chi_\mp ,  \psi_\pm \rangle c( w_\mp , u_\pm ) \cdot v_\pm \nonumber \\
&= \langle \psi_\pm , \chi_\mp \rangle ( 2^{p-1} c( v_\pm , w_\mp ) \cdot u_\pm +  c( u_\pm , w_\mp ) \cdot v_\pm ) \nonumber \\
&=0~,
\end{align}
using \eqref{eq:AdmissibleSpinorialBilinearForm1}, \eqref{eq:cSym} and $\dim \, \cS_\pm^{(n)} = 2^{p-1}$. \eqref{eq:SP111n2poddEndTrace} is equivalent to
\begin{equation}\label{eq:SP111n2poddEndTraceg}
2^{p-1} c( v_\pm , w_\mp ) \cdot u_\pm +  c( u_\pm , w_\mp ) \cdot v_\pm = 0~, 
\end{equation}
for all $v_\pm , w_\pm , u_\pm \in V_\pm$. If we momentarily denote   
\begin{equation}
\Phi ( v_\pm , w_\mp , u_\pm ) = 2^{p-1} c( v_\pm , w_\mp ) \cdot u_\pm +  c( u_\pm , w_\mp ) \cdot v_\pm
\end{equation}
then $\Phi ( v_\pm , w_\mp , u_\pm ) + \Phi ( u_\pm , w_\mp , v_\pm ) = 0$ (which follows from \eqref{eq:SP111n2poddEndTraceg}) implies
\begin{equation}\label{eq:SP111n2poddEndTraceg2}
c( v_\pm , w_\mp ) \cdot u_\pm +  c( u_\pm , w_\mp ) \cdot v_\pm = 0~, 
\end{equation}
for all $v_\pm , w_\pm , u_\pm \in V_\pm$. But substituting \eqref{eq:SP111n2poddEndTraceg2} back into \eqref{eq:SP111n2poddEndTraceg} implies 
\begin{equation}\label{eq:SP111n2poddEndTracegTrivial}
c( v_\pm , w_\mp ) \cdot u_\pm = 0~, 
\end{equation}
for all $v_\pm , w_\pm , u_\pm \in V_\pm$, if $p>1$. \eqref{eq:SP111n2poddEndTracegTrivial} is the conventional solution of \eqref{eq:SP111n2podd}. If $p=1$ then \eqref{eq:SP111n2poddEndTraceg2} also solves \eqref{eq:SP111n2podd} since $\dim \, \cS_\pm^{(2)} = 1$.

Finally, if $n=2p$ and $p$ is even then \eqref{eq:SP111} is equivalent to the conditions
\begin{equation}\label{eq:SP111n2peven1}
\langle \psi_\pm , \chi_\pm \rangle \phi_\pm \otimes c( v_\pm , w_\pm ) \cdot u_\pm + \langle \chi_\pm , \phi_\pm \rangle \psi_\pm \otimes c( w_\pm , u_\pm ) \cdot v_\pm + \langle \phi_\pm , \psi_\pm \rangle \chi_\pm \otimes c( u_\pm , v_\pm ) \cdot w_\pm = 0
\end{equation}
and
\begin{equation}\label{eq:SP111n2peven2}
\langle \phi_\pm , \psi_\pm \rangle \chi_\mp \otimes c( u_\pm , v_\pm ) \cdot w_\mp = 0~,  
\end{equation}
for all $\psi_\pm \otimes v_\pm , \chi_\pm \otimes w_\pm , \phi_\pm \otimes u_\pm \in \cS_\pm^{(n)} \otimes V_\pm$, since $\langle \cS_\pm^{(n)} , \cS_\mp^{(n)} \rangle = 0$. \eqref{eq:SP111n2peven2} is equivalent to 
\begin{equation}\label{eq:SP111n2peven2Final}
c( u_\pm , v_\pm ) \cdot w_\mp = 0~,  
\end{equation}
for all $ u_\pm , v_\pm , w_\pm \in V_\pm$. On the other hand, \eqref{eq:gInvariantb} implies
\begin{equation}\label{eq:SP111n4bInvariance}
b ( c ( u_\pm , v_\pm ) \cdot w_\pm , w_\mp ) + b ( w_\pm , c ( u_\pm , v_\pm ) \cdot  w_\mp ) =0~,
\end{equation}
for all $ u_\pm , v_\pm , w_\pm \in V_\pm$. Therefore \eqref{eq:SP111n2peven2Final} and \eqref{eq:SP111n4bInvariance} imply 
\begin{equation}\label{eq:SP111n2peven2Final2}
c( u_\pm , v_\pm ) \cdot w_\pm = 0~,  
\end{equation}
for all $ u_\pm , v_\pm , w_\pm \in V_\pm$, since $b$ is nondegenerate and $b ( V_\pm , V_\pm ) =0$. \eqref{eq:SP111n2peven2Final2} and \eqref{eq:SP111n2peven2Final} describe the conventional solution of \eqref{eq:SP111n2peven1} and \eqref{eq:SP111n2peven2}.
\end{proof}

%%%%%%%%%%%%%%%%%%%%%%

\section{Central charges and R-symmetry}
\label{sec:CentralChargesAndRSymmetry}

In order to provide some context for our construction of Poincar\'{e} superalgebras based on unconventional solutions of \eqref{eq:SP111} in the forthcoming sections, in this section we will give a brief account of the general nature of Poincar\'{e} superalgebras based on conventional solutions.

Let $\fh$ denote the image of the map $c$ in \eqref{eq:c}. The condition \eqref{eq:gEquivariantc} implies that $\fh$ is an ideal of $\fg$. For the conventional solution, by definition, the action of $\fh$ on $V$ is trivial. Restricting $X$ in \eqref{eq:gEquivariantc} to $\fh$ therefore implies that $\fh$ is an abelian Lie algebra.

Motivated by \cite{Coleman:1967ad}, it is often assumed that $\fg$ is reductive. This means
\begin{equation}\label{eq:gReductive}
\fg = \fs \oplus \fa~,
\end{equation}
in terms of a semisimple Lie algebra $\fs$ and an abelian Lie algebra $\fa$ which commute with each other. Since $\fh$ is an abelian ideal of $\fg$, it must be a Lie subalgebra of $\fa$. Therefore $\fh$ is central in $\fP (n)_{\bar 0}$. Furthermore, restricting $X$ in \eqref{eq:SP01X} to $\fh$ implies that $\fh$ is central in $\fP (n)$. Because of these properties, elements in $\fh$ are traditionally referred to as {\emph{central charges}}. 

Since $\fh$ is central in $\fg$, \eqref{eq:gEquivariantc} implies
\begin{equation}\label{eq:shBracket}
c ( X \cdot v , w ) + c ( v , X \cdot w ) =0~,
\end{equation}
for all $X \in \fg / \fh$ and $v,w \in V$. We can assume without loss of generality that the action  of $\fg / \fh$ on $V$ is faithful (otherwise we just quotient $\fg / \fh$ by the kernel of the action to make it faithful). Conditions \eqref{eq:gInvariantb} and \eqref{eq:shBracket} then imply that $\fg / \fh$ is isomorphic to a Lie subalgebra of 
\begin{equation}\label{eq:Rmax}
\fr = 
\begin{cases}
 \{ R \in \fgl (V)~|~b(Rv,w)+b(v,Rw)=0,~c(Rv,w)+c(v,Rw)=0,~v,w \in V \} \quad\quad\quad\quad\quad\quad\quad\quad\;\, {\mbox{if $n$ is odd}}~, \\
 \{ R \in \fgl (V_+) \oplus \fgl (V_-)~|~b(Rv,w)+b(v,Rw)=0,~c(Rv,w)+c(v,Rw)=0,~v,w \in V_+ \oplus V_- \} \quad {\mbox{if $n$ is even}}~.
\end{cases}
\end{equation}
This Lie subalgebra describes the {\emph{R-symmetry}} and $\fr$ in \eqref{eq:Rmax} is the largest possible R-symmetry algebra.

For the conventional solution, a Poincar\'{e} superalgebra is therefore determined simply by choosing $\dim \, \fh$ (the number of central charges) and a Lie subalgebra of $\fr$ (the R-symmetry).

If there are no central charges (i.e. $\fh =0$) then \eqref{eq:Rmax} only requires that $b$ must be $\fr$-invariant. If $n=2p+1$ then \eqref{eq:bSym} and Table~\ref{tab:SigmaTauSigns} imply that $b$ is symmetric if $p=0,1 \!\!\mod 4$ (in which case $\fr = \fso (V)$) and skewsymmetric if $p=2,3\!\!\mod 4$ (in which case $\fr = \fsp (V)$). If $n=2p$ and $p$ is odd then \eqref{eq:bcPlusMinuspOdd} implies that $b$ defines a nondegenerate bilinear form on both $V_+$ and $V_-$. Moreover, \eqref{eq:bSym} and Table~\ref{tab:SigmaTauSigns} imply that both these bilinear forms are symmetric if $p=1 \!\!\mod 4$ (in which case $\fr = \fso (V_+) \oplus \fso(V_-)$) and skewsymmetric if $p=3\!\!\mod 4$ (in which case $\fr = \fsp (V_+) \oplus \fsp(V_-)$). If $n=2p$ and $p$ is even then \eqref{eq:bcPlusMinuspEven} implies $b ( V_\pm , V_\pm ) = 0$ and any $R \in \fgl (V_+) \oplus \fgl (V_-)$ is in $\fr$ if and only if 
\begin{equation}\label{eq:Rmaxn4p}
b(R_+ v_+ , v_- )+b( v_+ , R_- v_- )=0~,
\end{equation}
for all $v_\pm \in V_\pm$, where $R = ( R_+ , R_- )$ and $R_\pm \in \fgl (V_\pm)$. Since $b$ is nondegenerate, it defines a dual pairing between $V_+$ and $V_-$ with respect to which \eqref{eq:Rmaxn4p} implies $R_+^* = - R_-$. Therefore
\begin{equation}\label{eq:Rmaxn4p2}
\fr = \{ ( R_\pm , - R_\pm^* )  \in \fgl (V_\pm) \oplus \fgl (V_\mp) \} \cong \fgl (V_\pm)~.
\end{equation}
Table~\ref{tab:RSymNoCentralCharges} summarises these maximal R-symmetry algebras in the absence of central charges.
\begin{table}
  \centering
  \caption{Maximal R-symmetry algebras without central charges}
  \label{tab:RSymNoCentralCharges}
  \begin{tabular}{|c||c|c|c|c|c|c|c|c|}
   \hline
   &&&&&&&& \\ [-.35cm]
    $n$ mod $8$ & 1 & 2 & 3 & 4 & 5 & 6 & 7 & 8 \\ [.05cm] \hline 
    &&&&&&&& \\ [-.35cm]
    $\fr$ & $\fso(V)$ & $\fso (V_+) \oplus \fso(V_-)$ & $\fso(V)$ & $\fgl (V_\pm)$ & $\fsp(V)$ & $\fsp (V_+) \oplus \fsp(V_-)$ & $\fsp(V)$ & $\fgl (V_\pm)$ \\ [.05cm]
    \hline
  \end{tabular}
\end{table}

%%%%%%%%%%%%%%%%%%%%%%

\section{Poincar\'{e} superalgebras as Lie superalgebra extensions}
\label{sec:DeconstructionAndReconstructionViaLieSuperalgebraExtensions}

In this section we will show how to express a Poincar\'{e} superalgebra in terms a more elementary {\emph{embedding superalgebra}} which encodes the map $c$ in \eqref{eq:c}. This will allow us to define a notion of equivalence between Poincar\'{e} superalgebras and we will classify (up to isomorphism) all classical embedding superalgebras which give rise to inequivalent Poincar\'{e} superalgebras.

%%%%%%%%%%%%%%%%%%%%%%

\subsection{Semidirect sums and abelian extensions}
\label{sec:SemidirectSumsAndAbelianExtensions}

Let $\fk$ be a Lie algebra and let $\Der \fG$ denote the Lie algebra of even derivations of a Lie superalgebra $\fG$, i.e.
\begin{equation}\label{eq:DerG}
\Der \fG = \{ D \in \fgl ( \fG )~|~D \fG_{\bar \alpha}  \subset \fG_{\bar \alpha} ,~D[a,b] = [Da,b] + [a,Db] ,~ {\bar \alpha} \in \{ {\bar 0} , {\bar 1} \} ,~a,b \in \fG \}~.
\end{equation}
For any Lie algebra homomorphism $\theta : \fk \rightarrow \Der \fG$, let $\fk \oplus_\theta \fG$ denote the vector superspace $\fk \oplus \fG$ (thinking of $\fk$ as a vector superspace whose odd part is zero) equipped with the brackets 
\begin{align}
[X,Y]_\theta &= [X,Y]~, \label{eq:SemiDirectSumBracketsXY} \\
[X,a]_\theta &= \theta (X) a~, \label{eq:SemiDirectSumBracketsXa} \\
[a,b]_\theta &= [a,b]~, \label{eq:SemiDirectSumBracketsab}
\end{align}
for all $X,Y \in \fk$ and $a,b \in \fG$. 

\begin{prop} \label{prop:SemiDirectSum}
$\fk \oplus_\theta \fG$ is a Lie superalgebra.
\end{prop}

\begin{proof}
Since $\fk$ is a Lie algebra and $\fG$ is a Lie superalgebra then \eqref{eq:SemiDirectSumBracketsXY} and \eqref{eq:SemiDirectSumBracketsab} imply that the $[ \fk \fk \fk ]$ and $[ \fG \fG \fG ]$ parts of the Jacobi identity for $\fk \oplus_\theta \fG$ are automatically satisfied. The remaining parts are satisfied since
\begin{align}\label{eq:SemiDirectSumLSA} 
[ X , [ Y , a ]_\theta ]_\theta &\underset{\eqref{eq:SemiDirectSumBracketsXa}}{=} \theta (X) \theta (Y) a \nonumber \\
&\hspace*{.15cm} = \theta ( [X,Y] ) a + \theta (Y) \theta (X) a \nonumber \\ 
&\underset{\eqref{eq:SemiDirectSumBracketsXa}}{=}  [ [ X , Y ] , a ]_\theta + [ Y , [ X , a ]_\theta ]_\theta \nonumber \\
&\underset{\eqref{eq:SemiDirectSumBracketsXY}}{=}  [ [ X , Y ]_\theta , a ]_\theta + [ Y , [ X , a ]_\theta ]_\theta~, \nonumber \\
[ X , [ a , b ]_\theta ]_\theta &\underset{\eqref{eq:SemiDirectSumBracketsab}}{=} [ X , [ a , b ] ]_\theta \nonumber \\
&\underset{\eqref{eq:SemiDirectSumBracketsXa}}{=} \theta ( X ) [ a , b ] \nonumber \\ 
&\underset{\eqref{eq:DerG}}{=}  [ \theta ( X ) a , b ] +  [ a , \theta ( X ) b ]  \nonumber \\
&\underset{\eqref{eq:SemiDirectSumBracketsXa}}{=}  [ [ X , a ]_\theta , b ] + [ a , [ X , b ]_\theta ]~, \nonumber \\
&\underset{\eqref{eq:SemiDirectSumBracketsab}}{=}  [ [ X , a ]_\theta , b ]_\theta + [ a , [ X , b ]_\theta ]_\theta~, 
\end{align}
for any $X,Y \in \fk$ and $a,b \in \fG$.
\end{proof}

The Lie superalgebra $\fk \oplus_\theta \fG$ is called a {\emph{semidirect sum}} of $\fk$ and $\fG$. If $\theta = 0$ then $\fk \oplus_0 \fG$ is the direct sum $\fk \oplus \fG$ of $\fk$ and $\fG$ as Lie superalgebras (not just as vector superspaces).

A vector superspace $\fM$ is called a {\emph{representation}} of a Lie superalgebra $\fG$ if it is equipped with an even bilinear map
\begin{equation}\label{eq:LSAAction} 
\fG \times \fM \longrightarrow \fM
\end{equation}
obeying 
\begin{equation}\label{eq:LSAActionBracket}
[ a , b ] \cdot x = a \cdot ( b \cdot x ) - (-1)^{|a||b|} b \cdot ( a \cdot x )~,
\end{equation}
for all homogeneous $a , b \in \fG$ and $x \in \fM$, where $\cdot$ denotes the action of $\fG$ on $\fM$ defined by \eqref{eq:LSAAction}.

\begin{remark} \label{rem:EvenSuperRepresentation}
If the odd part of $\fM$ is zero then $\fG_{\bar 1} \cdot \fM = 0$ and \eqref{eq:LSAActionBracket} says that $\fM$ is a representation of the Lie algebra $\fG_{\bar 0}$ with $[ \fG_{\bar 1} , \fG_{\bar 1} ] \cdot \fM = 0$.
\end{remark}

Let $C^1 ( \fG , \fM )$ denote the space of even linear maps $\fG \rightarrow \fM$ and let $C^2 ( \fG , \fM )$ denote the space of even superskewsymmetric bilinear maps $\fG \times \fG \rightarrow \fM$. The space $Z^2 ( \fG , \fM )$ is spanned by those $f \in C^2 ( \fG , \fM )$ which obey
\begin{equation}\label{eq:LSACocycles}
f(a,[b,c]) + a \cdot f(b,c) = f([a,b],c) - (-1)^{|c| ( |a| + |b| )} c \cdot f(a,b) + (-1)^{|a||b|} f(b,[a,c]) + (-1)^{|a||b|} b \cdot f(a,c)~,
\end{equation}
for all homogeneous $a,b,c \in \fG$. The space $B^2 ( \fG , \fM )$ is spanned by those $f \in C^2 ( \fG , \fM )$ of the form
\begin{equation}\label{eq:LSACoboundaries}
f(a,b) = a \cdot g(b) - (-1)^{|a||b|} b \cdot g(a) - g([a,b])~,
\end{equation}
for all homogeneous $a,b \in \fG$, in terms of some $g \in C^1 ( \fG , \fM )$. It is easily verified that $B^2 ( \fG , \fM ) \subset Z^2 ( \fG , \fM )$ and we define the quotient space 
\begin{equation}\label{eq:LSACohomology}
H^2 ( \fG , \fM ) =  Z^2 ( \fG , \fM ) / B^2 ( \fG , \fM )~.
\end{equation}

For any $f \in Z^2 ( \fG , \fM )$, let $\fG_f$ denote the vector superspace $\fG \oplus \fM$ equipped with the brackets 
\begin{align}
[a,b]_f &= [a,b] + f(a,b)~, \label{eq:AbelianExtensionBracketsab} \\
[a,x]_f &= a\cdot x~, \label{eq:AbelianExtensionBracketsax} \\
[x,y]_f &= 0~, \label{eq:AbelianExtensionBracketsxy}
\end{align}
for all $a,b \in \fG$ and $x,y \in \fM$. 

\begin{prop} \label{prop:AbelianExtension}
$\fG_f$ is a Lie superalgebra.
\end{prop}

\begin{proof}
For any homogeneous $a,b,c \in \fG$ and $x,y,z \in \fM$,
\begin{align}\label{eq:AbelianExtensionLSA} 
[ a , [ b , c ]_f ]_f &\underset{\eqref{eq:AbelianExtensionBracketsab}}{=} [ a , [ b , c ] ]_f + [ a , f(b,c) ]_f \nonumber \\
&\overset{\eqref{eq:AbelianExtensionBracketsab}}{\underset{\eqref{eq:AbelianExtensionBracketsax}}{=}} [ a , [ b , c ] ] + f(a,[b,c]) + a \cdot f(b,c)  \nonumber \\ 
&\overset{\eqref{eq:LSAJacobiG}}{\underset{\eqref{eq:LSACocycles}}{=}} [ [ a , b ] , c ] + (-1)^{|a||b|} [ b , [ a , c ] ] + f([a,b],c) - (-1)^{|c| ( |a| + |b| )} c \cdot f(a,b)  \nonumber \\ 
&\hspace*{.7cm} + (-1)^{|a||b|} f(b,[a,c]) + (-1)^{|a||b|} b \cdot f(a,c) \nonumber \\
&\underset{\eqref{eq:AbelianExtensionBracketsab}}{=} [ [ a , b ] , c ]_f + (-1)^{|a||b|} [ b , [ a , c ] ]_f - (-1)^{|c| ( |a| + |b| )} c \cdot f(a,b) + (-1)^{|a||b|} b \cdot f(a,c) \nonumber \\
&\overset{\eqref{eq:AbelianExtensionBracketsab}}{\underset{\eqref{eq:AbelianExtensionBracketsax}}{=}}  [ [ a , b ]_f , c ]_f + (-1)^{|a||b|} [ b , [ a , c ]_f ]_f~, \nonumber \\
[ a , [ b , x ]_f ]_f &\underset{\eqref{eq:AbelianExtensionBracketsax}}{=} [ a , b \cdot x ]_f \nonumber \\
&\underset{\eqref{eq:AbelianExtensionBracketsax}}{=} a \cdot ( b \cdot x ) \nonumber \\
&\hspace*{.1cm}\underset{\eqref{eq:LSAActionBracket}}{=} [ a , b ] \cdot x + (-1)^{|a||b|} b \cdot ( a \cdot x ) \nonumber \\
&\underset{\eqref{eq:AbelianExtensionBracketsax}}{=} [ [ a , b ] , x ]_f + (-1)^{|a||b|} [ b , [ a , x ]_f ]_f \nonumber \\
&\overset{\eqref{eq:AbelianExtensionBracketsab}}{\underset{\eqref{eq:AbelianExtensionBracketsxy}}{=}} [ [ a , b ]_f , x ]_f + (-1)^{|a||b|} [ b , [ a , x ]_f ]_f~, \nonumber \\
[ a , [ x , y ]_f ]_f &\underset{\eqref{eq:AbelianExtensionBracketsxy}}{=} 0 \nonumber \\
&\underset{\eqref{eq:AbelianExtensionBracketsxy}}{=} [ a \cdot x , y ]_f + (-1)^{|a||x|} [ x , a \cdot y ]_f  \nonumber \\
&\underset{\eqref{eq:AbelianExtensionBracketsax}}{=} [ [ a , x ]_f , y ]_f + (-1)^{|a||x|} [ x , [ a , y ]_f  ]_f~,  \nonumber \\
[ x , [ y , z ]_f ]_f &\underset{\eqref{eq:AbelianExtensionBracketsxy}}{=} 0 \nonumber \\
&\underset{\eqref{eq:AbelianExtensionBracketsxy}}{=} [ [ x , y ]_f , z ]_f + (-1)^{|x||y|} [ y , [ x , z ]_f  ]_f~.  
\end{align}
\end{proof}

The Lie superalgebra $\fG_f$ is called an {\emph{abelian extension}} of $\fG$ by $\fM$ (where $\fM$ is thought of as an abelian Lie superalgebra). If $\fM$ is the trivial representation of $\fG$ then the abelian extension $\fG_f$ is called a {\emph{central extension}} (since $[ \fM , \fG_f ] =0$ in that case).

If $f,g \in Z^2 ( \fG , \fM )$ then the Lie superalgebras $\fG_f$ and $\fG_g$ are said to be {\emph{equivalent}} to each other (as abelian extensions of $\fG$ by $\fM$) if $f - g \in B^2 ( \fG , \fM )$. 
 
\begin{prop} \label{prop:AbelianExtensionIsomorphism}
If $\fG_f$ and $\fG_g$ are equivalent to each other then $\fG_f \cong \fG_g$ as Lie superalgebras.
\end{prop}

\begin{proof}
Since $f-g \in B^2 ( \fG , \fM )$ then \eqref{eq:LSACoboundaries} implies there must exist some $h \in C^1 ( \fG , \fM )$ such that 
\begin{equation}\label{eq:LSACoboundaryh}
f(a,b) - g(a,b)= a \cdot h(b) - (-1)^{|a||b|} b \cdot h(a) - h([a,b])~,
\end{equation}
for all homogeneous $a,b \in \fG$.

Let us define the even linear map
\begin{equation}\label{eq:AbelianExtensionIsomorphism}
\phi : \fG_f \longrightarrow \fG_g
\end{equation}
such that 
\begin{equation}\label{eq:AbelianExtensionIsomorphism2}
\phi (a) = a + h(a)~, \quad \phi (x) = x~,
\end{equation}
for all $a \in \fG$ and $x \in \fM$. Clearly $\phi$ is bijective with inverse $\phi^{-1}$ defined such that $\phi^{-1} (a) = a - h(a)$, $\phi^{-1} (x) = x$, for all $a \in \fG$ and $x \in \fM$.

Furthermore, $\phi$ is a Lie superalgebra homomorphism since
\begin{align}\label{eq:AbelianExtensionHomomorphism} 
\phi ( [ a , b ]_f ) &\underset{\eqref{eq:AbelianExtensionBracketsab}}{=} \phi ( [a,b] ) + \phi ( f(a,b) ) \nonumber \\
&\underset{\eqref{eq:AbelianExtensionIsomorphism2}}{=} [a,b] + h ([a,b]) + f(a,b)  \nonumber \\ 
&\underset{\eqref{eq:LSACoboundaryh}}{=} [a,b] + g(a,b) + a \cdot h(b) - (-1)^{|a||b|} b \cdot h(a)   \nonumber \\ 
&\overset{\eqref{eq:AbelianExtensionBracketsab}}{\underset{\eqref{eq:AbelianExtensionBracketsax}}{=}} [a,b]_g + [ a , h(b) ]_g + [ h(a) , b ]_g \nonumber \\
&\overset{\eqref{eq:AbelianExtensionIsomorphism2}}{\underset{\eqref{eq:AbelianExtensionBracketsxy}}{=}} [ \phi (a) , \phi (b) ]_g~, \nonumber \\
\phi ( [ a , x ]_f ) &\underset{\eqref{eq:AbelianExtensionBracketsax}}{=} \phi ( a \cdot x ) \nonumber \\
&\underset{\eqref{eq:AbelianExtensionIsomorphism2}}{=} a \cdot x \nonumber \\ 
&\underset{\eqref{eq:AbelianExtensionBracketsax}}{=} [ a , x ]_g  \nonumber \\ 
&\overset{\eqref{eq:AbelianExtensionIsomorphism2}}{\underset{\eqref{eq:AbelianExtensionBracketsxy}}{=}} [ \phi (a) , \phi (x) ]_g~, \nonumber \\
\phi ( [ x , y ]_f ) &\underset{\eqref{eq:AbelianExtensionBracketsxy}}{=} 0 \nonumber \\
&\underset{\eqref{eq:AbelianExtensionBracketsxy}}{=} [ \phi (x) , \phi (y) ]_g~, 
\end{align}
for all homogeneous $a,b \in \fG$ and $x,y \in \fM$.
\end{proof}

The Lie superalgebra $\fG_f$ will be called {\emph{non-trivial}} (as an abelian extension of $\fG$ by $\fM$) if it is not equivalent to $\fG_0$ (i.e. if $f \notin B^2 ( \fG , \fM )$).

%%%%%%%%%%%%%%%%%%%%%%

\subsection{Embedding superalgebras}
\label{sec:EmbeddingLSA}

The {\emph{embedding superalgebra}} $\fG (n)$ in $n$ dimensions is defined in terms of the data for the Poincar\'{e} superalgebra $\fP (n)$ (see \eqref{eq:SP0} and  \eqref{eq:SP1}). 

As a vector superspace, $\fG (n)$ has even part 
\begin{equation}\label{eq:ELSA0}
\fG (n)_{\bar 0} = \fg
\end{equation}
and odd part 
\begin{equation}\label{eq:ELSA1}
\fG (n)_{\bar 1} = 
\begin{cases}
\cS^{(n)} \otimes V \quad\quad\quad\quad\quad\quad\quad\quad {\mbox{if $n$ is odd}}~, \\
\cS_+^{(n)} \otimes V_+ \oplus \cS_-^{(n)} \otimes V_- \quad\quad\, {\mbox{if $n$ is even}}~.
\end{cases}
\end{equation}

The $[ {\bar 0} {\bar 0} ]$ bracket for $\fG (n)$ is the Lie bracket for $\fg$. The remaining brackets are 
\begin{align} 
[ X , \psi \otimes v ] &= \psi \otimes X \cdot v~, \label{eq:ELSA01} \\
[ \psi \otimes v , \chi \otimes w ] &= \langle \psi , \chi \rangle c(v,w)~, \label{eq:ELSA11} 
\end{align}
for all $X \in \fg$ and $\psi \otimes v , \chi \otimes w \in \fG (n)_{\bar 1}$ (see \eqref{eq:SP01X} and  \eqref{eq:SP11}). 

If $\fP (n)$ is a Lie superalgebra then clearly so is $\fG (n)$ (see Propositions~\ref{prop:001Jacobi},~\ref{prop:011Jacobi} and Theorem~\ref{thm:111Jacobi}). However, as it stands, $\fG (n)$ is a Lie superalgebra if and only if 
\begin{itemize}
\item $\fg$ is a Lie algebra (the $[{\bar 0}{\bar 0}{\bar 0}]$ Jacobi identity),
 
\item $V$, $V_\pm$ are representations of $\fg$ (the $[{\bar 0}{\bar 0}{\bar 1}]$ Jacobi identity), 

\item $c$ is $\fg$-equivariant (the $[{\bar 0}{\bar 1}{\bar 1}]$ Jacobi identity), 

\item $c$ satsfies the conditions in Theorem~\ref{thm:111Jacobi} (the $[{\bar 1}{\bar 1}{\bar 1}]$ Jacobi identity). 
\end{itemize}
It follows that $\fP (n)$ is a Lie superalgebra if and only if $\fG (n)$ is a Lie superalgebra and $b$ is $\fg$-invariant. 

\begin{remark} \label{rem:GnLSA}
Any Lie superalgebra $\fG$ with odd part of the form \eqref{eq:ELSA1} (for some $n$) and $[ {\bar 1} {\bar 1} ]$ bracket of the form \eqref{eq:ELSA11} may be construed as an embedding superalgebra in $n$ dimensions. In this case, $\fG = \fG (n)$ will be called {\emph{$n$-admissible}} if $\fP (n)$ is a Lie superalgebra. A more precise characterisation of $n$-admissible Lie superalgebras with $n<4$ will be determined in Section~\ref{sec:CharacterisationEmbeddingLSA}.
\end{remark}

\begin{prop} \label{prop:Theta}
The linear map
\begin{equation}\label{eq:Theta}
\theta : \fso ( \cV ) \longrightarrow \Der \fG (n)
\end{equation}
defined such that 
\begin{equation}\label{eq:Theta2}
\theta (A) X = 0~, \quad \theta (A) ( \psi \otimes v ) = s_A \cdot \psi \otimes v~,
\end{equation}
for all $A \in \fso ( \cV )$, $X \in \fg$ and $\psi \otimes v \in \fG (n)_{\bar 1}$, is a Lie algebra homomorphism.
\end{prop}

\begin{proof}
If $A \in \fso ( \cV )$ then $\theta (A) \in \Der \fG (n)$ since
\begin{align}\label{eq:ThetaADerivation} 
\theta (A) [X,Y] &\underset{\eqref{eq:Theta2}}{=} 0 \nonumber \\
&\underset{\eqref{eq:Theta2}}{=} [ \theta (A) X,Y] + [ X , \theta (A) Y]~, \nonumber \\
\theta (A) [X, \psi \otimes v ] &\underset{\eqref{eq:ELSA01}}{=} \theta (A) ( \psi \otimes X \cdot v ) \nonumber \\
&\underset{\eqref{eq:Theta2}}{=} s_A \cdot \psi \otimes X \cdot v \nonumber \\
&\underset{\eqref{eq:ELSA01}}{=} [ X , s_A \cdot \psi \otimes v ]  \nonumber \\
&\underset{\eqref{eq:Theta2}}{=} [ \theta (A) X , \psi \otimes v ] + [ X , \theta (A) ( \psi \otimes v ) ]~, \nonumber \\
\theta (A) [ \psi \otimes v , \chi \otimes w ] &\underset{\eqref{eq:ELSA11}}{=} \langle \psi , \chi \rangle  \theta (A) c(v,w) \nonumber \\
&\underset{\eqref{eq:Theta2}}{=} 0 \nonumber \\
&\underset{\eqref{eq:SpinInvariance}}{=} ( \langle s_A \cdot \psi , \chi \rangle +  \langle \psi , s_A \cdot \chi \rangle ) c(v,w) \nonumber \\
&\underset{\eqref{eq:ELSA11}}{=} [ s_A \cdot \psi \otimes v , \chi \otimes w ] + [ \psi \otimes v , s_A \cdot \chi \otimes w ] \nonumber \\
&\underset{\eqref{eq:Theta2}}{=} [ \theta (A) ( \psi \otimes v ) , \chi \otimes w ] + [ \psi \otimes v , \theta (A) ( \chi \otimes w ) ]~, 
\end{align}
for all $X,Y \in \fg$ and $\psi \otimes v , \chi \otimes w \in \fG (n)_{\bar 1}$.

Furthermore, $\theta$ is a Lie algebra homomorphism since
\begin{align}\label{eq:ThetaHomomorphism} 
\theta ([A,B]) X &\underset{\eqref{eq:Theta2}}{=} 0 \nonumber \\
&\underset{\eqref{eq:Theta2}}{=} [ \theta (A) , \theta (B) ] X~, \nonumber \\
\theta ([A,B]) ( \psi \otimes v ) &\underset{\eqref{eq:Theta2}}{=} s_{[A,B]} \cdot \psi \otimes v \nonumber \\
&\hspace*{.05cm}\underset{\eqref{eq:OmegaABCommutator}}{=} [ s_A , s_B ] \cdot \psi \otimes v \nonumber \\
&\underset{\eqref{eq:Theta2}}{=} [ \theta (A) , \theta (B) ]  ( \psi \otimes v )~, 
\end{align}
for all $A,B \in \fso ( \cV )$, $X \in \fg$ and $\psi \otimes v \in \fG (n)_{\bar 1}$.
\end{proof}

\begin{cor} \label{cor:SemiDirectSumSOVGn}
By definition, the semidirect sum ${\bar \fP} (n) = \fso ( \cV ) \oplus_\theta \fG (n)$ is the Lie superalgebra with brackets 
\begin{align}
[A,B]_\theta &= [A,B]~, \label{eq:SemiDirectSOVGnSumBracketsAB} \\
[A,X]_\theta &= 0~, \label{eq:SemiDirectSumSOVGnBracketsAX} \\
[A, \psi \otimes v ]_\theta &= s_A \cdot \psi \otimes v~, \label{eq:SemiDirectSumSOVGnBracketsAv} \\
[X,Y]_\theta &= [X,Y]~, \label{eq:SemiDirectSumSOVGnBracketsXY} \\
[ X , \psi \otimes v ]_\theta &= \psi \otimes X \cdot v~, \label{eq:SemiDirectSumSOVGnBracketsXv} \\
[ \psi \otimes v , \chi \otimes w ]_\theta &= \langle \psi , \chi \rangle c(v,w)~, \label{eq:SemiDirectSumSOVGnBracketsvw} 
\end{align}
for all $A,B \in \fso ( \cV )$, $X,Y \in \fg$ and $\psi \otimes v , \chi \otimes w \in \fG (n)_{\bar 1}$.  
\end{cor}

Remark~\ref{rem:EvenSuperRepresentation} implies that $\cV$ (thought of as a vector superspace whose odd part is zero) is a representation of the Lie superalgebra ${\bar \fP} (n)$ if it is a representation of the Lie algebra ${\bar \fP} (n)_{\bar 0} = \fso ( \cV ) \oplus \fg$ with ${\bar \fP} (n)_{\bar 1} \cdot \cV = 0$ and $[ {\bar \fP} (n)_{\bar 1} , {\bar \fP} (n)_{\bar 1}  ] \cdot \cV = 0$. Since ${\bar \fP} (n)_{\bar 1} = \fG (n)_{\bar 1}$ and $[ \fG (n)_{\bar 1} , \fG (n)_{\bar 1} ] \subset \fg$, this is accomplished by defining 
\begin{align}
A \cdot x &= Ax~, \label{eq:PbarActionAx} \\
X \cdot x &= 0~, \label{eq:PbarActionXx} \\
( \psi \otimes v ) \cdot x &= 0~, \label{eq:PbarActionvx} 
\end{align}
for all $A \in \fso ( \cV )$, $x \in \cV$, $X \in \fg$ and $\psi \otimes v \in \fG (n)_{\bar 1}$. 

\begin{prop} \label{prop:PbarCocycle}
If $f \in C^2 ( {\bar \fP} (n) , \cV )$ is defined such that $f ( {\bar \fP} (n)_{\bar 0} , {\bar \fP} (n) ) =0$ and
\begin{equation}\label{eq:PbarCocycle}
f ( \psi \otimes v , \chi \otimes w ) = \xi ( \psi , \chi ) b (v,w)~,
\end{equation}
for all $\psi \otimes v , \chi \otimes w \in {\bar \fP} (n)_{\bar 1}$, then $f \in Z^2 ( {\bar \fP} (n) , \cV )$ if and only if $b$ is $\fg$-invariant.
\end{prop}

\begin{proof}
Since $f ( {\bar \fP} (n)_{\bar 0} , {\bar \fP} (n) ) =0$, it follows that \eqref{eq:LSACocycles} is automatically satisfied unless at least two of the three homogeneous elements are in ${\bar \fP} (n)_{\bar 1}$. Furthermore, if all three elements are in ${\bar \fP} (n)_{\bar 1}$ then \eqref{eq:PbarActionvx} implies that \eqref{eq:LSACocycles} is satisfied. 

Therefore \eqref{eq:LSACocycles} is equivalent to 
\begin{equation}\label{eq:PbarCocycle011}
a \cdot f( \psi \otimes v , \chi \otimes w ) = f([a, \psi \otimes v ]_\theta , \chi \otimes w )  + f( \psi \otimes v , [a, \chi \otimes w ]_\theta )~,
\end{equation}
for all $a \in {\bar \fP} (n)_{\bar 0}$ and $\psi \otimes v , \chi \otimes w \in {\bar \fP} (n)_{\bar 1}$. For any $A \in \fso ( \cV )$, \eqref{eq:PbarCocycle011} is automatically satisfied since
\begin{align}\label{eq:PbarCocycle011A} 
A \cdot f ( \psi \otimes v , \chi \otimes w ) &\underset{\eqref{eq:PbarActionAx}}{=} A f ( \psi \otimes v , \chi \otimes w ) \nonumber \\
&\underset{\eqref{eq:PbarCocycle}}{=} A \xi ( \psi , \chi ) b (v,w) \nonumber \\
&\underset{\eqref{eq:SpinEquivariance}}{=} ( \xi ( s_A \cdot \psi , \chi ) +  \xi ( \psi , s_A \cdot \chi ) ) b (v,w) \nonumber \\
&\underset{\eqref{eq:PbarCocycle}}{=} f ( s_A \cdot \psi \otimes v , \chi \otimes w) + f ( \psi \otimes v , s_A \cdot \chi \otimes w ) \nonumber \\
&\underset{\eqref{eq:SemiDirectSumSOVGnBracketsAv}}{=} f ( [ A , \psi \otimes v ]_\theta , \chi \otimes w ) + f ( \psi \otimes v , [ A , \chi \otimes w ]_\theta )~.
\end{align}
For any $X \in \fg$, \eqref{eq:PbarActionXx} implies $X \cdot f ( \psi \otimes v , \chi \otimes w ) = 0$ while 
\begin{align}\label{eq:PbarCocycle011X} 
f ( [ X , \psi \otimes v ]_\theta , \chi \otimes w ) + f ( \psi \otimes v , [ X , \chi \otimes w ]_\theta ) &\underset{\eqref{eq:SemiDirectSumSOVGnBracketsXv}}{=} f ( \psi \otimes X \cdot v , \chi \otimes w ) + f ( \psi \otimes v ,  \chi \otimes X \cdot w ) \nonumber \\
&\underset{\eqref{eq:PbarCocycle}}{=}  \xi ( \psi , \chi ) ( b ( X \cdot v , w ) + b ( v , X \cdot w ) )~.
\end{align}

Therefore \eqref{eq:LSACocycles} is satisfied if and only if the right hand side of \eqref{eq:PbarCocycle011X} is zero, which is only the case if $b$ is $\fg$-invariant.
\end{proof}

\begin{cor} \label{cor:AbeilianExtensionPbar}
If $\fG (n)$ is a Lie superalgebra and $b$ is $\fg$-invariant then, by definition, the abelian extension ${\bar \fP} (n)_f$ of ${\bar \fP} (n)$ by $\cV$ is the Lie superalgebra with brackets 
\begin{align}
[A,B]_f &= [A,B]~, \label{eq:AbeilianExtensionPbarBracketsAB} \\
[A,X]_f &= 0~, \label{eq:AbeilianExtensionPbarBracketsAX} \\
[A, \psi \otimes v ]_f &= s_A \cdot \psi \otimes v~, \label{eq:AbeilianExtensionPbarBracketsAv} \\
[X,Y]_f &= [X,Y]~, \label{eq:AbeilianExtensionPbarBracketsXY} \\
[ X , \psi \otimes v ]_f &= \psi \otimes X \cdot v~, \label{eq:AbeilianExtensionPbarBracketsXv} \\
[ \psi \otimes v , \chi \otimes w ]_f &= \langle \psi , \chi \rangle c(v,w) + \xi ( \psi , \chi ) b (v,w)~, \label{eq:AbeilianExtensionPbarBracketsvw} \\
[ A , x ]_f &= Ax~, \label{eq:AbeilianExtensionPbarBracketsAx} \\
[ X , x ]_f &= 0~, \label{eq:AbeilianExtensionPbarBracketsXx} \\
[ \psi \otimes v , x ]_f &= 0~, \label{eq:AbeilianExtensionPbarBracketsvx} \\
[x,y]_f &= 0~, \label{eq:AbeilianExtensionPbarBracketsxy} 
\end{align}
for all $A,B \in \fso ( \cV )$, $X,Y \in \fg$, $\psi \otimes v , \chi \otimes w \in \fG (n)_{\bar 1}$ and $x,y \in \cV$.  
\end{cor}

Comparing the brackets in Corollary~\ref{cor:AbeilianExtensionPbar} with those for the Poincar\'{e} superalgebra in Section~\ref{sec:ExtendedPoincareSuperalgebras} shows that $\fP (n) = {\bar \fP} (n)_f$ as Lie superalgebras. 

\begin{lem} \label{lem:EquivalentAbeilianExtensionPbar}
If $n>1$ then any two Poincar\'{e} superalgebras ${\bar \fP} (n)_{f_1}$ and ${\bar \fP} (n)_{f_2}$ are equivalent to each other (as abelian extensions of ${\bar \fP} (n)$ by $\cV$) if and only if $f_1 = f_2$.
\end{lem}

\begin{proof}
By definition, ${\bar \fP} (n)_{f_1}$ and ${\bar \fP} (n)_{f_2}$ are equivalent to each other if $f_1 - f_2 \in B^2 ( {\bar \fP} (n) , \cV )$. This means that there must exist some $g \in C^1 ( {\bar \fP} (n) , \cV )$ for which 
\begin{equation}\label{eq:EquivalentAbeilianExtensionPbar}
f_1 (a,b) - f_2 (a,b) = a \cdot g(b) - (-1)^{|a||b|} b \cdot g(a) - g( [a,b]_\theta )~,
\end{equation}
for all homogeneous $a,b \in {\bar \fP} (n)$.

Substituting $a = A \in \fso ( \cV )$ and $b = X \in \fg$ into \eqref{eq:EquivalentAbeilianExtensionPbar} implies
\begin{equation}\label{eq:EquivalentAbeilianExtensionPbarAX}
0 = A g(X)~,
\end{equation}
using \eqref{eq:PbarActionAx}, \eqref{eq:PbarActionXx} and \eqref{eq:SemiDirectSumSOVGnBracketsAX}. But since \eqref{eq:EquivalentAbeilianExtensionPbarAX} must hold for all $A \in \fso ( \cV )$ and $X \in \fg$, it is equivalent to
\begin{equation}\label{eq:EquivalentAbeilianExtensionPbarAX2}
g(X) = 0~,
\end{equation}
for all $X \in \fg$, since $n>1$. Substituting $a = \psi \otimes v \in {\bar \fP} (n)_{\bar 1}$ and $b = \chi \otimes w \in {\bar \fP} (n)_{\bar 1}$ into \eqref{eq:EquivalentAbeilianExtensionPbar} implies
\begin{equation}\label{eq:EquivalentAbeilianExtensionPbarvw}
f_1 ( \psi \otimes v , \chi \otimes w ) - f_2 ( \psi \otimes v , \chi \otimes w ) = - \langle \psi , \chi \rangle g ( c(v,w) )~,
\end{equation}
using \eqref{eq:PbarActionvx} and \eqref{eq:SemiDirectSumSOVGnBracketsvw}. But $c(v,w) \in \fg$, so \eqref{eq:EquivalentAbeilianExtensionPbarAX2} implies the right hand side of \eqref{eq:EquivalentAbeilianExtensionPbarvw} is zero and $f_1 = f_2$. 
\end{proof}

%%%%%%%%%%%%%%%%%%%%%%

\subsection{Characterisation of $n$-admissible Lie superalgebras with $n<4$}
\label{sec:CharacterisationEmbeddingLSA}

\subsubsection{$n=1$}
\label{sec:CharacterisationEmbeddingLSAn1}

In this case $\fso ( \cV ) = 0$ (so ${\bar \fP} (1) = \fG (1)$) and $\dim \, \cS^{(1)} = 1$. Let $e \in \cS^{(1)}$ be normalised such that $\langle e , e \rangle = 1$ and define $\xi = \xi (e,e)$ as a basis for $\cV$. We can then let $\fG (1)_{\bar 1} = V$ by identifying every $e \otimes v \in \cS^{(1)} \otimes V$ with $v \in V$. 

The brackets for $\fP (1) = \fG (1)_f$ now look like
\begin{align}
[X,Y]_f &= [X,Y]~, \label{eq:AbeilianExtensionPbarBracketsXYn1} \\
[ X , v ]_f &= X \cdot v~, \label{eq:AbeilianExtensionPbarBracketsXvn1} \\
[ v , w ]_f &= c(v,w) + b (v,w) \xi~, \label{eq:AbeilianExtensionPbarBracketsvwn1} \\
[ X , \xi ]_f &= 0~, \label{eq:AbeilianExtensionPbarBracketsXxn1} \\
[ v , \xi ]_f &= 0~, \label{eq:AbeilianExtensionPbarBracketsvxn1} 
\end{align}
for all $X,Y \in \fg$, $v , w \in V$. Clearly $\fG (1)_f$ is a central extension of $\fG (1)$ by $\cV = \CC \xi$. 

Any Lie superalgebra $\fG$ can play the role of $\fG (1)$ so long as we identify $\fG_{\bar 0} = \fg$ and $\fG_{\bar 1} = V$. Since $\fP (1)$ is a Lie superalgebra if and only if $b$ is $\fg$-invariant, it follows that a Lie superalgebra $\fG$ is $1$-admissible if and only if $\fG_{\bar 1}$ admits a nondegenerate $\fG_{\bar 0}$-invariant symmetric bilinear form. 

For any $a \in \{ 1,2 \}$, let $f_a  \in Z^2 ( \fG (1) , \cV )$ be defined as in Proposition~\ref{prop:PbarCocycle} with $f_a ( \fG (1)_{\bar 0} , \fG (1) ) =0$ and
\begin{equation}\label{eq:PbarCocyclen1}
f_a (v,w) = b_a (v,w) \xi~,
\end{equation}
for all $v , w \in V$, where $b_a$ is a nondegenerate $\fg$-invariant symmetric bilinear form on $V$. 

\begin{lem} \label{lem:EquivalentCentralExtensionPbar}
The Poincar\'{e} superalgebras $\fG (1)_{f_1}$ and $\fG (1)_{f_2}$ are equivalent to each other (as central extensions of $\fG (1)$ by $\cV$) if and only if there exists an element $\kappa \in \fg^*$ with $\kappa ( [ \fg , \fg ] ) =0$ and 
\begin{equation}\label{eq:KappaEquivalent}
b_1 (v,w) - b_2 (v,w) = \kappa ( c(v,w) )~,
\end{equation}
for all $v , w \in V$.
\end{lem}

\begin{proof}
By definition, $\fG (1)_{f_1}$ and $\fG (1)_{f_2}$ are equivalent to each other if $f_1 - f_2 \in B^2 ( \fG (1) , \cV )$. This means that there must exist some $g \in C^1 ( \fG (1) , \cV )$ for which 
\begin{equation}\label{eq:EquivalentCentralExtensionPbar}
f_1 (a,b) - f_2 (a,b) = - g( [a,b] )~,
\end{equation}
for all homogeneous $a,b \in \fG (1)$. The condition \eqref{eq:EquivalentCentralExtensionPbar} is equivalent to
\begin{align}
0 &= - g ([X,Y])~, \label{eq:EquivalentCentralExtensionPbarXY} \\
( b_1 (v,w) - b_2 (v,w) ) \xi &= - g ( c(v,w) )~, \label{eq:EquivalentCentralExtensionPbarvw} 
\end{align}
for all $X,Y \in \fg$, $v , w \in V$.

Let $\kappa \in \fg^*$ be defined such that 
\begin{equation}\label{eq:EquivalentCentralExtensionPbarKappa}
g (X) = - \kappa (X) \xi~,
\end{equation}
for all $X \in \fg$. Then \eqref{eq:EquivalentCentralExtensionPbarXY} is equivalent to $\kappa ( [ \fg , \fg ] ) =0$ and \eqref{eq:EquivalentCentralExtensionPbarvw} is equivalent to \eqref{eq:KappaEquivalent}.
\end{proof}

\begin{cor} \label{cor:ReductiveEquivalentCentralExtensionPbar}
If $\fg$ is reductive with semisimple part $\fs$ then $\kappa ( \fs ) =0$ in Lemma~\ref{lem:EquivalentCentralExtensionPbar}. It follows that if $\fg$ is semisimple then $\fG (1)_{f_1}$ and $\fG (1)_{f_2}$ are equivalent to each other if and only if $f_1 = f_2$.
\end{cor}

\begin{proof}
Since $[ \fg , \fg ] = \fs$ then $\kappa ( [ \fg , \fg ] ) =0$ is equivalent to $\kappa ( \fs ) = 0$. If $\fg = \fs$ then this implies $\kappa = 0$ and \eqref{eq:KappaEquivalent} implies $b_1 = b_2$ so $f_1 = f_2$.
\end{proof}

\subsubsection{$n=2$}
\label{sec:CharacterisationEmbeddingLSAn2}

In this case $\dim \, \fso ( \cV ) = 1$ and $\dim \, \cS_\pm^{(2)} = 1$. Therefore any $A \in \fso ( \cV )$ acts as a scalar on the irreducible representations $\cS_\pm^{(2)}$. Let $E \in \fso ( \cV )$ be normalised such that $s_E \cdot \psi_\pm = \pm \psi_\pm$, for all $\psi_\pm \in \cS_\pm^{(2)}$. This is possible since $\langle -,- \rangle$ is $\fso ( \cV )$-invariant and \eqref{eq:AdmissibleSpinorialBilinearFormPlusMinus} implies $\langle \cS_\pm^{(2)} , \cS_\pm^{(2)} \rangle = 0$. Furthermore, \eqref{eq:XiPlusMinus} implies $\xi ( \cS_\pm^{(2)} , \cS_\mp^{(2)} ) =0$. Let $e_\pm \in \cS_\pm^{(2)}$ be normalised such that $\langle e_+ , e_- \rangle = 1$ and define $\xi_\pm = \xi ( e_\pm , e_\pm )$ as a basis for $\cV$. We can then let $\fG (2)_{\bar 1} = V_+ \oplus V_-$ by identifying every $e_\pm \otimes v_\pm \in \cS_\pm^{(2)} \otimes V_\pm$ with $v_\pm \in V_\pm$.  

The brackets for $\fP (2) = \fG (2)_f$ now look like
\begin{align}
[ E ,X]_f &= 0~, \label{eq:AbeilianExtensionPbarBracketsAXn2} \\
[ E , v_\pm ]_f &= \pm v_\pm~, \label{eq:AbeilianExtensionPbarBracketsAvn2} \\
[X,Y]_f &= [X,Y]~, \label{eq:AbeilianExtensionPbarBracketsXYn2} \\
[ X , v_\pm ]_f &= X \cdot v_\pm~, \label{eq:AbeilianExtensionPbarBracketsXvn2} \\
[ v_\pm , w_\pm ]_f &= b ( v_\pm , w_\pm ) \xi_\pm~, \label{eq:AbeilianExtensionPbarBracketsvwPlusPlusn2} \\
[ v_+ , w_- ]_f &= c( v_+ , w_- )~, \label{eq:AbeilianExtensionPbarBracketsvwnPlusMinus2} \\
[ E , \xi_\pm ]_f &= \pm 2 \xi_\pm~, \label{eq:AbeilianExtensionPbarBracketsAxn2} \\
[ X , x ]_f &= 0~, \label{eq:AbeilianExtensionPbarBracketsXxn2} \\
[ v_\pm , x ]_f &= 0~, \label{eq:AbeilianExtensionPbarBracketsvxn2} \\
[ x , y ]_f &= 0~, \label{eq:AbeilianExtensionPbarBracketsxyn2} 
\end{align}
for all $X,Y \in \fg$, $v_\pm , w_\pm \in V_\pm$ and $x,y \in \cV$.  

A Lie superalgebra $\fG$ with $\fG_{\bar 1} = \fG_{\bar 1}^+ \oplus \fG_{\bar 1}^-$, $[ \fG_{\bar 0} , \fG_{\bar 1}^\pm ] \subset \fG_{\bar 1}^\pm$ and $[ \fG_{\bar 1}^\pm , \fG_{\bar 1}^\pm ] = 0$ is said to be {\emph{3-graded}}. Any 3-graded Lie superalgebra $\fG$ can therefore play the role of $\fG (2)$ so long as we identify $\fG_{\bar 0} = \fg$ and $\fG_{\bar 1}^\pm = V_\pm$. Since $\fP (2)$ is a Lie superalgebra if and only if $b$ is $\fg$-invariant, it follows that a 3-graded Lie superalgebra $\fG$ is $2$-admissible if and only if $\fG_{\bar 1}^+$ and $\fG_{\bar 1}^-$ both admit a nondegenerate $\fG_{\bar 0}$-invariant symmetric bilinear form. 

\subsubsection{$n=3$}
\label{sec:CharacterisationEmbeddingLSAn3}

In this case $\fso ( \cV ) \cong \fsl ( \cS^{(3)} )$ and $\dim \, \cS^{(3)} = 2$. Let $E_0 , E_\pm$ be a Cartan-Weyl basis for $\fsl ( \cS^{(3)} )$ with respect to which we have the Lie brackets
\begin{equation}\label{eq:CartanWeylBasisSL2}
[ E_0 , E_\pm ] = \pm 2 E_\pm~, \quad [ E_+ , E_- ] = E_0~.
\end{equation}
The action of $\fsl ( \cS^{(3)} )$ on weight vectors $e_\pm \in \cS^{(3)}$ is then defined by
\begin{equation}\label{eq:CartanWeylBasisSL2Action}
E_0 e_\pm = \pm e_\pm~, \quad E_\pm e_\pm = 0~, \quad E_\mp e_\pm = e_\mp~.
\end{equation}
Table~\ref{tab:SigmaTauSigns}, \eqref{eq:AdmissibleSpinorialBilinearForm1} and \eqref{eq:XiSym2} imply that $\langle -,- \rangle$ is skewsymmetric while $\xi$ is symmetric. Let the weight vectors $e_\pm \in \cS^{(3)}$ be normalised such that $\langle e_+ , e_- \rangle = 1$ and define $\xi_0 = \xi ( e_+ , e_- )$, $\xi_\pm = \xi ( e_\pm , e_\pm )$ as a basis for $\cV$. We can then let $\fG (3)_{\bar 1} = V \oplus V$ by identifying every $e_+ \otimes v + e_- \otimes w \in \cS^{(3)} \otimes V$ with $(v,w) \in V \oplus V$. For any $v \in V$, let us denote $v_+ = (v,0)$ and $v_- = (0,v)$ as elements in $V \oplus V$.

The brackets for $\fP (3) = \fG (3)_f$ now look like
\begin{align}
[ E_0 , E_\pm ]_f &= \pm 2 E_\pm~, \label{eq:AbeilianExtensionPbarBracketsE0EPlusMinusn3} \\
[ E_+ , E_- ]_f &= E_0~, \label{eq:AbeilianExtensionPbarBracketsEPlusEMinusn3} \\
[A,X]_f &= 0~, \label{eq:AbeilianExtensionPbarBracketsAXn3} \\
[ E_0, v_\pm ]_f &= \pm v_\pm~, \label{eq:AbeilianExtensionPbarBracketsE0vn3} \\
[ E_\pm, v_\pm ]_f &= 0~, \label{eq:AbeilianExtensionPbarBracketsEPlusMinusvn3} \\
[ E_\mp, v_\pm ]_f &= v_\mp~, \label{eq:AbeilianExtensionPbarBracketsEMinusPlusvn3} \\
[X,Y]_f &= [X,Y]~, \label{eq:AbeilianExtensionPbarBracketsXYn3} \\
[ X , v_\pm ]_f &= ( X \cdot v )_\pm~, \label{eq:AbeilianExtensionPbarBracketsXvn3} \\
[ v_\pm , w_\pm ]_f &= b (v,w) \xi_\pm~, \label{eq:AbeilianExtensionPbarBracketsvwPlusPlusn3} \\
[ v_+ , w_- ]_f &= c(v,w) + b (v,w) \xi_0~, \label{eq:AbeilianExtensionPbarBracketsvwPlusMinusn3} \\
[ E_0 , \xi_0 ]_f &= 0~, \label{eq:AbeilianExtensionPbarBracketsE0Xi0n3} \\
[ E_0 , \xi_\pm ]_f &= \pm 2 \xi_\pm~, \label{eq:AbeilianExtensionPbarBracketsE0XiPMn3} \\
[ E_\pm , \xi_0 ]_f &= \xi_\pm~, \label{eq:AbeilianExtensionPbarBracketsEPMXi0n3} \\
[ E_\pm , \xi_\pm ]_f &= 0~, \label{eq:AbeilianExtensionPbarBracketsEPMXiPMn3} \\
[ E_\pm , \xi_\mp ]_f &= 2 \xi_0~, \label{eq:AbeilianExtensionPbarBracketsEPMXiMPn3} \\
[ X , x ]_f &= 0~, \label{eq:AbeilianExtensionPbarBracketsXxn3} \\
[ v_\pm , x ]_f &= 0~, \label{eq:AbeilianExtensionPbarBracketsvxn3} \\
[x,y]_f &= 0~, \label{eq:AbeilianExtensionPbarBracketsxyn3} 
\end{align}
for all $A \in \fso ( \cV )$, $X,Y \in \fg$, $v,w \in V$ and $x,y \in \cV$. Note that Table~\ref{tab:SigmaTauSigns} and \eqref{eq:cSym} imply $c$ is skewsymmetric.

Now let $\fG$ be a 3-graded Lie superalgebra with $\fG_{\bar 1}^+ = \fG_{\bar 1}^-$ (as representations of $\fG_{\bar 0}$). In this case we will denote $\fG_{\bar 1}^\pm$ by $\fG_{\bar 1}^\bullet$ in order to distinguish it from $( \fG_{\bar 1}^\bullet , 0 )$ and $( 0 ,  \fG_{\bar 1}^\bullet )$ in $\fG_{\bar 1}$. For any $v \in \fG_{\bar 1}^\bullet$, we will write $v_+ = (v,0)$ and $v_- = (0,v)$ as elements in $\fG_{\bar 1}$. We will say that $\fG$ is {\emph{balanced}} if
\begin{equation}\label{eq:BalancedSkew}
[ v_+ , w_- ] = - [ w_+ , v_- ]~,
\end{equation}
for all $v , w \in \fG_{\bar 1}^\bullet$. Any balanced 3-graded Lie superalgebra $\fG$ can therefore play the role of $\fG (3)$ so long as we identify $\fG_{\bar 0} = \fg$ and $\fG_{\bar 1}^\bullet = V$. Since $\fP (3)$ is a Lie superalgebra if and only if $b$ is $\fg$-invariant, it follows that a balanced 3-graded Lie superalgebra $\fG$ is $3$-admissible if and only if $\fG_{\bar 1}^\bullet$ admits a nondegenerate $\fG_{\bar 0}$-invariant symmetric bilinear form. 

%%%%%%%%%%%%%%%%%%%%%%

\subsection{Classification of classical $n$-admissible Lie superalgebras}
\label{sec:ClassificationOfSimpleAdmissibleLSA}

A Poincar\'{e} superalgebra $\fP (n)$ is never simple since it always contains the abelian ideal $\cV$. Furthermore, since ${\bar \fP} (n) \cong \fP (n) / \cV$ has $[ {\bar \fP} (n)_{\bar 1} , {\bar \fP} (n)_{\bar 1} ] \subset \fg$ and ${\bar \fP} (n)_{\bar 0} = \fso ( \cV ) \oplus \fg$, the following proposition implies that ${\bar \fP} (n)$ is never simple if ${\bar \fP} (n)_{\bar 1} \neq 0$ and $n>1$.
\!\footnote{If ${\bar \fP} (n)_{\bar 1} = 0$ and $n>1$ then ${\bar \fP} (n) = \fso ( \cV ) \oplus \fg$ is simple if and only if $\fg = 0$ and $n \neq 2,4$.}

\begin{prop} \label{prop:GSimple}
If $\fG$ is a simple Lie superalgebra with $\fG_{\bar 1} \neq 0$ then $\fG_{\bar 1}$ is a faithful representation of $\fG_{\bar 0}$ and $[ \fG_{\bar 1} , \fG_{\bar 1} ] = \fG_{\bar 0}$. 
\end{prop}

\begin{proof}
See Proposition 1.2.7 in \cite{Kac:1977em}.
\end{proof}

However, as we will now demonstrate, it is possible for an embedding superalgebra $\fG (n)$ to be simple under less restrictive conditions. Simple Lie superalgebras were classified (up to isomorphism) in \cite{Kac:1977em}. It follows from this classification that every simple Lie superalgebra $\fG$ has $\fG_{\bar 0}$ reductive if and only if $\fG$ is {\emph{classical}}, meaning that $\fG_{\bar 1}$ is a completely reducible representation of $\fG_{\bar 0}$. 

\begin{remark} \label{rem:ClassicalLSAZeroG1}
It follows that a Lie superalgebra $\fG$ with $\fG_{\bar 1} = 0$ is classical if and only if $\fG$ is a simple Lie algebra. In this case $\fG$ can obviously play the role of any $\fG (n)$ with $\fG (n)_{\bar 1} = 0$ which automatically defines a Poincar\'{e} superalgebra $\fP (n)$ with $\fP (n)_{\bar 1} = 0$. Having dealt with this trivial case, let us now assume that every classical Lie superalgebra $\fG$ has $\fG_{\bar 1} \neq 0$.
\end{remark}

Theorem 2 in \cite{Kac:1977em} implies that every classical Lie superalgebra $\fG$ is isomorphic to one of the entries in Table~\ref{tab:ClassicalLSA}. If a Lie superalgebra $\fG$ is 3-graded then $\fG_{\bar 1} = \fG_{\bar 1}^+ \oplus \fG_{\bar 1}^-$, in terms of two representations $\fG_{\bar 1}^\pm$ of $\fG_{\bar 0}$. If $\fG_{\bar 1}^\pm$ are both non-zero then the only candidates in Table~\ref{tab:ClassicalLSA} are of type ${\bf A}$, ${\bf C}$ and ${\bf P}$. In each case, $\fG_{\bar 1}^\pm$ are both irreducible representations of $\fG_{\bar 0}$. By inspecting the $[ {\bar 1} {\bar 1} ]$ brackets for these classical Lie superalgebras in \cite{Kac:1977em}, it follows that they are all 3-graded since $[ \fG_{\bar 1}^\pm , \fG_{\bar 1}^\pm ] =0$ in each case. If $\fG_{\bar 1}^+ = 0$ or $\fG_{\bar 1}^- = 0$ then $\fG$ being 3-graded implies $[ \fG_{\bar 1} , \fG_{\bar 1} ] = 0$. But Proposition~\ref{prop:GSimple} implies that every entry in Table~\ref{tab:ClassicalLSA} has $[ \fG_{\bar 1} , \fG_{\bar 1} ] = \fG_{\bar 0} \neq 0$.

\begin{table}
  \caption{Classical Lie superalgebras}
  \label{tab:ClassicalLSA}
  \hspace*{-1.1cm}
  \begin{tabular}{|c|c|c|c|}
  &&& \\ [-.5cm] 
   \hline 
    $\fG$ & $\fG_{\bar 0}$ & $\fG_{\bar 1}$ & restrictions \\ [.05cm]
    \hline\hline
    &&& \\ [-.35cm] 
    ${\bf A} ( N_+ -1 , N_- - 1 )$ & $\fsl ( W_+ ) \oplus \fsl ( W_- ) \oplus \fz$ & $W_+ \otimes W_-^* \otimes \zeta_{( N_+ - N_- )} \oplus W_+^* \otimes W_- \otimes \zeta_{( N_- - N_+ )}$ & $2 < N_+ > N_-$  \\ [.05cm] \hline 
    &&& \\ [-.35cm]
    ${\bf A} ( N-1 , N-1 )$ & $\fsl ( W_1 ) \oplus \fsl ( W_2 )$ & $W_1 \otimes W_2^* \oplus W_1^* \otimes W_2$ & $W_{1,2} = W \;\! ,~N>2$ \\ [.05cm] \hline
    &&& \\ [-.35cm]
    ${\bf A} ( 1 , 1 )$ & $\fsp ( \Delta_1 ) \oplus \fsp ( \Delta_2 )$ & $\Delta_1 \otimes \Delta_2 \oplus \Delta_1 \otimes \Delta_2$ & $\Delta_{1,2} = \Delta$ \\ [.05cm] \hline
    &&& \\ [-.35cm]
     $\fosp ( W_+ , W_- )$ & $\fso ( W_+ ) \oplus \fsp ( W_- )$ & $W_+ \otimes W_-$ & $N_+ \neq 2$ \\ [.05cm] \hline
     &&& \\ [-.35cm] 
     ${\bf C} ( \frac{N}{2} +1 )$ & $\fz \oplus \fsp ( W )$ & $\zeta_{(1)} \otimes W \oplus \zeta_{(-1)} \otimes W$ & $-$ \\ [.05cm] \hline
     &&& \\ [-.35cm] 
      ${\bf D} ( 2,1; \alpha )$ & $\fsp ( \Delta_1 ) \oplus \fsp ( \Delta_2 ) \oplus \fsp ( \Delta_3 )$ & $\Delta_1 \otimes \Delta_2 \otimes \Delta_3$ & $\Delta_{1,2,3} = \Delta \;\! ,$ \\ [.05cm]
      &&& \\ [-.35cm] 
      & & & $\alpha \in \CC \,\backslash \, \{ 0,-1 \}$ \\ [.05cm] \hline
      &&& \\ [-.35cm] 
      ${\bf F} ( 4 )$ & $\fso ( W ) \oplus \fsp ( \Delta )$ & $\cS^{(7)} \otimes \Delta$ & $N = 7$ \\ [.05cm] \hline
      &&& \\ [-.35cm] 
      ${\bf G} ( 3 )$ & $\fg_2 \oplus \fsp ( \Delta )$ & $W \otimes \Delta$ & $N = 7$ \\ [.05cm] \hline
      &&& \\ [-.35cm]
      ${\bf P} ( N -1 )$ & $\fsl ( W )$ & $S^2 W \oplus \Lambda^2 W^*$ & $N > 2$ \\ [.05cm] \hline
      &&& \\ [-.35cm] 
      ${\bf Q} ( N - 1 )$ & $\fsl ( W )$ & $\fsl (W)$ & $N > 2$ \\ [.05cm] \hline 
  \end{tabular}
  \vspace*{.2cm}
  \caption*{This table summarises the data for every classical Lie superalgebra $\fG$ with $\fG_{\bar 1} \neq 0$. The notation is such that $\dim\, W_\pm = N_\pm >0$, $\dim\, W = N >0$ and $\dim\, \Delta = 2$. In rows one and five, $\fz$ denotes a one-dimensional Lie algebra and $\zeta_{( \lambda )}$ denotes a one-dimensional irreducible representation of $\fz$ with weight $\lambda$. In row seven, $\cS^{(7)}$ denotes the spinor representation of $\fso (W)$. In row eight, the exceptional Lie algebra $\fg_2$ should be thought of as a Lie subalgebra of $\fso (W)$. }
\end{table}

\begin{prop} \label{prop:ClassicalLSABalanced3Graded}
The only balanced 3-graded classical Lie superalgebra is ${\bf A} ( 1 , 1 )$.
\end{prop}

\begin{proof}
Of the 3-graded classical Lie superalgebras in Table~\ref{tab:ClassicalLSA}, only ${\bf A} ( 1 , 1 )$ has $\fG_{\bar 1}^+ = \fG_{\bar 1}^-$. In all the other cases, $\fG_{\bar 1}^+$ and $\fG_{\bar 1}^-$ are not isomorphic to each other as representations of $\fG_{\bar 0}$.

Now let $\fG = {\bf A} ( 1 , 1 )$. For any $X \in \fsp ( \Delta )$, let $X_+ = ( X , 0 )$ and $X_- = ( 0 , X )$ as elements of $\fG_{\bar 0}$. For any $\psi , \chi \in \Delta$, let $( \psi \otimes \chi )_+ = ( \psi \otimes \chi , 0 )$ and $( \psi \otimes \chi )_- = ( 0 , \psi \otimes \chi )$ as elements of $\fG_{\bar 1}$. Finally, let $\varepsilon$ denote a nondegenerate skewsymmetric $\fsp ( \Delta )$-invariant bilinear form on $\Delta$.

In terms of this notation, the brackets for $\fG$ that are not identically zero are given by
\begin{align}
[ X_\pm , Y_\pm ] &= [ X , Y ]_\pm~, \label{eq:A1100} \\
[ X_+ , ( \psi \otimes \chi )_\pm ] &= ( X \psi \otimes \chi )_\pm~, \label{eq:A1101+} \\
[ X_- , ( \psi \otimes \chi )_\pm ] &= ( \psi \otimes X \chi )_\pm~, \label{eq:A1101-} \\
[ ( \psi_1 \otimes \chi_1 )_+ , ( \psi_2 \otimes \chi_2 )_- ] &= \varepsilon ( \chi_1 , \chi_2 ) ( \psi_1 \otimes \psi_2^\flat + \psi_2 \otimes \psi_1^\flat )_+ - \varepsilon ( \psi_1 , \psi_2 ) ( \chi_1 \otimes \chi_2^\flat + \chi_2 \otimes \chi_1^\flat )_-~, \label{eq:A1111}
\end{align}
for all $X,Y \in \fsp ( \Delta )$ and $\psi , \chi , \psi_1 , \chi_1 , \psi_2 , \chi_2 \in \Delta$.

The condition \eqref{eq:BalancedSkew} is therefore satisfied since 
\begin{align} \label{eq:BalancedSkewA11}
[ ( \psi_1 \otimes \chi_1 )_+ , ( \psi_2 \otimes \chi_2 )_- ] &= \varepsilon ( \chi_1 , \chi_2 ) ( \psi_1 \otimes \psi_2^\flat + \psi_2 \otimes \psi_1^\flat )_+ - \varepsilon ( \psi_1 , \psi_2 ) ( \chi_1 \otimes \chi_2^\flat + \chi_2 \otimes \chi_1^\flat )_- \nonumber \\
&= - \varepsilon ( \chi_2 , \chi_1 ) ( \psi_2 \otimes \psi_1^\flat + \psi_1 \otimes \psi_2^\flat )_+ + \varepsilon ( \psi_2 , \psi_1 ) ( \chi_2 \otimes \chi_1^\flat + \chi_1 \otimes \chi_2^\flat )_- \nonumber \\
&= - [ ( \psi_2 \otimes \chi_2 )_+ , ( \psi_1 \otimes \chi_1 )_- ]~,
\end{align}
for all $\psi_1 , \chi_1 , \psi_2 , \chi_2 \in \Delta$.
\end{proof}

\begin{cor} \label{cor:ClassicalGn}
Any classical Lie superalgebra of type ${\bf A}$, ${\bf C}$ or ${\bf P}$ can play the role of $\fG (2)$ but only ${\bf A} ( 1 , 1 )$ can play the role of $\fG (3)$. 
\end{cor}

Proposition~\ref{prop:GSimple} implies that every classical Lie superalgebra $\fG$ has $[ [ \fG_{\bar 1} , \fG_{\bar 1} ] , \fG_{\bar 1} ] \neq 0$. But Theorem~\ref{thm:111Jacobi} implies that every $n$-admissible embedding superalgebra $\fG (n)$ must have $[ [ \fG (n)_{\bar 1} , \fG (n)_{\bar 1} ] , \fG (n)_{\bar 1} ] = 0$ if $n>3$. Therefore no classical Lie superalgebra $\fG$ can play the role of $\fG (n)$ if $n>3$.

The following theorem will allow us to determine which classical Lie superalgebras are $n$-admissible. 

\begin{thm} \label{thm:ClassicalLSA1Admissible}
If $\fG$ is a classical Lie superalgebra and $\fG_{\bar 1}$ admits a nondegenerate $\fG_{\bar 0}$-invariant symmetric bilinear form $b$ then either
\begin{itemize}
\item $\fG = {\bf A} ( N_+ -1 , N_- - 1 )$ with
\begin{equation}\label{eq:AN+N-BilinearForm}
b ( w^1_+ \otimes \alpha^1_- + \alpha^1_+ \otimes w^1_- , w^2_+ \otimes \alpha^2_- + \alpha^2_+ \otimes w^2_- ) = \lambda ( \alpha^1_+ ( w^2_+ ) \alpha^2_- ( w^1_- ) + \alpha^2_+ ( w^1_+ ) \alpha^1_- ( w^2_- ) )~,
\end{equation}
for all $w^1_+ \otimes \alpha^1_- + \alpha^1_+ \otimes w^1_- , w^2_+ \otimes \alpha^2_- + \alpha^2_+ \otimes w^2_- \in \fG_{\bar 1}$, in terms of some non-zero $\lambda \in \CC$. \\

\item $\fG = {\bf A} ( N -1 , N - 1 )$ with
\begin{equation}\label{eq:ANNBilinearForm}
b ( w^1 \otimes \alpha^1 + \beta^1 \otimes u^1 , w^2 \otimes \alpha^2 + \beta^2 \otimes u^2 ) = \lambda ( \alpha^1 ( u^2 ) \beta^2 ( w^1 ) +  \alpha^2 ( u^1 ) \beta^1 ( w^2 ) )~,
\end{equation}
for all $w^1 \otimes \alpha^1 + \beta^1 \otimes u^1 , w^2 \otimes \alpha^2 + \beta^2 \otimes u^2 \in \fG_{\bar 1}$, in terms of some non-zero $\lambda \in \CC$. \\

\item $\fG = {\bf A} ( 1 , 1 )$ with
\begin{align}\label{eq:A11BilinearForm}
&b ( ( \psi_1 \otimes \chi_1 )_+ + ( \phi_1 \otimes \upsilon_1 )_- , ( \psi_2 \otimes \chi_2 )_+ + ( \phi_2 \otimes \upsilon_2 )_- ) \nonumber \\
&= \lambda_+ \varepsilon ( \psi_1 , \psi_2 ) \varepsilon ( \chi_1 , \chi_2 ) +  \lambda_0 ( \varepsilon ( \psi_1 , \phi_2 ) \varepsilon ( \chi_1 , \upsilon_2 ) + \varepsilon ( \psi_2 , \phi_1 ) \varepsilon ( \chi_2 , \upsilon_1 ) ) + \lambda_- \varepsilon ( \phi_1 , \phi_2 ) \varepsilon ( \upsilon_1 , \upsilon_2 )~,
\end{align}
for all $( \psi_1 \otimes \chi_1 )_+ + ( \phi_1 \otimes \upsilon_1 )_- , ( \psi_2 \otimes \chi_2 )_+ + ( \phi_2 \otimes \upsilon_2 )_- \in \fG_{\bar 1}$, in terms of some $\lambda_+ , \lambda_0 , \lambda_-  \in \CC$ with $\lambda_+ \lambda_- \neq \lambda_0^2$, where $\varepsilon$ is a nondegenerate skewsymmetric $\fsp ( \Delta )$-invariant bilinear form on $\Delta$. \\

\item $\fG = {\bf C} ( \frac{N}{2} +1 )$ with
\begin{equation}\label{eq:CNBilinearForm}
b ( w^1_+ + u^1_- , w^2_+ + u^2_- ) =  \lambda ( \omega ( w^1 , u^2 ) + \omega ( w^2 , u^1 ) )~,
\end{equation}
for all $w^1_+ + u^1_- , w^2_+ + u^2_- \in \fG_{\bar 1}$, in terms of some non-zero $\lambda \in \CC$, where $\omega$ is a nondegenerate skewsymmetric $\fsp ( W )$-invariant bilinear form on $W$. \\

\item $\fG = {\bf Q} ( N-1 )$ with
\begin{equation}\label{eq:QNBilinearForm}
b ( X,Y ) =  \lambda \tr (XY)~,
\end{equation}
for all $X,Y \in \fG_{\bar 1}$, in terms of some non-zero $\lambda \in \CC$, where $\tr$ denotes the trace in $\End W$.
\end{itemize}
\end{thm}

\begin{proof}
If $b$ is a nondegenerate $\fsl ( W_+ ) \oplus \fsl ( W_- ) \oplus \fz$-invariant bilinear form on $W_+ \otimes W_-^* \otimes \zeta_{( N_+ - N_- )} \oplus W_+^* \otimes W_- \otimes \zeta_{( N_- - N_+ )}$ then, for any $w_\pm , u_\pm \in W_\pm$ and $\alpha_\pm , \beta_\pm \in W_\pm^*$, $b ( - \otimes \alpha_-  + - \otimes w_-  , - \otimes \beta_- + - \otimes u_- )$ defines an $\fsl ( W_+ )$-invariant bilinear form on $W_+ \oplus W_+^*$ and $b ( w_+ \otimes -  + \alpha_+ \otimes - , u_+ \otimes - + \beta_+ \otimes - )$ defines an $\fsl ( W_- )$-invariant bilinear form on $W_-^* \oplus W_-$. Since $N_+ >2$, any $\fsl ( W_+ )$-invariant bilinear form on $W_+$ or $W_+^*$ is zero and any $\fsl ( W_+ )$-invariant bilinear map $W_+ \times W_+^* \rightarrow \CC$ is a scalar multiple of the dual pairing between $W_+$ and $W_+^*$ (see Appendix~\ref{sec:WPlusW*}). If $b$ is symmetric then this means
\begin{equation}\label{eq:AN+N-BilinearFormHalf}
b ( w^1_+ \otimes \alpha^1_- + \alpha^1_+ \otimes w^1_- , w^2_+ \otimes \alpha^2_- + \alpha^2_+ \otimes w^2_- ) =  \alpha^1_+ ( w^2_+ ) \mu ( w^1_- , \alpha^2_- ) + \alpha^2_+ ( w^1_+ ) \mu ( w^2_- , \alpha^1_- )~,
\end{equation}
for all $w^1_+ \otimes \alpha^1_- + \alpha^1_+ \otimes w^1_- , w^2_+ \otimes \alpha^2_- + \alpha^2_+ \otimes w^2_- \in W_+ \otimes W_-^* \otimes \zeta_{( N_+ - N_- )} \oplus W_+^* \otimes W_- \otimes \zeta_{( N_- - N_+ )}$, in terms of some $\fsl ( W_- )$-invariant bilinear map $\mu : W_- \times W_-^* \rightarrow \CC$. Any such $\mu$ must be a non-zero scalar multiple of the dual pairing between $W_-$ and $W_-^*$. For $N_- >2$, this is proven in Appendix~\ref{sec:WPlusW*}. For $N_- = 2$, any non-zero skewsymmetric bilinear form $\omega$ on $W_-$ is $\fsl ( W_- )$-invariant. Proposition~\ref{prop:gVbExistence} therefore implies $W_- \cong W_-^*$ as representations of $\fsl ( W_- )$. Furthermore, $\mu ( - , -^\flat )$ defines an $\fsl ( W_- )$-invariant bilinear form on $W_-$ which must be a non-zero scalar multiple of $\omega$, i.e.
\begin{equation}\label{eq:AN+N-BilinearFormN-2}
\mu ( w_- , u_-^\flat ) = \lambda \omega ( w_- , u_- )~,
\end{equation}
for all $w_- , u_- \in W_-$, in terms of some non-zero $\lambda \in \CC$. But \eqref{eq:AN+N-BilinearFormN-2} is equivalent to 
\begin{equation}\label{eq:AN+N-BilinearFormN-2Alpha}
\mu ( w_- , \alpha_- ) = \lambda \omega ( w_- , \alpha_-^\sharp ) = - \lambda \omega ( \alpha_-^\sharp , w_- ) = - \lambda \alpha_- ( w_- )~,
\end{equation}
for all $w_- \in W_-$ and $\alpha_- \in W_-^*$, using the canonical isomorphisms in \eqref{eq:FlatSharp}. Thus $\mu$ is a non-zero scalar multiple of the dual pairing between $W_-$ and $W_-^*$ if $N_- = 2$ as well. For $N_- =1$, this result is trivial. One then recovers \eqref{eq:AN+N-BilinearForm} by substituting this back into \eqref{eq:AN+N-BilinearFormHalf}. It follows that \eqref{eq:AN+N-BilinearForm} is automatically $\fz$-invariant since $b$ only pairs the factors in $\zeta_{( N_+ - N_- )}$ and $\zeta_{( N_- - N_+ )}$ with each other and not with themselves. It is clear by inspection that \eqref{eq:AN+N-BilinearForm} is nondegenerate. This rules in ${\bf A} ( N_+ -1 , N_- - 1 )$ with $b$ as in \eqref{eq:AN+N-BilinearForm}. The same logic with $W_\pm = W$ and $N_\pm = N >2$ also rules in ${\bf A} ( N -1 , N - 1 )$ with $b$ as in \eqref{eq:ANNBilinearForm}.

If $b$ is a nondegenerate $\fsp ( \Delta_1 ) \oplus \fsp ( \Delta_2 )$-invariant bilinear form on $\Delta_1 \otimes \Delta_2 \oplus \Delta_1 \otimes \Delta_2$ then, for any $\psi , \chi \in \Delta$, $b ( ( - \otimes \psi )_\pm , ( - \otimes \chi )_\pm )$, $b ( ( - \otimes \psi )_\pm , ( - \otimes \chi )_\mp )$,  $b ( ( \psi \otimes - )_\pm , ( \chi \otimes - )_\pm )$ and $b ( ( \psi \otimes - )_\pm , ( \chi \otimes - )_\mp )$ define $\fsp ( \Delta )$-invariant bilinear forms on $\Delta$. Proposition~\ref{prop:gVbExistence} implies that (up to scaling) $\Delta$ admits a unique non-zero $\fsp ( \Delta )$-invariant bilinear form $\varepsilon$ that is nondegenerate and skewsymmetric. Therefore
\begin{align}\label{eq:A11BilinearFormHalf}
&b ( ( \psi_1 \otimes \chi_1 )_+ + ( \phi_1 \otimes \upsilon_1 )_- , ( \psi_2 \otimes \chi_2 )_+ + ( \phi_2 \otimes \upsilon_2 )_- ) \nonumber \\
&= \lambda_+ \varepsilon ( \psi_1 , \psi_2 ) \varepsilon ( \chi_1 , \chi_2 ) +  \lambda_0 \varepsilon ( \psi_1 , \phi_2 ) \varepsilon ( \chi_1 , \upsilon_2 ) + \lambda_0^\prime \varepsilon ( \phi_1 , \psi_2 ) \varepsilon ( \upsilon_1 , \chi_2 )  + \lambda_- \varepsilon ( \phi_1 , \phi_2 ) \varepsilon ( \upsilon_1 , \upsilon_2 )~,
\end{align}
for all $( \psi_1 \otimes \chi_1 )_+ + ( \phi_1 \otimes \upsilon_1 )_- , ( \psi_2 \otimes \chi_2 )_+ + ( \phi_2 \otimes \upsilon_2 )_- \in \Delta_1 \otimes \Delta_2 \oplus \Delta_1 \otimes \Delta_2$, in terms of some $\lambda_+ , \lambda_0 , \lambda_0^\prime , \lambda_- \in \CC$. If $b$ is symmetric then $\lambda_0 = \lambda_0^\prime$ and \eqref{eq:A11BilinearFormHalf} becomes \eqref{eq:A11BilinearForm}. If \eqref{eq:A11BilinearForm} is zero for all $( \psi_2 \otimes \chi_2 )_+ + ( \phi_2 \otimes \upsilon_2 )_- \in \Delta_1 \otimes \Delta_2 \oplus \Delta_1 \otimes \Delta_2$ then nondegeneracy of $\varepsilon$ implies
\begin{equation}\label{eq:A11BilinearFormNondegenerate1}
\lambda_+ \psi_1 \otimes \chi_1 + \lambda_0 \phi_1 \otimes \upsilon_1 = 0
\end{equation}
and 
\begin{equation}\label{eq:A11BilinearFormNondegenerate2}
\lambda_0 \psi_1 \otimes \chi_1 + \lambda_- \phi_1 \otimes \upsilon_1 = 0~.
\end{equation}
Subtracting $\lambda_0$ times \eqref{eq:A11BilinearFormNondegenerate2} from $\lambda_-$ times \eqref{eq:A11BilinearFormNondegenerate1} and $\lambda_0$ times \eqref{eq:A11BilinearFormNondegenerate1} from $\lambda_+$ times \eqref{eq:A11BilinearFormNondegenerate2} implies $( \psi_1 \otimes \chi_1 )_+ + ( \phi_1 \otimes \upsilon_1 )_- = 0$ unless $\lambda_+ \lambda_- = \lambda_0^2$. If $\lambda_+ \lambda_- = \lambda_0^2$ and $\lambda_0 = 0$ then either $\psi_1 \otimes \chi_1$ or $\phi_1 \otimes \upsilon_1$ is zero (but not necessarily both). If $\lambda_+ \lambda_- = \lambda_0^2$ and $\lambda_0 \neq 0$ then $\psi_1 \otimes \chi_1$ and $\phi_1 \otimes \upsilon_1$ are collinear. Therefore \eqref{eq:A11BilinearForm} is nondegenerate if and only if $\lambda_+ \lambda_- \neq \lambda_0^2$. This rules in ${\bf A} ( 1 , 1 )$ with $b$ as in \eqref{eq:A11BilinearForm}. 

If $b$ is a nondegenerate $\fso ( W_+ ) \oplus \fsp ( W_- )$-invariant bilinear form on $W_+ \otimes W_-$ then, for any $w_\pm , u_\pm \in W_\pm$, $b ( - \otimes w_- , - \otimes u_- )$ defines an $\fso ( W_+ )$-invariant bilinear form on $W_+$ and $b ( w_+ \otimes - , u_+ \otimes - )$ defines an $\fsp ( W_- )$-invariant bilinear form on $W_-$. Proposition~\ref{prop:gVbExistence} implies that (up to scaling) $W_+$ admits a unique non-zero $\fso ( W_+ )$-invariant bilinear form that is nondegenerate and symmetric while $W_-$ admits a unique non-zero $\fsp ( W_- )$-invariant bilinear form that is nondegenerate and skewsymmetric. Therefore $b$ must be skewsymmetric since it is (up to scaling) a tensor product of these two bilinear forms. This rules out $\fosp ( W_+ , W_- )$.

If $b$ is a nondegenerate $\fz \oplus \fsp ( W )$-invariant bilinear form on $\zeta_{(1)} \otimes W \oplus \zeta_{(-1)} \otimes W$ then $\fz$-invariance implies $b ( \zeta_{( \pm 1 )} \otimes W ,  \zeta_{( \pm 1 )} \otimes W ) = 0$. The remaining components $b ( \zeta_{( \pm 1 )} \otimes W ,  \zeta_{( \mp 1 )} \otimes W )$ are automatically $\fz$-invariant and define $\fsp (W)$-invariant bilinear forms on $W$. Proposition~\ref{prop:gVbExistence} implies that both these components of $b$ must be non-zero scalar multiples of a nondegenerate $\fsp (W)$-invariant skewsymmetric bilinear form $\omega$ on $W$. If $b$ is symmetric then the relative values of the two scalars are fixed. This rules in ${\bf C} ( \frac{N}{2} +1 )$ with $b$ as in \eqref{eq:CNBilinearForm}. It is clear by inspection that \eqref{eq:CNBilinearForm} is nondegenerate.

If $b$ is a nondegenerate $\fsp ( \Delta_1 ) \oplus \fsp ( \Delta_2 ) \oplus \fsp ( \Delta_3 )$-invariant bilinear form on $\Delta_1 \otimes \Delta_2 \otimes \Delta_3$ then, for any $\psi , \chi , \phi , \upsilon \in \Delta$, $b ( - \otimes \psi \otimes \chi , - \otimes \phi \otimes \upsilon )$, $b ( \psi \otimes - \otimes \chi , \phi \otimes - \otimes \upsilon )$ and $b ( \psi \otimes \chi \otimes - , \phi \otimes \upsilon \otimes - )$ define $\fsp ( \Delta )$-invariant bilinear forms on each factor of $\Delta$ in $\Delta_1 \otimes \Delta_2 \otimes \Delta_3$. Proposition~\ref{prop:gVbExistence} implies that (up to scaling) $\Delta$ admits a unique non-zero $\fsp ( \Delta )$-invariant bilinear form that is nondegenerate and skewsymmetric. Therefore $b$ must be skewsymmetric since it is (up to scaling) a tensor product of three copies of this bilinear form. This rules out ${\bf D} ( 2,1; \alpha )$.

If $b$ is a nondegenerate $\fso ( W ) \oplus \fsp ( \Delta )$-invariant bilinear form on $\cS^{(7)} \otimes \Delta$ then, for any $\psi , \chi \in \cS^{(7)}$ and $\phi , \upsilon \in \Delta$, $b ( - \otimes \phi , - \otimes \upsilon )$ defines an $\fso ( W )$-invariant bilinear form on $\cS^{(7)}$ and $b ( \psi \otimes - , \chi \otimes - )$ defines an $\fsp ( \Delta )$-invariant bilinear form on $\Delta$. Theorem~\ref{thm:AdmissibleSpinorialBilinearForm} and Table~\ref{tab:SigmaTauSigns} imply that $\cS^{(7)}$ admits a nondegenerate $\fso ( W )$-invariant symmetric bilinear form which, according to Proposition~\ref{prop:gVbExistence}, is unique (up to scaling). Therefore $b$ must be skewsymmetric since it is (up to scaling) a tensor product of this symmetric bilinear form on $\cS^{(7)}$ and a nondegenerate $\fsp ( \Delta )$-invariant skewsymmetric bilinear form on $\Delta$. This rules out ${\bf F} ( 4 )$.

If $b$ is a nondegenerate $\fg_2 \oplus \fsp ( \Delta )$-invariant bilinear form on $W \otimes \Delta$ then, for any $w,u \in W$ and $\psi , \chi \in \Delta$, $b ( - \otimes \psi , - \otimes \chi )$ defines a $\fg_2$-invariant bilinear form on $W$ and $b ( w \otimes - , u \otimes - )$ defines an $\fsp ( \Delta )$-invariant bilinear form on $\Delta$. Proposition~\ref{prop:gVbExistence} implies that any non-zero $\fg_2$-invariant bilinear form on $W$ is a non-zero scalar multiple of a nondegenerate $\fso (W)$-invariant symmetric bilinear form on $W$ (since $\fg_2$ is a Lie subalgebra of $\fso (W)$). Therefore $b$ must be skewsymmetric since it is (up to scaling) a tensor product of any such symmetric bilinear form on $W$ and a nondegenerate $\fsp ( \Delta )$-invariant skewsymmetric bilinear form on $\Delta$. This rules out ${\bf G} ( 3 )$. 

If $N>2$ then $S^2 W \oplus \Lambda^2 W^*$ does not admit a nondegenerate $\fsl ( W )$-invariant bilinear form (see Appendix~\ref{sec:S2WPlusWedge2W*}). This rules out ${\bf P} ( N -1 )$.   

Let $\kappa$ be the bilinear form on $\fsl (W)$ defined such that $\kappa (X,Y) = \tr ( XY )$, for all $X , Y \in \fsl (W)$. Clearly $\kappa$ is non-zero, symmetric and invariant under the adjoint action of $\fsl (W)$. Proposition~\ref{prop:gVbExistence} therefore implies that any nondegenerate $\fsl (W)$-invariant symmetric bilinear form on $\fsl (W)$ must be a non-zero scalar multiple of $\kappa$. This rules in ${\bf Q} ( N-1 )$ with $b$ as in \eqref{eq:QNBilinearForm}.
\end{proof}

\begin{cor} \label{cor:nAdmissibleClassical}
A classical Lie superalgebra is $1$-admissible if and only if it is of type ${\bf A}$, ${\bf C}$ or ${\bf Q}$ and the only $n$-admissible classical Lie superalgebra with $n > 1$ is ${\bf A} ( 1 , 1 )$ for $n=2$ and $n=3$.
\end{cor}

\begin{proof}
The conditions for $1$-admissibility follow directly from Theorem~\ref{thm:ClassicalLSA1Admissible} and we have already established that no classical Lie superalgebra can be $n$-admissible if $n>3$.

The only $3$-graded Lie superalgebras in Theorem~\ref{thm:ClassicalLSA1Admissible} are of type ${\bf A}$ and ${\bf C}$, and none of them except $\fG = {\bf A} ( 1 , 1 )$ admits a nondegenerate $\fG_{\bar 0}$-invariant symmetric bilinear form on both $\fG_{\bar 1}^+$ and $\fG_{\bar 1}^-$. More precisely, any nondegenerate $\fsp ( \Delta_1 ) \oplus \fsp ( \Delta_2 )$-invariant symmetric bilinear form on $\fG_{\bar 1}^\pm = \Delta_1 \otimes \Delta_2$ is of the form $\lambda_\pm \varepsilon \otimes \varepsilon$, for some non-zero $\lambda_\pm \in \CC$, where $\varepsilon$ is a nondegenerate $\fsp ( \Delta )$-invariant skewsymmetric bilinear form on $\Delta$. Therefore only ${\bf A} ( 1 , 1 )$ is $2$-admissible.

Proposition~\ref{prop:ClassicalLSABalanced3Graded} implies that the only balanced $3$-graded Lie superalgebra in Theorem~\ref{thm:ClassicalLSA1Admissible} is ${\bf A} ( 1 , 1 )$ and, as above, any nondegenerate $\fsp ( \Delta_1 ) \oplus \fsp ( \Delta_2 )$-invariant symmetric bilinear form on $\Delta_1 \otimes \Delta_2$ is of the form $\lambda \varepsilon \otimes \varepsilon$, for some non-zero $\lambda \in \CC$. Therefore only ${\bf A} ( 1 , 1 )$ is $3$-admissible. 
\end{proof}

\begin{prop} \label{prop:NontrivialClassicalLSAAdmissible}
The Poincar\'{e} superalgebra $\fP (n)$ defined by any $n$-admissible classical Lie superalgebra $\fG (n)$ is non-trivial if and only if $\fG (n)_{\bar 0}$ is semisimple.
\end{prop}

\begin{proof}
If $n=1$ and $\fG (1)_{\bar 0}$ is semisimple then Corollary~\ref{cor:nAdmissibleClassical} implies $\fG (1)$ is either ${\bf A} ( N -1 , N - 1 )$ with $N \geq 2$ or ${\bf Q} ( N -1 )$. In this case Corollary~\ref{cor:ReductiveEquivalentCentralExtensionPbar} implies $\fP (1)$ is non-trivial. 

If $n=1$ and $\fG (1)_{\bar 0}$ is not semisimple then Corollary~\ref{cor:nAdmissibleClassical} implies $\fG (1)$ is either ${\bf A} ( N_+ -1 , N_- - 1 )$ or ${\bf C} ( \frac{N}{2} +1 )$. In this case the centre of $\fG (1)_{\bar 0}$ is $\fz$, so Lemma~\ref{lem:EquivalentCentralExtensionPbar} and Corollary~\ref{cor:ReductiveEquivalentCentralExtensionPbar} imply $\fP (1)$ is non-trivial unless $b$ is a scalar multiple of the $\fz$-part of the $[ {\bar 1}{\bar 1} ]$ bracket for $\fG (1)$. But the $[ {\bar 1}{\bar 1} ]$ bracket for $\fG (1)$ defines a $\fG (1)_{\bar 0}$-equivariant symmetric bilinear map $\fG (1)_{\bar 1} \times \fG (1)_{\bar 1} \rightarrow \fG (1)_{\bar 0}$, so its $\fz$-part must define a $\fG (1)_{\bar 0}$-invariant symmetric bilinear form on $\fG (1)_{\bar 1}$ (since $[ \fG (1)_{\bar 0} , [ \fG (1)_{\bar 1} , \fG (1)_{\bar 1} ] ]$ does not contain $\fz$). Theorem~\ref{thm:ClassicalLSA1Admissible} therefore implies that the $\fz$-part of the $[ {\bar 1}{\bar 1} ]$ bracket for $\fG (1)$ must be a non-zero scalar multiple of $b$ unless it is identically zero. But if it were identically zero then removing $\fz$ from $\fG (1)$ would define a proper ideal of $\fG (1)$ which cannot exist since $\fG (1)$ is simple. Therefore $\fP (1)$ is trivial.

If $n>1$ then Corollary~\ref{cor:nAdmissibleClassical} implies that only ${\bf A} ( 1 , 1 )$ is $n$-admissible for $n=2$ and $n=3$. In this case ${\bf A} ( 1 , 1 )_{\bar 0} = \fsp ( \Delta_1 ) \oplus \fsp ( \Delta_2 )$ is semisimple and Lemma~\ref{lem:EquivalentAbeilianExtensionPbar} implies $\fP (n)$ is non-trivial. 
\end{proof}

%%%%%%%%%%%%%%%%%%%%%%

\section{Admissible Lie superalgebras from triple systems}
\label{sec:EmbeddingLSAFromTripleSystems}

In this section we will show how to construct an embedding superalgebra with $n<4$ from a certain type of triple system with derivations. Furthermore, we will demonstrate that the embedding superalgebra is $n$-admissible if and only if the corresponding triple system with derivations admits a certain nondegenerate derivation-invariant symmetric bilinear form.

%%%%%%%%%%%%%%%%%%%%%%

\subsection{Triple systems}
\label{sec:TripleSystems}

In general, a vector space $V$ is called a {\emph{triple system}} (or {\emph{3-algebra}}) if it is equipped with a trilinear map (or {\emph{3-bracket}})
\begin{equation}\label{eq:3Bracket}
[-,-,-] : V \times V \times V \longrightarrow V~.
\end{equation}

Let us define 
\begin{equation}\label{eq:3Derivation}
\Der^\pm V = \{ X \in \End V~|~X [ v , w , u ] = [ Xv , w , u ] \pm [ v , Xw , u ] + [ v , w , Xu ] ,~v,w,u \in V \}~.
\end{equation}
Any element in $\Der^+ V$ is called a {\emph{derivation}} of $V$ while any element in $\Der^- V$ is called an {\emph{anti-derivation}} of $V$. It is easily verified that $[X,Y] \in \Der^+ V$ if $X,Y \in \Der^\pm V$ while $[X,Y] \in \Der^- V$ if $X \in \Der^\pm V$ and $Y \in \Der^\mp V$, where brackets denote commutators of endomorphisms. It follows that $\Der^+ V$ forms a Lie subalgebra of $\fgl (V)$. 

For any $v,w \in V$, let us define $d(v,w) \in \End V$ such that
\begin{equation}\label{eq:dVV}
d(v,w) u = [v,w,u]~,
\end{equation}
for all $u \in V$. When written in terms of $d$, the definition of $X \in \Der^\pm V$ in \eqref{eq:3Derivation} is equivalent to
\begin{equation}\label{eq:3dDerivation}
[ X , d(v,w) ] = d ( Xv , w ) \pm d ( v , Xw )~,
\end{equation}
for all $v,w \in V$. 

%%%%%%%%%%%%%%%%%%%%%%

\subsection{$1$-admissible Lie superalgebras from anti-Lie triple systems}
\label{sec:AntiLieTripleSystems}

A triple system $V$ is called an {\emph{anti-Lie triple system}} if 
\begin{equation}\label{eq:ALTSsym}
[ v , w , u ] = [ w , v , u ]~,
\end{equation}
for all $v,w,u \in V$, 
\begin{equation}\label{eq:ALTScyc}
[ v , w , u ] + [ w , u , v ] + [ u , v , w ] = 0~,
\end{equation}
for all $v,w,u \in V$, and
\begin{equation}\label{eq:ALTSend}
[ v_1 , v_2 , [ w_1 , w_2 , w_3 ] ] = [ [ v_1 , v_2 , w_1 ] , w_2 , w_3 ] + [ w_1 , [ v_1 , v_2 , w_2 ] , w_3 ] + [ w_1 , w_2 , [ v_1 , v_2 , w_3 ] ]~,
\end{equation}
for all $v_1 , v_2 , w_1 , w_2 , w_3 \in V$. 

The condition \eqref{eq:ALTSend} simply means $d(V,V) \subset \Der^+ V$. In fact, \eqref{eq:3dDerivation} implies that $d(V,V)$ is an ideal of $\Der^+ V$. Any Lie subalgebra $D(V)$ of $\Der^+ V$ which contains $d(V,V)$ will be referred to as a {\emph{derivation algebra}} of the anti-Lie triple system $V$.

We will now summarise the relationship between anti-Lie triple systems with a derivation algebra and Lie superalgebras that was first explored in \cite{FaulknerFerrar}.

\begin{prop} \label{prop:ALTS-LSA}
Any anti-Lie triple system $V$ with derivation algebra $D(V)$ defines a Lie superalgebra $\fG$ with $\fG_{\bar 0} = D(V)$, $\fG_{\bar 1} = V$ and brackets
\begin{align}
[X,Y] &= XY - YX~, \label{eq:ALTSLSAbracket00} \\
[X,v] &= Xv~, \label{eq:ALTSLSAbracket01} \\
[v,w] &= d(v,w)~, \label{eq:ALTSLSAbracket11}
\end{align}
for all $X,Y \in \fG_{\bar 0}$ and $v,w \in \fG_{\bar 1}$.
\end{prop}

\begin{proof}
Clearly the $[{\bar 0}{\bar 0}]$ bracket in \eqref{eq:ALTSLSAbracket00} is skewsymmetric while the $[{\bar 1}{\bar 1}]$ bracket in \eqref{eq:ALTSLSAbracket11} is symmetric due to \eqref{eq:ALTSsym}. The $[ {\bar 0}{\bar 0}{\bar 0} ]$ Jacobi identity for $\fG$ is satisfied since $D(V)$ is a Lie algebra. The $[ {\bar 0}{\bar 0}{\bar 1} ]$ Jacobi identity for $\fG$ is satisfied since $V$ is a representation of $D(V)$. The $[ {\bar 0}{\bar 1}{\bar 1} ]$ Jacobi identity for $\fG$ is satisfied due to \eqref{eq:3dDerivation}. The $[ {\bar 1}{\bar 1}{\bar 1} ]$ Jacobi identity for $\fG$ is satisfied due to \eqref{eq:ALTScyc}.
\end{proof}

\begin{prop} \label{prop:LSA-ALTS}
Any Lie superalgebra $\fG$ defines an anti-Lie triple system $\fG_{\bar 1}$ with 3-bracket
\begin{equation}\label{eq:LSAALTS3bracket}
[ v , w , u ] = [ [ v , w ] , u ]~,
\end{equation}
for all $v,w,u \in \fG_{\bar 1}$, and a derivation algebra $\ad ( \fG_{\bar 0} ) |_{\fG_{\bar 1}}$ of $\fG_{\bar 1}$. 
\end{prop}

\begin{proof}
The 3-bracket \eqref{eq:LSAALTS3bracket} obeys \eqref{eq:ALTSsym} since the $[ {\bar 1}{\bar 1} ]$ bracket for $\fG$ is symmetric. The condition \eqref{eq:ALTScyc} is precisely the $[ {\bar 1}{\bar 1}{\bar 1} ]$ Jacobi identity for $\fG$. Any $X \in \fG_{\bar 0}$ defines $\rho (X) \in \ad ( \fG_{\bar 0} ) |_{\fG_{\bar 1}}$ such that 
\begin{equation}\label{eq:LSAALTSRho}
\rho (X) v = [X,v]~,
\end{equation}
for all $v \in \fG_{\bar 1}$. Therefore \eqref{eq:dVV} and \eqref{eq:LSAALTS3bracket} imply
\begin{equation}\label{eq:LSAALTSd}
d(v,w) = \rho ( [v,w] )~,
\end{equation}
for all $v,w \in \fG_{\bar 1}$. Acting with $\rho$ on the $[ {\bar 0}{\bar 1}{\bar 1} ]$ Jacobi identity for $\fG$ gives
\begin{equation}\label{eq:LSAALTS011}
\rho ( [ X , [ v , w ] ] ) \underset{\eqref{eq:LSAALTSRho}}{=} \rho ( [ \rho ( X ) v , w ] ) + \rho ( [ v , \rho ( X ) w ] ) \underset{\eqref{eq:LSAALTSd}}{=} d ( \rho ( X )  v , w ) + d ( v , \rho ( X ) w )~,
\end{equation}
for all $X \in \fG_{\bar 0}$ and $v,w \in \fG_{\bar 1}$. On the other hand, the $[ {\bar 0}{\bar 0}{\bar 1} ]$ Jacobi identity for $\fG$ says that $\rho : \fG_{\bar 0} \rightarrow \fgl ( \fG_{\bar 1} )$ is a Lie algebra homomorphism, so
\begin{equation}\label{eq:LSAALTS001}
\rho ( [ X , [ v , w ] ] ) =  [ \rho (X) , \rho ( [v,w] ) ] \underset{\eqref{eq:LSAALTSd}}{=} [ \rho (X) , d (v,w) ]~,
\end{equation}
for all $X \in \fG_{\bar 0}$ and $v,w \in \fG_{\bar 1}$. Combining \eqref{eq:LSAALTS011} and  \eqref{eq:LSAALTS001} implies that any $\rho (X) \in \ad ( \fG_{\bar 0} ) |_{\fG_{\bar 1}}$ is a derivation of the triple system $\fG_{\bar 1}$. Furthermore, \eqref{eq:LSAALTSd} implies $d ( \fG_{\bar 1} , \fG_{\bar 1} ) = \ad ( [ \fG_{\bar 1} , \fG_{\bar 1} ] ) |_{\fG_{\bar 1}} \subset  \ad ( \fG_{\bar 0} ) |_{\fG_{\bar 1}}$. Therefore $\fG_{\bar 1}$ is an anti-Lie triple system with derivation algebra $\ad ( \fG_{\bar 0} ) |_{\fG_{\bar 1}}$.
\end{proof}

\begin{cor} \label{cor:LSAALTSbijection}
Up to isomorphism, the maps $( V , D(V) ) \rightarrow D(V) \oplus V$ and $\fG \rightarrow ( \fG_{\bar 1} , \ad ( \fG_{\bar 0} ) |_{\fG_{\bar 1}} )$ defined by Propositions~\ref{prop:ALTS-LSA} and \ref{prop:LSA-ALTS} provide a bijection between anti-Lie triple systems with a derivation algebra and Lie superalgebras $\fG$ with $\fG_{\bar 0}$ acting faithfully on $\fG_{\bar 1}$.
\end{cor}

\begin{proof}
Proposition~\ref{prop:ALTS-LSA} implies that any anti-Lie triple system $V$ with derivation algebra $D(V)$ defines a Lie superalgebra $\fG$ with $\fG_{\bar 1} = V$ and $\fG_{\bar 0} = D(V)$. If $X \in \fG_{\bar 0}$ then \eqref{eq:LSAALTSRho} and \eqref{eq:ALTSLSAbracket01} imply $\rho (X) = X$, so $\ad ( \fG_{\bar 0} ) |_{\fG_{\bar 1}} = \fG_{\bar 0}$. Proposition~\ref{prop:LSA-ALTS} then implies that this Lie superalgebra $\fG$ defines the original anti-Lie triple system $\fG_{\bar 1} = V$ with derivation algebra $\ad ( \fG_{\bar 0} ) |_{\fG_{\bar 1}} = \fG_{\bar 0} = D(V)$.

Conversely, Proposition~\ref{prop:LSA-ALTS} implies that any Lie superalgebra $\fG$ defines an anti-Lie triple system $\fG_{\bar 1}$ with derivation algebra $\ad ( \fG_{\bar 0} ) |_{\fG_{\bar 1}}$. Proposition~\ref{prop:ALTS-LSA} then implies that $( \fG_{\bar 1} , \ad ( \fG_{\bar 0} ) |_{\fG_{\bar 1}} )$ defines a Lie superalgebra with even part $\ad ( \fG_{\bar 0} ) |_{\fG_{\bar 1}}$ and odd part $ \fG_{\bar 1}$. This Lie superalgebra is isomorphic to $\fG$ if and only if $\ad ( \fG_{\bar 0} ) |_{\fG_{\bar 1}} \cong \fG_{\bar 0}$ as Lie algebras. Since $\ad ( \fG_{\bar 0} ) |_{\fG_{\bar 1}} \cong \fG_{\bar 0} / \ker \rho$, where $\rho : \fG_{\bar 0} \rightarrow \ad ( \fG_{\bar 0} ) |_{\fG_{\bar 1}}$ is the surjective Lie algebra homomorphism defining the adjoint action of $\fG_{\bar 0}$ on $\fG_{\bar 1}$, it follows that $\ad ( \fG_{\bar 0} ) |_{\fG_{\bar 1}} \cong \fG_{\bar 0}$ if and only if $\rho$ is injective, which is precisely what it means for $\fG_{\bar 0}$ to act faithfully on $\fG_{\bar 1}$.
\end{proof}

\begin{remark}
The requirement of faithfulness in Corollary~\ref{cor:LSAALTSbijection} does not represent much loss of generality since one can always quotient $\fG_{\bar 0}$ by the kernel of its adjoint action on $\fG_{\bar 1}$ to make the action of $\fG_{\bar 0}$ on $\fG_{\bar 1}$ faithful.
\end{remark}

It follows that, up to isomorphism, every embedding superalgebra $\fG (1)$ with $\fG (1)_{\bar 0}$ acting faithfully on $\fG (1)_{\bar 1}$ can be constructed from a unique anti-Lie triple system $V$ with derivation algebra $D(V)$. Clearly every nondegenerate $\fG (1)_{\bar 0}$-invariant symmetric bilinear form on $\fG (1)_{\bar 1}$ corresponds to a unique nondegenerate $D(V)$-invariant symmetric bilinear form on $V$, so we can use the existence of either to recognise when $\fG (1)$ is $1$-admissible.

%%%%%%%%%%%%%%%%%%%%%%

\subsubsection{Anti-Lie triple systems from anti-Jordan triple systems}
\label{sec:AntiLieTripleSystemsFromAntiJordanTripleSystems}

A triple system $V$ is called an {\emph{anti-Jordan triple system}} if 
\begin{equation}\label{eq:AJTSsym}
[ v , w , u ] = - [ u , w , v ]~,
\end{equation}
for all $v,w,u \in V$, and
\begin{equation}\label{eq:AJTSend}
[ v_1 , v_2 , [ w_1 , w_2 , w_3 ] ] = [ [ v_1 , v_2 , w_1 ] , w_2 , w_3 ] + [ w_1 , [ v_2 , v_1 , w_2 ] , w_3 ] + [ w_1 , w_2 , [ v_1 , v_2 , w_3 ] ]~,
\end{equation}
for all $v_1 , v_2 , w_1 , w_2 , w_3 \in V$. 

When written in terms of $d$, \eqref{eq:AJTSend} is equivalent to 
\begin{equation}\label{eq:AJTSendd}
[ d ( v_1 , v_2 ) , d ( w_1 , w_2 ) ] = d ( d ( v_1 , v_2 ) w_1 , w_2 ) + d ( w_1 , d ( v_2 , v_1 ) w_2 )~,
\end{equation}
for all $v_1 , v_2 , w_1 , w_2 \in V$. The reversed order of $v_1$ and $v_2$ in the second term on the right hand side in \eqref{eq:AJTSendd} means $d(V,V)$ is not necessarily contained in $\Der^\pm V$. However, if we define
\begin{equation}\label{eq:dPMVV}
d^\pm (v,w) = d(v,w) \pm d(w,v)~,
\end{equation}
for all $v,w \in V$, then clearly $d^\pm (V,V) \subset \Der^\pm V$. Any Lie subalgebra $D(V)$ of $\Der^+ V$ which contains $d^+ (V,V)$ will be referred to as a {\emph{derivation algebra}} of the anti-Jordan triple system $V$.

\begin{example} \label{ex:AJTSLA}
Let $V$ be a Lie algebra equipped with Lie bracket $[-,-]$ and an element $\phi \in V^*$ with $\phi ( [V,V] ) = 0$. Now define a 3-bracket on $V$ such that
\begin{equation}\label{eq:AJTSExLA}
[ v , w , u ] = \phi (w) [v,u]~,
\end{equation}
for all $v,w,u \in V$. Clearly \eqref{eq:AJTSsym} is satisfied since $[-,-]$ is skewsymmetric. 

From \eqref{eq:AJTSExLA} it follows that 
\begin{equation}\label{eq:AJTSExLAd}
d ( v,w ) = \phi (w) \ad (v)~,
\end{equation}
for all $v,w \in V$, where 
\begin{equation}\label{eq:LAad}
\ad (v) w = [v,w]~,
\end{equation}
for all $v,w \in V$. Therefore 
\begin{equation}\label{eq:AJTSExLAdTerm2}
d ( w_1 , d ( v_2 , v_1 ) w_2 ) \underset{\eqref{eq:AJTSExLAd}}{=} \phi ( d ( v_2 , v_1 ) w_2 ) \ad ( w_1 ) \underset{\eqref{eq:AJTSExLA}}{=} \phi ( v_1 ) \phi ( [ v_2 , w_2 ] ) \ad ( w_1 ) = 0~,
\end{equation}
for all $v_1 , v_2 , w_1 , w_2 \in V$, since $\phi ( [V,V] ) = 0$. It follows that the second term on the right hand side in \eqref{eq:AJTSendd} is identically zero, so $d (V,V) \subset \Der^+ V$ if  \eqref{eq:AJTSendd} is satisfied. This is indeed the case since 
\begin{align}\label{eq:AJTSExLAdTerms}
[ d ( v_1 , v_2 ) , d ( w_1 , w_2 ) ] &\underset{\eqref{eq:AJTSExLAd}}{=} \phi ( v_2 ) \phi ( w_2 ) [ \ad ( v_1 ) , \ad ( w_1 ) ] \nonumber \\
&\;\;\, = \;\;\,  \phi ( v_2 ) \phi ( w_2 ) \ad ( [ v_1 , w_1 ] ) \nonumber \\
&\underset{\eqref{eq:AJTSExLA}}{=} \phi ( w_2 ) \ad ( [ v_1 , v_2 , w_1 ] ) \nonumber \\
&\underset{\eqref{eq:AJTSExLAd}}{=} d ( d ( v_1 , v_2 ) w_1 , w_2 )~,
\end{align}
for all $v_1 , v_2 , w_1 , w_2 \in V$. Therefore $V$ is an anti-Jordan triple system.

From \eqref{eq:3Derivation} and \eqref{eq:AJTSExLA} it follows that $X \in \Der^+ V$ if and only if
\begin{equation}\label{eq:AJTSExLADer}
\phi (w) ( X [v,u] - [ Xv , u ] - [ v , Xu ] ) = \phi (Xw) [v,u]~,
\end{equation}
for all $v,w,u \in V$. If $X \in \ad (V)$ then the left and right hand sides of \eqref{eq:AJTSExLADer} both vanish due to the Jacobi identity for $V$ and $\phi ( [V,V] ) = 0$. Therefore $\ad (V)$ is a Lie subalgebra of $\Der^+ V$ which contains $d (V,V)$ and, as such, constitutes a derivation algebra of $V$.
\end{example}

\begin{example} \label{ex:AJTSDouble}
Let $V$ be an anti-Jordan triple system and let $\Delta$ be a two-dimensional vector space. Now define a 3-bracket on $V \otimes \Delta$ whose non-zero components are given by
\begin{equation}\label{eq:AJTSDouble}
[ v \otimes e_\pm , w \otimes e_\mp , u \otimes e_\pm ] = [v,w,u] \otimes e_\pm~,
\end{equation}
for all $v,w,u \in V$, where $e_\pm$ is a basis for $\Delta$. Clearly \eqref{eq:AJTSsym} is satisfied for $V \otimes \Delta$ since it is satisfied for $V$. 

From \eqref{eq:AJTSDouble} it follows that 
\begin{equation}\label{eq:AJTSDoubled}
d ( v \otimes e_\pm , w \otimes e_\mp ) = d(v,w) \otimes P_\pm~,
\end{equation}
for all $v,w \in V$, where $P_\pm \in \End \Delta$ are projections defined such that $P_\pm e_\pm = e_\pm$ and $P_\pm e_\mp = 0$. Therefore \eqref{eq:AJTSendd} is satisfied for $V \otimes \Delta$ since
\begin{align}\label{eq:AJTSDoubleEnd++}
[ d ( v_1 \otimes e_\pm , v_2 \otimes e_\mp ) , d ( w_1 \otimes e_\pm , w_2 \otimes e_\mp ) ]  &\underset{\eqref{eq:AJTSDoubled}}{=} [ d ( v_1 , v_2 ) , d ( w_1 , w_2 ) ] \otimes P_\pm \nonumber \\
&\underset{\eqref{eq:AJTSendd}}{=} d ( d ( v_1 , v_2 ) w_1 , w_2 ) \otimes P_\pm + d ( w_1 , d ( v_2 , v_1 ) w_2 ) \otimes P_\pm \nonumber \\
&\underset{\eqref{eq:AJTSDoubled}}{=} d ( d ( v_1 \otimes e_\pm , v_2 \otimes e_\mp )  ( w_1 \otimes e_\pm ) , w_2 \otimes e_\mp) \nonumber \\
&\hspace*{.8cm}+ d ( w_1 \otimes e_\pm , d ( v_2 \otimes e_\mp , v_1 \otimes e_\pm ) ( w_2 \otimes e_\mp ) )
\end{align}
and
\begin{align}\label{eq:AJTSDoubleEnd+-}
[ d ( v_1 \otimes e_\pm , v_2 \otimes e_\mp ) , d ( w_1 \otimes e_\mp , w_2 \otimes e_\pm ) ]  &\underset{\eqref{eq:AJTSDoubled}}{=} 0 \nonumber \\
&\underset{\eqref{eq:AJTSDoubled}}{=} d ( d ( v_1 \otimes e_\pm , v_2 \otimes e_\mp )  ( w_1 \otimes e_\mp ) , w_2 \otimes e_\pm) \nonumber \\
&\hspace*{.8cm}+ d ( w_1 \otimes e_\mp , d ( v_2 \otimes e_\mp , v_1 \otimes e_\pm ) ( w_2 \otimes e_\pm ) )~,
\end{align}
for all $v_1 , v_2 , w_1 , w_2 \in V$. So $V \otimes \Delta$ is an anti-Jordan triple system. (This construction is equivalent to Example 4 in \cite{Kamiya}.)

Now let us denote $\mathbb{1} = P_+ + P_-$ and $Q = P_+ - P_-$. If $X \in \Der^+ V$ then 
\begin{align}\label{eq:AJTSDoubleDer}
[ X \otimes \mathbb{1} , d ( v \otimes e_\pm , w \otimes e_\mp ) ]  &\underset{\eqref{eq:AJTSDoubled}}{=} [ X , d ( v , w ) ] \otimes P_\pm \nonumber \\
&\, \underset{\eqref{eq:3dDerivation}}{=} \, d ( Xv , w ) \otimes P_\pm + d ( v , Xw ) \otimes P_\pm \nonumber \\
&\underset{\eqref{eq:AJTSDoubled}}{=} d ( Xv \otimes e_\pm , w \otimes e_\mp ) + d ( v \otimes e_\pm , Xw \otimes  e_\mp ) \nonumber \\
&\;\;\, = \;\;\, d ( ( X \otimes \mathbb{1} ) ( v \otimes e_\pm ) , w \otimes e_\mp ) + d ( v \otimes e_\pm , ( X \otimes \mathbb{1} ) ( w \otimes  e_\mp ) )~,
\end{align}
for all $v,w \in V$, so $X \otimes \mathbb{1} \in \Der^+ ( V \otimes \Delta )$. If $X \in \Der^- V$ then\begin{align}\label{eq:AJTSDoubleAntiDer}
[ X \otimes Q , d ( v \otimes e_\pm , w \otimes e_\mp ) ]  &\underset{\eqref{eq:AJTSDoubled}}{=} \pm [ X , d ( v , w ) ] \otimes P_\pm \nonumber \\
&\, \underset{\eqref{eq:3dDerivation}}{=} \, \pm d ( Xv , w ) \otimes P_\pm \mp d ( v , Xw ) \otimes P_\pm \nonumber \\
&\underset{\eqref{eq:AJTSDoubled}}{=} \pm d ( Xv \otimes e_\pm , w \otimes e_\mp ) \mp d ( v \otimes e_\pm , Xw \otimes  e_\mp ) \nonumber \\
&\;\;\, = \;\;\, d ( ( X \otimes Q ) ( v \otimes e_\pm ) , w \otimes e_\mp ) + d ( v \otimes e_\pm , ( X \otimes Q ) ( w \otimes  e_\mp ) )~,
\end{align}
for all $v,w \in V$, so $X \otimes Q \in \Der^+ ( V \otimes \Delta )$. It follows that $\Der^+ V \otimes \mathbb{1} \oplus \Der^- V \otimes Q$ is a Lie subalgebra of $\Der^+ ( V \otimes \Delta )$.

Furthermore,
\begin{align}\label{eq:AJTSDoubled+}
d^+ ( v \otimes e_+ , w \otimes e_- ) &\underset{\eqref{eq:dPMVV}}{=} d ( v \otimes e_+ , w \otimes e_- ) + d ( w \otimes e_- ,  v \otimes e_+ )  \nonumber \\
&\underset{\eqref{eq:AJTSDoubled}}{=} d ( v , w ) \otimes P_+ + d ( w , v ) \otimes P_- \nonumber \\
 &\underset{\eqref{eq:dPMVV}}{=} \half ( d^+ (v,w) \otimes P_+ + d^- (v,w) \otimes P_+ + d^+ (w,v) \otimes P_- + d^- (w,v) \otimes P_- )  \nonumber \\
&\;\;\, = \;\;\,  \half ( d^+ (v,w) \otimes \mathbb{1} + d^- (v,w) \otimes Q )~,
\end{align}
for all $v,w \in V$. This means $\Der^+ V \otimes \mathbb{1} \oplus \Der^- V \otimes Q$ contains $d^+ ( V \otimes \Delta , V \otimes \Delta )$ and is therefore a derivation algebra of $V \otimes \Delta$. 
\end{example}

We will now describe the method introduced in \cite{Kamiya} to construct an anti-Lie triple system from an anti-Jordan triple system.

\begin{prop} \label{prop:ALTS-AJTS}
If $V$ is an anti-Jordan triple system with 3-bracket $[-,-,-]_J$ then $V$ is an anti-Lie triple system with 3-bracket $[-,-,-]_L$ defined by
\begin{equation}\label{eq:ALTSAJTS3bracket}
[ v , w , u ]_L = [ v , w , u ]_J + [ w , v , u ]_J~,
\end{equation}
for all $v,w,u \in V$. Furthermore, if $D(V)$ is a derivation algebra of $V$ as an anti-Jordan triple system then $D(V)$ is a derivation algebra of $V$ as an anti-Lie triple system. 
\end{prop}

\begin{proof}
Condition \eqref{eq:ALTSsym} is satisfied since
\begin{equation}\label{eq:ALTSAJTS3bracketSym}
[ v , w , u ]_L \underset{\eqref{eq:ALTSAJTS3bracket}}{=} [ v , w , u ]_J + [ w , v , u ]_J \underset{\eqref{eq:ALTSAJTS3bracket}}{=} [ w , v , u ]_L~,
\end{equation}
for all $v,w,u \in V$. Condition \eqref{eq:ALTScyc} is satisfied since
\begin{align}\label{eq:ALTSAJTS3bracketCyc}
[ v , w , u ]_L + [ w , u , v ]_L + [ u , v , w ]_L &\underset{\eqref{eq:ALTSAJTS3bracket}}{=} [ v , w , u ]_J + [ w , v , u ]_J + [ w , u , v ]_J + [ u , w , v ]_J + [ u , v , w ]_J + [ v , u , w ]_J \nonumber \\
&\underset{\eqref{eq:AJTSsym}}{=} 0~,
\end{align}
for all $v,w,u \in V$. Condition \eqref{eq:ALTSend} is satisfied since
\begin{align}\label{eq:ALTSAJTS3bracketEnd}
[ v_1 , v_2 , [ w_1 , w_2 , w_3 ]_L ]_L &\underset{\eqref{eq:ALTSAJTS3bracket}}{=} [ v_1 , v_2 , [ w_1 , w_2 , w_3 ]_J ]_J + ( v_1 \leftrightarrow v_2, w_1 \leftrightarrow w_2 ) \nonumber \\
&\underset{\eqref{eq:AJTSend}}{=} [ [ v_1 , v_2 , w_1 ]_J , w_2 , w_3 ]_J + [ w_1 , [ v_2 , v_1 , w_2 ]_J , w_3 ]_J + [ w_1 , w_2 , [ v_1 , v_2 , w_3 ]_J ]_J + ( v_1 \leftrightarrow v_2, w_1 \leftrightarrow w_2 ) \nonumber \\
&\underset{\eqref{eq:ALTSAJTS3bracket}}{=} [ [ v_1 , v_2 , w_1 ]_L , w_2 , w_3 ]_J + [ w_1 , [ v_1 , v_2 , w_2 ]_L , w_3 ]_J + [ w_1 , w_2 , [ v_1 , v_2 , w_3 ]_L ]_J + ( w_1 \leftrightarrow w_2 )  \nonumber \\
&\underset{\eqref{eq:ALTSAJTS3bracket}}{=} [ [ v_1 , v_2 , w_1 ]_L , w_2 , w_3 ]_L + [ w_1 , [ v_1 , v_2 , w_2 ]_L , w_3 ]_L + [ w_1 , w_2 , [ v_1 , v_2 , w_3 ]_L ]_L~,
\end{align}
for all $v_1 , v_2 , w_1 , w_2 , w_3 \in V$. 

If $X$ is a derivation of $V$ as an anti-Jordan triple system then it is also a derivation of $V$ as an anti-Lie triple system since 
\begin{align}\label{eq:ALTSAJTS3bracketDer}
X [ v , w , u ]_L &\underset{\eqref{eq:ALTSAJTS3bracket}}{=} X [ v , w , u ]_J + X [ w , v , u ]_J \nonumber \\
&\;\underset{\eqref{eq:3Derivation}}{=} [ Xv , w , u ]_J + [ v , Xw , u ]_J + [ v , w , Xu ]_J + [ Xw , v , u ]_J + [ w , Xv , u ]_J  + [ w , v , Xu ]_J \nonumber \\
&\underset{\eqref{eq:ALTSAJTS3bracket}}{=} [ Xv , w , u ]_L + [ v , Xw , u ]_L + [ v , w , Xu ]_L~,
\end{align}
for all $v,w,u \in V$. Furthermore, \eqref{eq:ALTSAJTS3bracket} implies $d_L (V,V) = d^+_J (V,V)$. It follows that if $D(V)$ is a Lie subalgebra of $\Der^+_J V$ which contains $d^+_J (V,V)$ then it is also a Lie subalgebra of $\Der^+_L V$ which contains $d_L (V,V)$.
\end{proof}

\begin{example} \label{ex:ALTSLA}
Proposition~\ref{prop:ALTS-AJTS} implies that the anti-Jordan triple system $V$ in Example~\ref{ex:AJTSLA} with 3-bracket
\begin{equation}\label{eq:ALTSExLAJ}
[ v , w , u ]_J = \phi (w) [v,u]~,
\end{equation}
for all $v,w,u \in V$, can be thought of as an anti-Lie triple system with 3-bracket
\begin{equation}\label{eq:ALTSExLAL}
[ v , w , u ]_L = \phi (w) [v,u] + \phi (v) [w,u]~,
\end{equation}
for all $v,w,u \in V$. Moreover,
\begin{equation}\label{eq:ALTSExLAd}
d_L ( v,w ) \underset{\eqref{eq:ALTSAJTS3bracket}}{=} d_J^+ (v,w) \underset{\eqref{eq:AJTSExLAd}}{=} \phi (w) \ad (v) + \phi (v) \ad (w)~,
\end{equation}
for all $v,w \in V$. 

Proposition~\ref{prop:ALTS-LSA} further implies that the anti-Lie triple system $V$ above with derivation algebra $\ad (V)$ defines a Lie superalgebra $\fG$ with $\fG_{\bar 0} = \ad (V)$, $\fG_{\bar 1} = V$ and brackets
\begin{align}
[ \ad (v) , \ad (w) ]_\fG &= \ad ( [v,w] )~, \label{eq:ALTSLSAbracket00ExLA} \\
[ \ad (v) , w ]_\fG &= [v,w]~, \label{eq:ALTSLSAbracket01ExLA} \\
[v,w]_\fG &= \phi (w) \ad (v) + \phi (v) \ad (w)~, \label{eq:ALTSLSAbracket11ExLA}
\end{align}
for all $v,w \in V$. By definition, $\fG$ is $1$-admissible if and only if $V$ admits a nondegenerate $\ad (V)$-invariant symmetric bilinear form. 

For example, let $V = \fs \oplus \fz$ be a reductive Lie algebra with $\fs$ semisimple and $\fz$ one-dimensional. Since $\fz$ is central in $V$, we have $\ad (V) = \ad ( \fs ) \cong \fs$. Since $[V,V] = [ \fs , \fs ]  = \fs$, the condition $\phi ( [V,V] ) =0$ is equivalent to $\phi ( \fs ) = 0$. We can therefore assume $\phi \in \fz^*$. If $\phi \neq 0$ then there exists a unique $z \in \fz$ with $\phi (z) = 1$.

In this case $\fG$ has $\fG_{\bar 0} = \ad ( \fs )$, $\fG_{\bar 1} = \fs \oplus \fz$ with non-zero brackets
\begin{align}
[ \ad ( s ) , \ad ( t ) ]_\fG &= \ad ( [ s , t ] )~, \label{eq:ALTSLSAbracket00ExLASS} \\
[ \ad ( s ) , t ]_\fG &= [s,t]~, \label{eq:ALTSLSAbracket01ExLASS} \\
[z,s]_\fG &= \ad (s)~,  \label{eq:ALTSLSAbracket11ExLASS}
\end{align}
for all $s,t \in \fs$. Any nondegenerate $\ad ( \fs )$-invariant symmetric bilinear form on $\fs \oplus \fz$ is the direct sum of non-zero scalar multiples of Killing forms on each simple factor of $\fs$ and a non-zero bilinear form on $\fz$. Therefore $\fG$ is $1$-admissible but not simple since $\ad ( \fs ) \oplus \fs$ is a proper ideal of $\fG$. Furthermore, since $\fG_{\bar 0}$ is semisimple, Corollary~\ref{cor:ReductiveEquivalentCentralExtensionPbar} implies that the Poincar\'{e} superalgebra defined by $\fG$ must be non-trivial.
\end{example}

\begin{example} \label{ex:ALTSDouble}
Proposition~\ref{prop:ALTS-AJTS} implies that the anti-Jordan triple system $V \otimes \Delta$ in Example~\ref{ex:AJTSDouble} with 3-bracket whose non-zero components are
\begin{equation}\label{eq:ALTSExDoubleJ}
[ v \otimes e_\pm , w \otimes e_\mp , u \otimes e_\pm ]_J = [v,w,u] \otimes e_\pm~,
\end{equation}
for all $v,w,u \in V$, can be thought of as an anti-Lie triple system with 3-bracket whose non-zero components are
\begin{equation}\label{eq:ALTSExDoubleL}
[ v \otimes e_\pm , w \otimes e_\mp , u \otimes e_\pm ]_L = [v,w,u] \otimes e_\pm~,
\end{equation}
for all $v,w,u \in V$. Moreover,
\begin{equation}\label{eq:ALTSExDoubled}
d_L ( v \otimes e_+ , w \otimes e_- ) \underset{\eqref{eq:ALTSAJTS3bracket}}{=} d_J^+ ( v \otimes e_+ , w \otimes e_- ) \underset{\eqref{eq:AJTSDoubled+}}{=} \half ( d^+ (v,w) \otimes \mathbb{1} + d^- (v,w) \otimes Q )~,
\end{equation}
for all $v,w \in V$. 

Proposition~\ref{prop:ALTS-LSA} further implies that the anti-Lie triple system $V \otimes \Delta$ above with derivation algebra $\Der^+ V \otimes \mathbb{1} \oplus \Der^- V \otimes Q$ defines a Lie superalgebra $\fG$ with $\fG_{\bar 0} = \Der^+ V \otimes \mathbb{1} \oplus \Der^- V \otimes Q$, $\fG_{\bar 1} = V \otimes \Delta$ and non-zero brackets
\begin{align}
[ X^+ \otimes \mathbb{1} , Y^+ \otimes \mathbb{1} ]_\fG &= [ X^+ , Y^+ ] \otimes \mathbb{1}~, \label{eq:ALTSLSAbracket00ExDouble++} \\
[ X^+ \otimes \mathbb{1} , Y^- \otimes Q ]_\fG &= [ X^+ , Y^- ] \otimes Q~, \label{eq:ALTSLSAbracket00ExDouble+-} \\
[ X^- \otimes Q , Y^- \otimes Q ]_\fG &= [ X^- , Y^- ] \otimes \mathbb{1}~, \label{eq:ALTSLSAbracket00ExDouble--} \\
[ X^+ \otimes \mathbb{1} , v \otimes e_\pm ]_\fG &= X^+ v \otimes e_\pm~, \label{eq:ALTSLSAbracket01ExDouble+} \\
[ X^- \otimes Q , v \otimes e_\pm ]_\fG &= \pm X^- v \otimes e_\pm~, \label{eq:ALTSLSAbracket01ExDouble-} \\
[ v \otimes e_+ , w \otimes e_-  ]_\fG &=  \half ( d^+ (v,w) \otimes \mathbb{1} + d^- (v,w) \otimes Q )~, \label{eq:ALTSLSAbracket11ExDouble}
\end{align}
for all $X^\pm , Y^\pm \in \Der^\pm V$ and $v,w \in V$. 

Now let $b$ be a nondegenerate symmetric bilinear form on $V$ and let $\beta$ be a nondegenerate bilinear form on $\Delta$ with $\beta ( e_\pm , e_\mp ) = 0$. It follows that $b \otimes \beta$ is a nondegenerate symmetric bilinear form on $V \otimes \Delta$ that is invariant under $\Der^+ V \otimes \mathbb{1} \oplus \Der^- V \otimes Q$ if and only if $b$ is invariant under both $\Der^+ V$ and $\Der^- V$, in which case $\fG$ is $1$-admissible.
\end{example}

%%%%%%%%%%%%%%%%%%%%%%

\subsection{$2$-admissible Lie superalgebras from polarised anti-Jordan triple systems}
\label{sec:PolarisedAntiJordanTripleSystems}

An anti-Lie or anti-Jordan triple system $V$ with derivation algebra $D(V)$ will be called {\emph{polarised}} if $V = V_+ \oplus V_-$, $D(V) V_\pm \subset V_\pm$ and $d ( V_\pm , V_\pm ) = 0$.

\begin{prop} \label{prop:PolALTS-3GradedLSA}
The Lie superalgebra defined by a polarised anti-Lie triple system with a derivation algebra as in Proposition~\ref{prop:ALTS-LSA} is 3-graded. Conversely, the anti-Lie triple system with a derivation algebra defined by a 3-graded Lie superalgebra as in Proposition~\ref{prop:LSA-ALTS} is polarised. 
\end{prop}

\begin{proof}
If $V$ is an anti-Lie triple system with derivation algebra $D(V)$ then Proposition~\ref{prop:ALTS-LSA} implies $\fG$ is a Lie superalgebra with $\fG_{\bar 0} = D(V)$ and $\fG_{\bar 1} = V$. If $( V , D(V) )$ is polarised then define $\fG_{\bar 1}^\pm = V_\pm$ so that $\fG_{\bar 1} = \fG_{\bar 1}^+ \oplus \fG_{\bar 1}^-$. The $[{\bar 0}{\bar 1}]$ bracket for $\fG$ in \eqref{eq:ALTSLSAbracket01} has $[ \fG_{\bar 0} , \fG_{\bar 1}^\pm ] \subset \fG_{\bar 1}^\pm$ since $D(V) V_\pm \subset V_\pm$. The $[{\bar 1}{\bar 1}]$ bracket for $\fG$ in \eqref{eq:ALTSLSAbracket11} has $[ \fG_{\bar 1}^\pm , \fG_{\bar 1}^\pm ] = 0$ since $d ( V_\pm , V_\pm ) = 0$. Therefore $\fG$ is 3-graded.

If $\fG$ is a Lie superalgebra then Proposition~\ref{prop:LSA-ALTS} implies $\fG_{\bar 1}$ is an anti-Lie triple system with derivation algebra $\ad ( \fG_{\bar 0} ) |_{\fG_{\bar 1}}$. If $\fG$ is 3-graded then $\fG_{\bar 1} = \fG_{\bar 1}^+ \oplus \fG_{\bar 1}^-$. Furthermore, $\ad ( \fG_{\bar 0} ) \fG_{\bar 1}^\pm = [ \fG_{\bar 0} , \fG_{\bar 1}^\pm ] \subset \fG_{\bar 1}^\pm$ and $d ( \fG_{\bar 1}^\pm , \fG_{\bar 1}^\pm ) = \ad ( [ \fG_{\bar 1}^\pm , \fG_{\bar 1}^\pm ] ) |_{\fG_{\bar 1}} = 0$ with respect to the 3-bracket for $\fG_{\bar 1}$ in \eqref{eq:LSAALTS3bracket}. Therefore $( \fG_{\bar 1} , \ad ( \fG_{\bar 0} ) |_{\fG_{\bar 1}} )$ is polarised.
\end{proof}

\begin{cor} \label{cor:3GradedLSAPolALTSbijection}
Up to isomorphism, there is a bijective correspondence between polarised anti-Lie triple systems with a derivation algebra and 3-graded Lie superalgebras $\fG$ with $\fG_{\bar 0}$ acting faithfully on $\fG_{\bar 1}$.
\end{cor}

\begin{proof}
Substitute Proposition~\ref{prop:PolALTS-3GradedLSA} into Corollary~\ref{cor:LSAALTSbijection}.
\end{proof}

Proposition~\ref{prop:ALTS-AJTS} implies that any anti-Jordan triple system with a derivation algebra defines an anti-Lie triple system with the same derivation algebra. For polarised triple systems, this correspondence is actually bijective in the sense of the following theorem (which generalises Theorem 1.4 in \cite{Kamiya}).

\begin{thm} \label{thm:PolALTS-PolAJTS}
There is a bijective correspondence between polarised anti-Lie triple systems and polarised anti-Jordan triple systems (both with the same derivation algebra). 
\end{thm}

\begin{proof}
If $V$ is an anti-Jordan triple system with 3-bracket $[-,-,-]_J$ and derivation algebra $D(V)$ then Proposition~\ref{prop:ALTS-AJTS} implies that $V$ is an anti-Lie triple system with 3-bracket $[-,-,-]_L$ given by \eqref{eq:ALTSAJTS3bracket} and derivation algebra $D(V)$. If $( V , D(V) )$ is polarised as an anti-Jordan triple system then clearly it is also polarised as an anti-Lie triple system since $d_L ( V_\pm , V_\pm ) = d^+_J ( V_\pm , V_\pm ) = 0$.

Conversely, if $V$ is a polarised anti-Lie triple system with 3-bracket $[-,-,-]_L$ and derivation algebra $D(V)$, let us define a 3-bracket $[-,-,-]_J$ on $V$ such that
\begin{align}
[ v , w_\mp , u_\pm ]_J &= [ v , w_\mp , u_\pm ]_L~, \label{eq:PolALTSAJTS3bracket+-} \\
[ v , w_\pm , u_\pm ]_J &= 0~, \label{eq:PolALTSAJTS3bracket++} 
\end{align}
for all $v \in V$ and $w_\pm , u_\pm \in V_\pm$. 

Since $d_L ( V_\pm , V_\pm ) = 0$, we have
\begin{equation}\label{eq:PolALTSAJTS3bracket++-}
[ v_\pm , w_\pm , u ]_J \underset{\eqref{eq:PolALTSAJTS3bracket++}}{=} [ v_\pm , w_\pm , u_\mp ]_J \underset{\eqref{eq:PolALTSAJTS3bracket+-}}{=} [ v_\pm , w_\pm , u_\mp ]_L = 0~,
\end{equation}
for all $v_\pm , w_\pm \in V_\pm$ and $u \in V$, so $d_J ( V_\pm , V_\pm ) = 0$. Furthermore, 
\begin{equation}\label{eq:PolALTSAJTS3bracket+-+}
[ v_\pm , w_\mp , u ]_J \underset{\eqref{eq:PolALTSAJTS3bracket++}}{=} [ v_\pm , w_\mp , u_\pm ]_J \underset{\eqref{eq:PolALTSAJTS3bracket+-}}{=} [ v_\pm , w_\mp , u_\pm ]_L~,
\end{equation}
for all $v_\pm , w_\pm \in V_\pm$ and $u \in V$, so $d_J ( V_\pm , V_\mp ) V = d_L ( V_\pm , V_\mp ) V_\pm \subset V_\pm$ since $d_L (V,V) \subset D(V)$. 

It follows that the only non-zero component of $[-,-,-]_J$ is $[ V_\pm , V_\mp , V_\pm ]_J \subset V_\pm$. Therefore \eqref{eq:AJTSsym} is satisfied since
\begin{equation}\label{eq:PolALTSAJTS3bracketSym}
[ v_\pm , w_\mp , u_\pm ]_J + [ u_\pm , w_\mp , v_\pm ]_J \underset{\eqref{eq:PolALTSAJTS3bracket+-}}{=} [ v_\pm , w_\mp , u_\pm ]_L + [ u_\pm , w_\mp , v_\pm ]_L \underset{\eqref{eq:ALTScyc}}{=} - [ u_\pm , v_\pm , w_\mp ]_L = 0~,
\end{equation}
for all $v_\pm , w_\pm , u_\pm \in V_\pm$, using $d_L ( V_\pm , V_\pm ) = 0$ in the last equality. Moreover, \eqref{eq:AJTSend} is satisfied since
\begin{align}\label{eq:PolALTSAJTS3bracketEnd}
[ v_\pm^1 , v_\mp^2 , [ w_\pm^1 , w_\mp^2 , w_\pm^3 ]_J ]_J &\underset{\eqref{eq:PolALTSAJTS3bracket+-}}{=} [ v_\pm^1 , v_\mp^2 , [ w_\pm^1 , w_\mp^2 , w_\pm^3 ]_L ]_L \nonumber \\
&\, \underset{\eqref{eq:ALTSend}}{=} \,  [ [ v_\pm^1 , v_\mp^2 , w_\pm^1 ]_L , w_\mp^2 , w_\pm^3 ]_L + [ w_\pm^1 , [ v_\pm^1 , v_\mp^2 , w_\mp^2 ]_L , w_\pm^3 ]_L + [ w_\pm^1 , w_\mp^2 , [ v_\pm^1 , v_\mp^2 , w_\pm^3 ]_L ]_L \nonumber \\
&\underset{\eqref{eq:PolALTSAJTS3bracket+-}}{=} [ [ v_\pm^1 , v_\mp^2 , w_\pm^1 ]_J , w_\mp^2 , w_\pm^3 ]_J + [ w_\pm^1 , [ v_\mp^2 , v_\pm^1 , w_\mp^2 ]_J , w_\pm^3 ]_J + [ w_\pm^1 , w_\mp^2 , [ v_\pm^1 , v_\mp^2 , w_\pm^3 ]_J ]_J~,
\end{align}
for all $v_\pm^1 , v_\pm^2 , w_\pm^1 , w_\pm^2 , w_\pm^3 \in V_\pm$. Therefore $V$ is an anti-Jordan triple system with 3-bracket $[-,-,-]_J$.
 
Since $[ V_\pm , V_\mp , V_\pm ]_J = [ V_\pm , V_\mp , V_\pm ]_L$ then clearly $D(V)$ is a Lie algebra of derivations of $V$ as an anti-Jordan triple system because the same is true of $V$ as an anti-Lie triple system. Furthermore, 
\begin{equation}\label{eq:PolALTSAJTSdJ++}
d_J^+ ( v_+ , w_- ) u_+ \underset{\eqref{eq:dPMVV}}{=} d_J ( v_+ , w_- ) u_+ + d_J ( w_- , v_+ ) u_+  \underset{\eqref{eq:PolALTSAJTS3bracket++}}{=} d_J ( v_+ , w_- ) u_+ \underset{\eqref{eq:PolALTSAJTS3bracket+-}}{=} d_L ( v_+ , w_- ) u_+~,
\end{equation}
for all $v_+ , u_+ \in V_+$, $w_- \in V_-$, and 
\begin{equation}\label{eq:PolALTSAJTSdJ+-}
d_J^+ ( v_+ , w_- ) u_- \underset{\eqref{eq:dPMVV}}{=} d_J ( v_+ , w_- ) u_- + d_J ( w_- , v_+ ) u_- \underset{\eqref{eq:PolALTSAJTS3bracket++}}{=}  d_J ( w_- , v_+ ) u_- \underset{\eqref{eq:PolALTSAJTS3bracket+-}}{=} d_L ( w_- , v_+ ) u_- \underset{\eqref{eq:ALTSsym}}{=} d_L ( v_+ , w_- ) u_-~,
\end{equation}
for all $v_+ \in V_+$, $w_- , u_- \in V_-$, so $d_J^+ ( V_+ , V_- ) = d_L ( V_+ , V_- )$. It follows that $d_J^+ ( V , V ) = d_L ( V , V ) \subset D(V)$ since $d_J ( V_\pm , V_\pm ) = 0$ and $d_L ( V_\pm , V_\pm ) = 0$. Therefore $D(V)$ is a derivation algebra of $V$ as an anti-Jordan triple system and $( V , D(V) )$ remains polarised.

Applying the construction in Proposition~\ref{prop:ALTS-AJTS} to the polarised anti-Jordan triple system with 3-bracket $[-,-,-]_J$ above defines a polarised anti-Lie triple system with precisely the original 3-bracket $[-,-,-]_L$. Conversely, applying the construction in Proposition~\ref{prop:ALTS-AJTS} to any polarised anti-Jordan triple system with 3-bracket $[-,-,-]_J$ and then applying the construction above to the resulting polarised anti-Lie triple system with 3-bracket $[-,-,-]_L$ defines a polarised anti-Jordan triple system with precisely the original 3-bracket $[-,-,-]_J$. 
\end{proof}

\begin{remark}
Yet another equivalent way to present a polarised anti-Lie or anti-Jordan triple system $V = V_+ \oplus V_-$ is as a so-called {\emph{anti-Jordan pair}} $( V_+ , V_- )$ equipped with a pair of trilinear maps $V_\pm \times V_\mp \times V_\pm \rightarrow V_\pm$ which encode the 3-bracket (see \cite{FaulknerFerrar,Kamiya}).
\end{remark}

It follows that, up to isomorphism, every embedding superalgebra $\fG (2)$ with $\fG (2)_{\bar 0}$ acting faithfully on $\fG (2)_{\bar 1}$ can be constructed from a unique polarised anti-Jordan triple system $V$ with derivation algebra $D(V)$. Clearly every pair of nondegenerate $\fG (2)_{\bar 0}$-invariant symmetric bilinear forms on $\fG (2)_{\bar 1}^\pm$ correspond to a unique pair of nondegenerate $D(V)$-invariant symmetric bilinear forms on $V_\pm$, so we can use the existence of either to recognise when $\fG (2)$ is $2$-admissible.

\begin{example} \label{ex:AJTSLAPM}
Let $V = V_+ \oplus V_-$ be a Lie algebra equipped with Lie bracket $[-,-]$ and elements $\phi_\pm \in V_\pm^*$ with $\phi_\pm ( [ V_\pm , V_\pm ] ) = 0$. Now define a 3-bracket $[-,-,-]_J$ on $V$ with non-zero components
\begin{equation}\label{eq:AJTSExLAPM}
[ v_\pm , w_\mp , u_\pm ]_J = \phi_\mp ( w_\mp ) [ v_\pm , u_\pm ]~,
\end{equation}
for all $v_\pm , w_\pm , u_\pm \in V_\pm$. Clearly \eqref{eq:AJTSsym} is satisfied since $[-,-]$ is skewsymmetric. 

From \eqref{eq:AJTSExLAPM} it follows that 
\begin{equation}\label{eq:AJTSExLAPMd}
d_J ( v_\pm , w_\mp ) = \phi_\mp ( w_\mp ) \ad ( v_\pm )~,
\end{equation}
for all $v_\pm , w_\pm \in V_\pm$. Furthermore, $d_J ( V_\pm , V_\pm ) = 0$. 

Therefore 
\begin{equation}\label{eq:AJTSExLAPMdTerm2}
d_J ( w_\pm^1 , d_J ( v_\mp^2 , v_\pm^1 ) w_\mp^2 ) \underset{\eqref{eq:AJTSExLAPMd}}{=} \phi_\mp ( d_J ( v_\mp^2 , v_\pm^1 ) w_\mp^2 ) \ad ( w_\pm^1 ) \underset{\eqref{eq:AJTSExLAPM}}{=} \phi_\pm ( v_\pm^1 ) \phi_\mp ( [ v_\mp^2 , w_\mp^2 ] ) \ad ( w_\pm^1 ) = 0~,
\end{equation}
for all $v_\pm^1 , v_\pm^2 , w_\pm^1 , w_\pm^2 \in V_\pm$, since $\phi_\pm ( [ V_\pm , V_\pm ] ) = 0$. It follows that the second term on the right hand side in \eqref{eq:AJTSendd} is identically zero, so $d_J (V,V) \subset \Der^+ V$ if  \eqref{eq:AJTSendd} is satisfied since $d_J ( V_\pm , V_\mp ) V_\mp = 0$. This is indeed the case since 
\begin{align}\label{eq:AJTSExLAPMdTerms}
[ d_J ( v_\pm^1 , v_\mp^2 ) , d_J ( w_\pm^1 , w_\mp^2 ) ] &\underset{\eqref{eq:AJTSExLAPMd}}{=} \phi_\mp ( v_\mp^2 ) \phi_\mp ( w_\mp^2 ) [ \ad ( v_\pm^1 ) , \ad ( w_\pm^1 ) ] \nonumber \\
&\;\;\, = \;\;\,  \phi_\mp ( v_\mp^2 ) \phi_\mp ( w_\mp^2 ) \ad ( [ v_\pm^1 , w_\pm^1 ] ) \nonumber \\
&\underset{\eqref{eq:AJTSExLAPM}}{=} \phi_\mp ( w_\mp^2 ) \ad ( [ v_\pm^1 , v_\mp^2 , w_\pm^1 ]_J ) \nonumber \\
&\underset{\eqref{eq:AJTSExLAPMd}}{=} d_J ( d_J ( v_\pm^1 , v_\mp^2 ) w_\pm^1 , w_\mp^2 )~,
\end{align}
for all $v_\pm^1 , v_\pm^2 , w_\pm^1 , w_\pm^2 \in V_\pm$. Therefore $V$ is an anti-Jordan triple system.

From \eqref{eq:3Derivation} and \eqref{eq:AJTSExLAPM} it follows that any $X \in \End V_+ \oplus \End V_-$ is a derivation of $V$ if and only if
\begin{equation}\label{eq:AJTSExLAPMDer}
\phi_\mp ( w_\mp ) ( X [ v_\pm , u_\pm ] - [ X v_\pm , u_\pm ] - [ v_\pm , X u_\pm ] ) = \phi_\mp ( X w_\mp ) [ v_\pm , u_\pm ]~,
\end{equation}
for all $v_\pm , w_\pm , u_\pm \in V_\pm$. If $X \in \ad V$ then the left and right hand sides of \eqref{eq:AJTSExLAPMDer} both vanish due to the Jacobi identity for $V_\pm$ and $\phi_\mp ( [ V_\mp , V_\mp ] ) = 0$. Therefore $\ad ( V )$ is a Lie subalgebra of $\Der^+ V$ which contains $d_J (V,V)$ and, as such, constitutes a derivation algebra of $V$. Since $\ad ( V ) V_\pm = [ V_\pm , V_\pm ] \subset V_\pm$ and $d_J ( V_\pm , V_\pm ) = 0$, it follows that $( V , \ad ( V ) )$ is polarised.

Proposition~\ref{prop:ALTS-AJTS} and Theorem~\ref{thm:PolALTS-PolAJTS} imply that $V$ can be thought of as a polarised anti-Lie triple system with 3-bracket $[-,-,-]_L$ whose non-zero components are
\begin{equation} \label{eq:ALTSExLAPMLPM} 
[ v_\pm , w_\mp , u_\pm ]_L = \phi_\mp ( w_\mp ) [ v_\pm , u_\pm ]~,
\end{equation}
for all $v_\pm , w_\pm , u_\pm \in V_\pm$, and derivation algebra $\ad ( V )$. Moreover,
\begin{equation}\label{eq:ALTSExLAPMd}
d_L ( v_+ , w_- ) \underset{\eqref{eq:ALTSAJTS3bracket}}{=} d_J^+ ( v_+ , w_- ) \underset{\eqref{eq:AJTSExLAPMd}}{=} \phi_+ ( v_+ ) \ad ( w_- ) + \phi_- ( w_- ) \ad ( v_+ )~,
\end{equation}
for all $v_+ \in V_+$ and $w_- \in V_-$. 

Propositions~\ref{prop:ALTS-LSA} and \ref{prop:PolALTS-3GradedLSA} further imply that the polarised anti-Lie triple system $V$ with derivation algebra $\ad ( V )$ defines a 3-graded Lie superalgebra $\fG$ with $\fG_{\bar 0} = \ad ( V )$, $\fG_{\bar 1}^\pm = V_\pm$ and non-zero brackets
\begin{align}
[ \ad ( v_\pm ) , \ad ( w_\pm ) ]_\fG &= \ad ( [ v_\pm , w_\pm ] )~, \label{eq:ALTSLSAbracket00ExLAPM} \\
[ \ad ( v_\pm ) , w_\pm ]_\fG &= [ v_\pm , w_\pm ]~, \label{eq:ALTSLSAbracket01ExLAPM} \\
[ v_+ , w_- ]_\fG &= \phi_+ ( v_+ ) \ad ( w_- ) + \phi_- ( w_- ) \ad ( v_+ )~, \label{eq:ALTSLSAbracket11ExLAPM}
\end{align}
for all $v_\pm , w_\pm \in V_\pm$. By definition, $\fG$ is $2$-admissible if and only if $V_\pm$ admit nondegenerate $\ad ( V_\pm )$-invariant symmetric bilinear forms. 
\end{example}

\begin{example} \label{ex:ALTSDouble3graded}
The anti-Jordan triple system $V \otimes \Delta = V \otimes e_+ \oplus V \otimes e_-$ with derivation algebra $\Der^+ V \otimes \mathbb{1} \oplus \Der^- V \otimes Q$ in Example~\ref{ex:ALTSDouble} is polarised and the associated Lie superalgebra $\fG$ with $\fG_{\bar 0} = \Der^+ V \otimes \mathbb{1} \oplus \Der^- V \otimes Q$ and $\fG_{\bar 1}^\pm = V \otimes e_\pm$ is 3-graded. 

If $\fG$ is $1$-admissible with respect to the bilinear form $b \otimes \beta$ on $V \otimes \Delta$ in Example~\ref{ex:ALTSDouble} then it is also $2$-admissible with respect to the bilinear forms $b \otimes \beta |_{e_\pm}$ on $V \otimes e_\pm$.
\end{example}

%%%%%%%%%%%%%%%%%%%%%%

\subsection{$3$-admissible Lie superalgebras from Filippov triple systems}
\label{sec:FilippovTripleSystems}

Now let $V$ be a polarised anti-Lie triple system with $V_+ = V_-$ and derivation algebra $D(V)$. In this case we will denote $V_\pm$ by $V_\bullet$ in order to distinguish it from $( V_\bullet , 0 )$ and $( 0 , V_\bullet )$ in $V$. For any $v \in V_\bullet$, we will write $v_+ = (v,0)$ and $v_- = (0,v)$ as elements in $V$. 

Any $X \in D(V)$ defines $X_\pm \in \End V_\bullet$ such that
\begin{equation}\label{eq:BalancedEndPolarisedAJTS}
X v_\pm = ( X_\pm v )_\pm~,
\end{equation}
for all $v \in V_\bullet$. We will say that $D(V)$ is {\emph{diagonal}} if $X_+ = X_-$, for all $X \in D(V)$, and write $X_\bullet = X_\pm$.

The pair $( V , D(V) )$ will be called {\emph{balanced}} if 
\begin{equation}\label{eq:BalancedSkewPolarisedAJTS}
d ( v_+ , w_- ) = - d ( w_+ , v_- )~,
\end{equation}
for all $v, w \in V_\bullet$, and $D(V)$ is diagonal. 

\begin{prop} \label{prop:BalPolALTS-Bal3GradedLSA}
The Lie superalgebra defined by a balanced polarised anti-Lie triple system with a derivation algebra as in Proposition~\ref{prop:ALTS-LSA} is 3-graded and balanced. Conversely, the anti-Lie triple system with a derivation algebra defined by a balanced 3-graded Lie superalgebra as in Proposition~\ref{prop:LSA-ALTS} is polarised and balanced. 
\end{prop}

\begin{proof}
If $V$ is a polarised anti-Lie triple system with derivation algebra $D(V)$ then Proposition~\ref{prop:PolALTS-3GradedLSA} implies $\fG$ is a 3-graded Lie superalgebra with $\fG_{\bar 0} = D(V)$ and $\fG_{\bar 1}^\pm = V_\pm$. If $( V , D(V) )$ is balanced then $V_+ = V_-$ implies $\fG_{\bar 1}^+ = \fG_{\bar 1}^-$ (as representations of $\fG_{\bar 0}$ since $D(V)$ is diagonal) and \eqref{eq:BalancedSkewPolarisedAJTS} implies \eqref{eq:BalancedSkew}. Therefore $\fG$ is balanced. 

If $\fG$ is a 3-graded Lie superalgebra then Proposition~\ref{prop:PolALTS-3GradedLSA} implies $\fG_{\bar 1} = \fG_{\bar 1}^+ \oplus \fG_{\bar 1}^-$ is a polarised anti-Lie triple system with derivation algebra $\ad ( \fG_{\bar 0} ) |_{\fG_{\bar 1}}$. If $\fG$ is balanced then $\fG_{\bar 1}^+ = \fG_{\bar 1}^-$ as representations of $\fG_{\bar 0}$ so $\ad ( \fG_{\bar 0} ) |_{\fG_{\bar 1}}$ is diagonal. Furthermore, \eqref{eq:BalancedSkew} implies \eqref{eq:BalancedSkewPolarisedAJTS}. Therefore $( \fG_{\bar 1} , \ad ( \fG_{\bar 0} ) |_{\fG_{\bar 1}} )$ is balanced.
\end{proof}

\begin{cor} \label{cor:Bal3GradedLSABalPolALTSbijection}
Up to isomorphism, there is a bijective correspondence between balanced polarised anti-Lie triple systems with a derivation algebra and balanced 3-graded Lie superalgebras $\fG$ with $\fG_{\bar 0}$ acting faithfully on $\fG_{\bar 1}$.
\end{cor}

\begin{proof}
Substitute Proposition~\ref{prop:BalPolALTS-Bal3GradedLSA} into Corollary~\ref{cor:3GradedLSAPolALTSbijection}.
\end{proof}

A triple system $V$ will be called a {\emph{Filippov triple system}} if 
\begin{equation}\label{eq:FTSskewsym}
[ v , w , u ] = - [ w , v , u ] = - [ u , w , v ] = - [ v , u , w ]~,
\end{equation}
for all $v,w,u \in V$, and
\begin{equation}\label{eq:FTSend}
[ v_1 , v_2 , [ w_1 , w_2 , w_3 ] ] = [ [ v_1 , v_2 , w_1 ] , w_2 , w_3 ] + [ w_1 , [ v_1 , v_2 , w_2 ] , w_3 ] + [ w_1 , w_2 , [ v_1 , v_2 , w_3 ] ]~,
\end{equation}
for all $v_1 , v_2 , w_1 , w_2 , w_3 \in V$. This is a special case of a so-called {\emph{$n$-Lie algebra}} introduced in \cite{Filippov} with $n=3$.

The condition \eqref{eq:FTSend} simply means $d(V,V) \subset \Der^+ V$. Any Lie subalgebra $D(V)$ of $\Der^+ V$ which contains $d(V,V)$ will be referred to as a {\emph{derivation algebra}} of the Filippov triple system $V$.

\begin{thm} \label{thm:BalPolALTS-FTS}
There is a bijective correspondence between balanced polarised anti-Lie triple systems and Filippov triple systems (with isomorphic derivation algebras). 
\end{thm}

\begin{proof}
If $V$ is a balanced polarised anti-Lie triple system with 3-bracket $[-,-,-]^L$ and derivation algebra $D(V)$ then let us define a 3-bracket $[-,-,-]^F$ on $V_\bullet$ such that
\begin{equation}\label{eq:ALTS-FTS3bracket}
[ v , w , u ]^F_\pm = [ v_+ , w_- , u_\pm ]^L~,
\end{equation}
for all $v,w,u \in V_\bullet$. 

The 3-bracket $[-,-,-]^F$ satisfies \eqref{eq:FTSskewsym} since
\begin{align}
[ v , w , u ]^F_\pm &\underset{\eqref{eq:ALTS-FTS3bracket}}{=} [ v_+ , w_- , u_\pm ]^L \underset{\eqref{eq:BalancedSkewPolarisedAJTS}}{=}  - [ w_+ , v_- , u_\pm ]^L \underset{\eqref{eq:ALTS-FTS3bracket}}{=} - [ w , v , u ]^F_\pm~, \label{eq:ALTS-FTS3bracketskewsym12} \\
[ v , w , u ]^F_+ &\underset{\eqref{eq:ALTS-FTS3bracket}}{=} [ v_+ , w_- , u_+ ]^L 
\overset{\eqref{eq:ALTSsym}}{\underset{\eqref{eq:ALTScyc}}{=}}  - [ u_+ , w_- , v_+ ]^L \underset{\eqref{eq:ALTS-FTS3bracket}}{=} - [ u , w , v ]^F_+~, \label{eq:ALTS-FTS3bracketskewsym+13} \\
[ v , w , u ]^F_- &\underset{\eqref{eq:ALTS-FTS3bracket}}{=} [ v_+ , w_- , u_- ]^L \overset{\eqref{eq:ALTSsym}}{\underset{\eqref{eq:ALTScyc}}{=}} - [ v_+ , u_- , w_- ]^L \underset{\eqref{eq:ALTS-FTS3bracket}}{=} - [ v , u , w ]^F_-~, \label{eq:ALTS-FTS3bracketskewsym-23}
\end{align}
for all $v,w,u \in V_\bullet$. The remaining properties in \eqref{eq:FTSskewsym} follow from those above since
\begin{align}
[ v , w , u ]^F_+ &\underset{\eqref{eq:ALTS-FTS3bracketskewsym12}}{=} - [ w , v , u ]^F_+ \underset{\eqref{eq:ALTS-FTS3bracketskewsym+13}}{=} [ u , v , w ]^F_+ \underset{\eqref{eq:ALTS-FTS3bracketskewsym12}}{=} - [ v , u , w ]^F_+~, \label{eq:ALTS-FTS3bracketskewsym+23} \\
[ v , w , u ]^F_- &\underset{\eqref{eq:ALTS-FTS3bracketskewsym12}}{=} - [ w , v , u ]^F_- \underset{\eqref{eq:ALTS-FTS3bracketskewsym-23}}{=} [ w , u , v ]^F_- \underset{\eqref{eq:ALTS-FTS3bracketskewsym12}}{=} - [ u , w , v ]^F_-~, \label{eq:ALTS-FTS3bracketskewsym-13}
\end{align}
for all $v,w,u \in V_\bullet$. Moreover, \eqref{eq:FTSend} is satisfied since
\begin{align}\label{eq:ALTS-FTS3bracketEnd}
[ v^1 , v^2 , [ w^1 , w^2 , w^3 ]^F ]^F_\pm &\underset{\eqref{eq:ALTS-FTS3bracket}}{=} [ v_+^1 , v_-^2 , [ w_+^1 , w_-^2 , w_\pm^3 ]^L ]^L \nonumber \\
&\, \underset{\eqref{eq:ALTSend}}{=} \,  [ [ v_+^1 , v_-^2 , w_+^1 ]^L , w_-^2 , w_\pm^3 ]^L + [ w_+^1 , [ v_+^1 , v_-^2 , w_-^2 ]^L , w_\pm^3 ]^L + [ w_+^1 , w_-^2 , [ v_+^1 , v_-^2 , w_\pm^3 ]^L ]^L \nonumber \\
&\underset{\eqref{eq:ALTS-FTS3bracket}}{=} [ [ v^1 , v^2 , w^1 ]^F , w^2 , w^3 ]^F_\pm + [ w^1 , [ v^1 , v^2 , w^2 ]^F , w^3 ]^F_\pm + [ w^1 , w^2 , [ v^1 , v^2 , w^3 ]^F ]^F_\pm~,
\end{align}
for all $v^1 , v^2 , w^1 , w^2 , w^3 \in V_\bullet$. Therefore $V_\bullet$ is a Filippov triple system with 3-bracket $[-,-,-]^F$.

If $X \in D(V)$ then 
\begin{align}\label{eq:ALTS-FTS3bracketDer}
( X_\bullet [ v , w , u ]^F )_\pm &\underset{\eqref{eq:BalancedEndPolarisedAJTS}}{=} X [ v , w , u ]^F_\pm \nonumber \\
&\underset{\eqref{eq:ALTS-FTS3bracket}}{=} X [ v_+ , w_- , u_\pm ]^L \nonumber \\
&\, \underset{\eqref{eq:3Derivation}}{=} \, [ X v_+ , w_- , u_\pm ]^L + [ v_+ , X w_- , u_\pm ]^L + [ v_+ , w_- , X u_\pm ]^L \nonumber \\
&\underset{\eqref{eq:BalancedEndPolarisedAJTS}}{=}  [ ( X_\bullet v )_+ , w_- , u_\pm ]^L + [ v_+ , ( X_\bullet w)_- , u_\pm ]^L + [ v_+ , w_- , ( X_\bullet u)_\pm ]^L \nonumber \\
&\underset{\eqref{eq:ALTS-FTS3bracket}}{=} [ X_\bullet v , w , u ]^F_\pm + [ v , X_\bullet w , u ]^F_\pm + [ v , w , X_\bullet u ]^F_\pm~,
\end{align}
for all $v,w,u \in V_\bullet$, so $X_\bullet \in \Der^+ V_\bullet$. Furthermore, \eqref{eq:ALTS-FTS3bracket} implies $d^F ( V_\bullet , V_\bullet ) = d^L ( V , V )_\bullet$. Therefore $D ( V_\bullet ) = \{ X_\bullet \in \End V_\bullet~|~ X \in D(V) \}$ is a derivation algebra of $V_\bullet$. Clearly $D ( V_\bullet ) \cong D(V)$ as Lie algebras since $D(V)$ is diagonal.

Conversely, if $V_\bullet$ is a Filippov triple system with 3-bracket $[-,-,-]^F$ and derivation algebra $D( V_\bullet )$ then let us define a 3-bracket $[-,-,-]^L$ on $V = V_\bullet \oplus V_\bullet$ with non-zero components 
\begin{equation}\label{eq:FTS-ALTS3bracket}
[ v_+ , w_- , u_\pm ]^L = [ w_- , v_+ , u_\pm ]^L = [ v , w , u ]^F_\pm~,
\end{equation}
for all $v,w,u \in V_\bullet$. 

Clearly \eqref{eq:FTS-ALTS3bracket} implies that $[-,-,-]^L$ obeys \eqref{eq:ALTSsym}. Condition \eqref{eq:ALTScyc} is satisfied since
\begin{align}\label{eq:FTS-ALTS3bracketCyc}
[ v_+ , w_+ , u_- ]^L + [ w_+ , u_- , v_+ ]^L + [ u_- , v_+ , w_+ ]^L &\underset{\eqref{eq:FTS-ALTS3bracket}}{=} [ w , u , v ]^F_+ + [ v , u , w ]^F_+ \underset{\eqref{eq:FTSskewsym}}{=} 0~, \nonumber \\
[ v_- , w_- , u_+ ]^L + [ w_- , u_+ , v_- ]^L + [ u_+ , v_- , w_- ]^L &\underset{\eqref{eq:FTS-ALTS3bracket}}{=} [ u , w , v ]^F_- + [ u , v , w ]^F_- \underset{\eqref{eq:FTSskewsym}}{=} 0~,
\end{align}
for all $v,w,u \in V_\bullet$. Furthermore, \eqref{eq:ALTSend} is satisfied since
\begin{align}\label{eq:FTS-ALTS3bracketEnd}
[ v_+^1 , v_-^2 , [ w_+^1 , w_-^2 , w_\pm^3 ]^L ]^L &\underset{\eqref{eq:FTS-ALTS3bracket}}{=} [ v^1 , v^2 , [ w^1 , w^2 , w^3 ]^F ]^F_\pm \nonumber \\
&\underset{\eqref{eq:FTSend}}{=} [ [ v^1 , v^2 , w^1 ]^F , w^2 , w^3 ]^F_\pm + [ w^1 , [ v^1 , v^2 , w^2 ]^F , w^3 ]^F_\pm + [ w^1 , w^2 , [ v^1 , v^2 , w^3 ]^F ]^F_\pm \nonumber \\
&\underset{\eqref{eq:FTS-ALTS3bracket}}{=}  [ [ v_+^1 , v_-^2 , w_+^1 ]^L , w_-^2 , w_\pm^3 ]^L + [ w_+^1 , [ v_+^1 , v_-^2 , w_-^2 ]^L , w_\pm^3 ]^L + [ w_+^1 , w_-^2 , [ v_+^1 , v_-^2 , w_\pm^3 ]^L ]^L~,
\end{align}
for all $v^1 , v^2 , w^1 , w^2 , w^3 \in V_\bullet$. Therefore $V$ is an anti-Lie triple system with 3-bracket $[-,-,-]^L$. 

If $X_\bullet \in D( V_\bullet )$ then let $X \in \End V$ be defined such that 
\begin{equation}\label{eq:FTS-ALTS3bracketXvPM}
X v_\pm = ( X_\bullet v )_\pm~,
\end{equation}
for all $v \in V_\bullet$. It follows that
\begin{align}\label{eq:FTS-ALTS3bracketDer}
X [ v_+ , w_- , u_\pm ]^L &\underset{\eqref{eq:FTS-ALTS3bracket}}{=} X [ v , w , u ]^F_\pm \nonumber \\
&\underset{\eqref{eq:FTS-ALTS3bracketXvPM}}{=} ( X_\bullet [ v , w , u ]^F )_\pm \nonumber \\
&\, \underset{\eqref{eq:3Derivation}}{=} \, [ X_\bullet v , w , u ]^F_\pm + [ v , X_\bullet w , u ]^F_\pm + [ v , w , X_\bullet u ]^F_\pm \nonumber \\
&\underset{\eqref{eq:FTS-ALTS3bracket}}{=}  [ ( X_\bullet v )_+ , w_- , u_\pm ]^L + [ v_+ , ( X_\bullet w)_- , u_\pm ]^L + [ v_+ , w_- , ( X_\bullet u)_\pm ]^L \nonumber \\
&\underset{\eqref{eq:FTS-ALTS3bracketXvPM}}{=} [ X v_+ , w_- , u_\pm ]^L + [ v_+ , X w_- , u_\pm ]^L + [ v_+ , w_- , X u_\pm ]^L~,
\end{align}
for all $v,w,u \in V_\bullet$, so $X \in \Der^+ V$. Furthermore, $d^L (V,V) \subset \{ X \in \End V~|~ X_\bullet  \in D ( V_\bullet ) \}$ since \eqref{eq:FTS-ALTS3bracket} implies 
\begin{equation}\label{eq:FTS-ALTS3bracketdLVV}
d^L ( v_+ , w_- ) u_\pm = ( d^F (v,w) u )_\pm~,
\end{equation}
for all $v,w,u \in V_\bullet$, so $d^L ( V , V )_\bullet = d^F ( V_\bullet , V_\bullet ) \subset D ( V_\bullet )$. Therefore $D( V ) = \{ X \in \End V~|~ X_\bullet  \in D ( V_\bullet ) \}$ is a derivation algebra of $V$ and clearly $D(V) \cong D ( V_\bullet )$ as Lie algebras. 

By construction, $D(V)$ is diagonal and
\begin{equation}\label{eq:FTS-ALTS3bracketSkew}
d^L ( v_+ , w_- )_\bullet \underset{\eqref{eq:FTS-ALTS3bracketdLVV}}{=} d^F (v,w) \underset{\eqref{eq:FTSskewsym}}{=} - d^F (w,v) \underset{\eqref{eq:FTS-ALTS3bracketdLVV}}{=} - d^L ( w_+ , v_- )_\bullet ~,
\end{equation}
for all $v, w \in V_\bullet$. Therefore $( V , D(V) )$ is polarised and balanced. 
\end{proof}

It follows that, up to isomorphism, every embedding superalgebra $\fG (3)$ with $\fG (3)_{\bar 0}$ acting faithfully on $\fG (3)_{\bar 1}$ can be constructed from a unique Filippov triple system $V_\bullet$ with derivation algebra $D( V_\bullet )$. Clearly every nondegenerate $\fG (3)_{\bar 0}$-invariant symmetric bilinear form on $\fG (3)_{\bar 1}^\bullet$ corresponds to a unique nondegenerate $D( V_\bullet )$-invariant symmetric bilinear form on $V_\bullet$, so we can use the existence of either to recognise when $\fG (3)$ is $3$-admissible.

\begin{example} \label{ex:SimpleFTS}
Let $V_\bullet = \Delta_1 \otimes \Delta_2$, where $\Delta_{1,2} = \Delta$ is a two-dimensional vector space equipped with a nondegenerate skewsymmetric bilinear form $\varepsilon$. It follows that $b = \varepsilon \otimes \varepsilon$ is a nondegenerate symmetric bilinear form on $V_\bullet$ and $b$ is invariant under $\fso ( V_\bullet ) = \fsp ( \Delta_1 ) \oplus \fsp ( \Delta_2 )$ since $\varepsilon$ is invariant under $\fsp ( \Delta )$. 

We define $\Omega \in \Lambda^4 V_\bullet^*$ such that
\begin{align}\label{eq:SimpleFTSOmega}
&\Omega ( \psi_1 \otimes \chi_1 , \psi_2 \otimes \chi_2 , \psi_3 \otimes \chi_3 , \psi_4 \otimes \chi_4 ) \nonumber \\
&= 2 \varepsilon ( \psi_1 , \psi_4 ) \varepsilon ( \psi_2 , \psi_3 ) \varepsilon ( \chi_1 , \chi_3 ) \varepsilon ( \chi_2 , \chi_4 ) - 2  \varepsilon ( \psi_2 , \psi_4 ) \varepsilon ( \psi_1 , \psi_3 ) \varepsilon ( \chi_2 , \chi_3 ) \varepsilon ( \chi_1, \chi_4 )~,
\end{align}
for all $\psi_1 \otimes \chi_1 , \psi_2 \otimes \chi_2 , \psi_3 \otimes \chi_3 , \psi_4 \otimes \chi_4 \in V_\bullet$. Up to a scalar multiple, $\Omega$ is unique since $\dim \, V_\bullet = 4$. 

The 3-bracket $[-,-,-]^F$ on $V_\bullet$ is defined such that
\begin{equation}\label{eq:SimpleFTS3bracketb}
b ( [ \psi_1 \otimes \chi_1 , \psi_2 \otimes \chi_2 , \psi_3 \otimes \chi_3 ]^F , \psi_4 \otimes \chi_4 ) = \Omega ( \psi_1 \otimes \chi_1 , \psi_2 \otimes \chi_2 , \psi_3 \otimes \chi_3 , \psi_4 \otimes \chi_4 )~,
\end{equation}
for all $\psi_1 \otimes \chi_1 , \psi_2 \otimes \chi_2 , \psi_3 \otimes \chi_3 , \psi_4 \otimes \chi_4 \in V_\bullet$. Clearly $[-,-,-]^F$ satisfies \eqref{eq:FTSskewsym} since $\Omega \in \Lambda^4 V_\bullet^*$. Since $b$ is nondegenerate, \eqref{eq:SimpleFTS3bracketb} is equivalent to 
\begin{equation}\label{eq:SimpleFTS3bracket}
[ \psi_1 \otimes \chi_1 , \psi_2 \otimes \chi_2 , \psi_3 \otimes \chi_3 ]^F = 2 \varepsilon ( \psi_2 , \psi_3 ) \varepsilon ( \chi_1 , \chi_3 ) \psi_1 \otimes \chi_2 - 2 \varepsilon ( \psi_1 , \psi_3 ) \varepsilon ( \chi_2 , \chi_3 ) \psi_2 \otimes \chi_1~,
\end{equation}
for all $\psi_1 \otimes \chi_1 , \psi_2 \otimes \chi_2 , \psi_3 \otimes \chi_3 \in V_\bullet$. Since $\dim\, \Delta = 2$, it follows that 
\begin{equation}\label{eq:SimpleFTSDelta}
\varepsilon ( \psi_1 , \psi_2 ) \psi_3 + \varepsilon ( \psi_2 , \psi_3 ) \psi_1 + \varepsilon ( \psi_3 , \psi_1 ) \psi_2 =0~,
\end{equation}
for all $\psi_1 , \psi_2 , \psi_3 \in \Delta$. Substituting \eqref{eq:SimpleFTSDelta} into \eqref{eq:SimpleFTS3bracket} implies
\begin{align}\label{eq:SimpleFTS3bracket2}
&[ \psi_1 \otimes \chi_1 , \psi_2 \otimes \chi_2 , \psi_3 \otimes \chi_3 ]^F \nonumber \\
&= \varepsilon ( \chi_1 , \chi_2 )  ( \varepsilon ( \psi_1 , \psi_3 ) \psi_2 + \varepsilon ( \psi_2 , \psi_3 ) \psi_1 ) \otimes \chi_3 - \varepsilon ( \psi_1 , \psi_2 ) \psi_3 \otimes ( \varepsilon ( \chi_1 , \chi_3 ) \chi_2 + \varepsilon ( \chi_2 , \chi_3 ) \chi_1 )~,
\end{align}
for all $\psi_1 \otimes \chi_1 , \psi_2 \otimes \chi_2 , \psi_3 \otimes \chi_3 \in V_\bullet$. Therefore
\begin{equation}\label{eq:SimpleFTS3bracketdF}
d^F ( \psi_1 \otimes \chi_1 , \psi_2 \otimes \chi_2 ) = \varepsilon ( \chi_1 , \chi_2 )  ( \psi_1 \otimes \psi_2^\flat + \psi_2 \otimes \psi_1^\flat )_+ - \varepsilon ( \psi_1 , \psi_2 ) ( \chi_1 \otimes \chi_2^\flat + \chi_2 \otimes \chi_1^\flat )_-~,
\end{equation}
for all $\psi_1 \otimes \chi_1 , \psi_2 \otimes \chi_2 \in V_\bullet$, where $X_+ = ( X , 0 )$ and $X_- = ( 0 , X )$ as elements of $\fso ( V_\bullet )$, for any $X \in \fsp ( \Delta )$. It follows that $d^F ( V_\bullet , V_\bullet ) \subset \fso ( V_\bullet )$.

If $X \in \fsp ( \Delta )$ and $\psi \in \Delta$ then
\begin{equation}\label{eq:SimpleFTSFlatX}
\psi^\flat X = - ( X \psi )^\flat
\end{equation}
since $\psi^\flat ( X \chi ) = \varepsilon ( \psi , X \chi ) = - \varepsilon ( X \psi , \chi ) = - ( X \psi )^\flat ( \chi )$, for all $\chi \in \Delta$. Therefore
\begin{align}\label{eq:SimpleFTS3bracketXdF}
[ X_+ , d^F ( \psi_1 \otimes \chi_1 , \psi_2 \otimes \chi_2 ) ] &\underset{\eqref{eq:SimpleFTS3bracketdF}}{=} \varepsilon ( \chi_1 , \chi_2 )  ( X \psi_1 \otimes \psi_2^\flat + X \psi_2 \otimes \psi_1^\flat - \psi_1 \otimes \psi_2^\flat X - \psi_2 \otimes \psi_1^\flat X )_+  \nonumber \\
&\underset{\eqref{eq:SimpleFTSFlatX}}{=} \varepsilon ( \chi_1 , \chi_2 )  ( X \psi_1 \otimes \psi_2^\flat + X \psi_2 \otimes \psi_1^\flat + \psi_1 \otimes ( X \psi_2 )^\flat + \psi_2 \otimes ( X \psi_1 )^\flat )_+ \nonumber \\
&\underset{\eqref{eq:SimpleFTS3bracketdF}}{=} d^F ( X_+ ( \psi_1 \otimes \chi_1) , \psi_2 \otimes \chi_2 ) + d^F ( \psi_1 \otimes \chi_1 , X_+ ( \psi_2 \otimes \chi_2 ) )~, \nonumber \\
[ X_- , d^F ( \psi_1 \otimes \chi_1 , \psi_2 \otimes \chi_2 ) ] &\underset{\eqref{eq:SimpleFTS3bracketdF}}{=} - \varepsilon ( \psi_1 , \psi_2 ) ( X \chi_1 \otimes \chi_2^\flat + X \chi_2 \otimes \chi_1^\flat - \chi_1 \otimes \chi_2^\flat X - \chi_2 \otimes \chi_1^\flat X )_- \nonumber \\
&\underset{\eqref{eq:SimpleFTSFlatX}}{=} - \varepsilon ( \psi_1 , \psi_2 ) ( X \chi_1 \otimes \chi_2^\flat + X \chi_2 \otimes \chi_1^\flat + \chi_1 \otimes ( X \chi_2 )^\flat + \chi_2 \otimes ( X \chi_1 )^\flat )_- \nonumber \\
&\underset{\eqref{eq:SimpleFTS3bracketdF}}{=} d^F ( X_- ( \psi_1 \otimes \chi_1) , \psi_2 \otimes \chi_2 ) + d^F ( \psi_1 \otimes \chi_1 , X_- ( \psi_2 \otimes \chi_2 ) )~,
\end{align}
for all $X \in \fsp ( \Delta )$ and $\psi_1 , \chi_1 , \psi_2 , \chi_2 \in \Delta$. Since $d^F ( V_\bullet , V_\bullet ) \subset \fso ( V_\bullet )$, it follows that $[-,-,-]^F$ satisfies \eqref{eq:FTSend} so $V_\bullet$ is a Filippov triple system and $\fso ( V_\bullet ) \subset \Der^+ V_\bullet$ is a derivation algebra of $V_\bullet$. It was proven in \cite{Ling} that there is, up to isomorphism, a unique simple Filippov triple system and it is easily verified that this is isomorphic to $V_\bullet$. 

Following the construction in Theorem~\ref{thm:BalPolALTS-FTS} defines $V = V_\bullet \oplus V_\bullet$ as a balanced polarised anti-Lie triple system with derivation algebra $D(V) \cong \fso ( V_\bullet )$. Furthermore, \eqref{eq:FTS-ALTS3bracketSkew} and \eqref{eq:SimpleFTS3bracketdF} imply
\begin{equation}\label{eq:SimpleFTS-ALTSdL}
d^L ( ( \psi_1 \otimes \chi_1 )_+ , ( \psi_2  \otimes \chi_2 )_- )_\bullet = \varepsilon ( \chi_1 , \chi_2 )  ( \psi_1 \otimes \psi_2^\flat + \psi_2 \otimes \psi_1^\flat )_+ - \varepsilon ( \psi_1 , \psi_2 ) ( \chi_1 \otimes \chi_2^\flat + \chi_2 \otimes \chi_1^\flat )_-~,
\end{equation}
for all $\psi_1 \otimes \chi_1 , \psi_2 \otimes \chi_2 \in V_\bullet$. Proposition~\ref{prop:BalPolALTS-Bal3GradedLSA} then implies $( V , D(V) )$ defines a balanced 3-graded Lie superalgebra $\fG$ with $\fG_{\bar 0} = D(V)$ and $\fG_{\bar 1}^\bullet = V_\bullet$. If we identify $D(V) \cong \fso ( V_\bullet )$ then the non-zero brackets for $\fG$ are precisely \eqref{eq:A1100}, \eqref{eq:A1101+} , \eqref{eq:A1101-} and \eqref{eq:A1111}. Therefore $\fG \cong {\bf A} ( 1 , 1 )$ as $3$-admissible Lie superalgebras. 
\end{example}

A Filippov triple system $V$ that is equipped with a nondegenerate $d(V,V)$-invariant symmetric bilinear form was referred to as a {\emph{metric Lie 3-algebra}} in \cite{deMedeiros:2008bf}. Any such metric Lie 3-algebra corresponds to a $3$-admissible Filippov triple system $V$ with derivation algebra $d(V,V)$. A structure theorem in \cite{deMedeiros:2008bf} says that every metric Lie 3-algebra can be obtained by taking orthogonal direct sums and double extensions of simple and one-dimensional Filippov triple systems. Combining this result with Theorem~\ref{thm:BalPolALTS-FTS} and Proposition~\ref{prop:BalPolALTS-Bal3GradedLSA} therefore provides a way of constructing a large class of $3$-admissible Lie superalgebras.
 
%%%%%%%%%%%%%%%%%%%%%%

\section*{Acknowledgements}

I would like to thank Andrew Beckett and Jos\'{e} Figueroa-O'Farrill for some entertaining discussions related to this work which took place during and after the workshop \lq Pseudo-Riemannian Geometry and Invariants in General Relativity' held in June 2022 at the University of Stavanger. Their visit was supported by the Research Council of Norway, Toppforsk grant no. 250367, held by Sigbj\o{}rn Hervik. I am especially grateful to Jos\'{e} for his continued encouragement to write this work up, and for presenting some of its early results on my behalf during the \lq Ternary Day 2' meeting held in Edinburgh on 2 December 2022.

%%%%%%%%%%%%%%%%%%%%%%

\appendix

%%%%%%%%%%%%%%%%%%%%%

\section{$\fsl (W)$-invariant symmetric bilinear forms on $W \oplus W^*$ and $S^2 W \oplus \Lambda^2 W^*$}
\label{sec:InvariantBilinearFormsOnRepresentationsOFSLW}

Let $e_1 ,..., e_N$ be a basis for a vector space $W$ of dimension $N>2$. For any $i,j,k \in \{ 1 ,..., N \}$, let $E_{ij} \in \fgl (W)$ be defined such that
\begin{equation}\label{eq:Eijek}
E_{ij} e_k = \delta_{jk} e_i~.
\end{equation}
It follows from \eqref{eq:Eijek} that $E_{ii} e_j = \delta_{ij} e_j$. For any $i \neq j$, let 
\begin{equation}\label{eq:Xij}
H_{ij} = E_{ii} - E_{jj} 
\end{equation}
so that
\begin{equation}\label{eq:Xijek}
H_{ij} e_k = ( \delta_{ik} - \delta_{jk} ) e_k~.
\end{equation}
It follows that $H_{ij} \in \fsl (W)$ since $\tr H_{ij} = \sum_{k=1}^N ( \delta_{ik} - \delta_{jk} ) = 1-1 = 0$. 

Let $f_1 ,..., f_N$ be a basis for $W^*$ defined such that 
\begin{equation}\label{eq:fiej}
f_i ( e_j ) = \delta_{ij}~,
\end{equation}
for all $i,j \in \{ 1 ,..., N \}$. Then
\begin{equation}\label{eq:Eij*fk}
(E_{ij}^* f_k ) ( e_l ) \underset{\eqref{eq:EndStar}}{=} f_k ( E_{ij} e_l ) \underset{\eqref{eq:Eijek}}{=} \delta_{jl} f_k ( e_i ) \underset{\eqref{eq:fiej}}{=} \delta_{jl} \delta_{ki} \underset{\eqref{eq:fiej}}{=} \delta_{ik} f_j ( e_l )~,
\end{equation}
for all $i,j,k,l \in \{ 1 ,..., N \}$, so
\begin{equation}\label{eq:Eijfk}
E_{ij} \cdot f_k \underset{\eqref{eq:DualRepGL}}{=} - E_{ij}^* f_k \underset{\eqref{eq:Eij*fk}}{=} - \delta_{ik} f_j
\end{equation}
and 
\begin{equation}\label{eq:Xijfk}
H_{ij} \cdot f_k \underset{\eqref{eq:Xij}}{=} ( E_{ii} - E_{jj} )\cdot f_k \underset{\eqref{eq:Eijfk}}{=} - ( \delta_{ik} - \delta_{jk} ) f_k~.
\end{equation}

%%%%%%%%%%%%%%%%%%%%%%

\subsection{$W \oplus W^*$}
\label{sec:WPlusW*}

Let $b$ be an $\fsl (W)$-invariant symmetric bilinear form on $W \oplus W^*$. For any $i,j,k \in \{ 1 ,..., N \}$ with $i \neq j$, invariance of $b$ under $H_{ij} \in \fsl (W)$ implies
\begin{align}
0 &= b ( H_{ij} e_i , e_k ) + b ( e_i , H_{ij} e_k ) \underset{\eqref{eq:Xijek}}{=} ( 1 + \delta_{ik} - \delta_{jk} ) b ( e_i , e_k )~,  \label{eq:WW*ee} \\
0 &= b ( H_{ij} e_i , f_k ) + b ( e_i , H_{ij} \cdot f_k ) \overset{\eqref{eq:Xijek}}{\underset{\eqref{eq:Xijfk}}{=}} ( 1 - \delta_{ik} + \delta_{jk} ) b ( e_i , f_k )~,  \label{eq:WW*ef} \\
0 &= b ( H_{ij} \cdot f_i , f_k ) + b ( f_i , H_{ij} \cdot f_k ) \underset{\eqref{eq:Xijfk}}{=} - ( 1 + \delta_{ik} - \delta_{jk} ) b ( f_i , f_k )~.  \label{eq:WW*ff}
\end{align}
It is always possible to choose $j \neq k$ in the expressions above. Doing this in \eqref{eq:WW*ee} and \eqref{eq:WW*ff} implies $b (W,W) =0$ and $b ( W^* , W^* ) = 0$. Doing this in \eqref{eq:WW*ef} implies $b ( e_i , f_k ) = 0$ unless $i=k$. Furthermore, for any $i \neq j$, invariance of $b$ under $E_{ij} \in \fsl (W)$ implies
\begin{equation}\label{eq:WW*efE}
0 = b ( E_{ij} e_j , f_i ) + b ( e_j , E_{ij} \cdot f_i ) \overset{\eqref{eq:Eijek}}{\underset{\eqref{eq:Eijfk}}{=}} b ( e_i , f_i ) - b ( e_j , f_j )~.
\end{equation}
Therefore 
\begin{equation}\label{eq:WW*bef}
b ( e_i , f_j ) = \lambda \delta_{ij} = \lambda f_j ( e_i )~,
\end{equation}
for some $\lambda \in \CC$. In other words, up to scaling, $b (W, W^* )$ is just the dual pairing between $W$ and $W^*$ which is, by construction, $\fsl (W)$-invariant.

Therefore any $\fsl (W)$-invariant symmetric bilinear form $b$ on $W \oplus W^*$ must be of the form
\begin{equation}\label{eq:WW*b}
b ( w_1 + \alpha_1 , w_2 + \alpha_2 ) = \lambda ( \alpha_2 ( w_1 ) + \alpha_1 ( w_2 ) )~,
\end{equation}
for all $w_1 , w_2 \in W$ and $\alpha_1 , \alpha_2 \in W^*$, in terms of some $\lambda \in \CC$.

%%%%%%%%%%%%%%%%%%%%%%

\subsection{$S^2 W \oplus \Lambda^2 W^*$}
\label{sec:S2WPlusWedge2W*}

Let $e_{ij} = e_i \otimes e_j + e_j \otimes e_i$ be a basis for $S^2 W$ and let $f_{ij} = f_i \otimes f_j - f_j \otimes f_i$ be a basis for $\Lambda^2 W^*$. For any $i,j,k,l \in \{ 1 ,..., N \}$, it follows that
\begin{align}
E_{ij} \cdot e_{kl} &= E_{ij} e_k \otimes e_l + e_l \otimes E_{ij} e_k + e_k \otimes E_{ij} e_l + E_{ij} e_l \otimes e_k \underset{\eqref{eq:Eijek}}{=} \delta_{jk} e_{il} + \delta_{jl} e_{ki}~,  \label{eq:Eijekl} \\
E_{ij} \cdot f_{kl} &= E_{ij} \cdot f_k \otimes f_l - f_l \otimes E_{ij} \cdot f_k + f_k \otimes E_{ij} \cdot f_l - E_{ij} \cdot f_l \otimes f_k \underset{\eqref{eq:Eijfk}}{=} - \delta_{ik} f_{jl} - \delta_{il} f_{kj}~. \label{eq:Eijfkl}
\end{align}
Furthermore, if $i \neq j$ then \eqref{eq:Xij}, \eqref{eq:Eijekl} and \eqref{eq:Eijfkl} imply 
\begin{align}
H_{ij} \cdot e_{kl} &= ( \delta_{ik} + \delta_{il} - \delta_{jk} - \delta_{jl} ) e_{kl}~,  \label{eq:Xijekl} \\
H_{ij} \cdot f_{kl} &= - ( \delta_{ik} + \delta_{il} - \delta_{jk} - \delta_{jl} ) f_{kl}~. \label{eq:Xijfkl}
\end{align}

Now let $b$ be an $\fsl (W)$-invariant symmetric bilinear form on $S^2 W \oplus \Lambda^2 W^*$. For any $i,j,k,l,m \in \{ 1 ,..., N \}$ with $i \neq j$, invariance of $b$ under $H_{ij} \in \fsl (W)$ implies
\begin{align}
0 &= b ( H_{ij} \cdot e_{ik} , e_{lm} ) + b ( e_{ik} , H_{ij} \cdot e_{lm} ) \underset{\eqref{eq:Xijekl}}{=} ( 1 + \delta_{ik} - \delta_{jk} +  \delta_{il} + \delta_{im} - \delta_{jl} - \delta_{jm} ) b ( e_{ik} , e_{lm} )~,  \label{eq:SWLWee} \\
0 &= b ( H_{ij} \cdot e_{ik} , f_{lm} ) + b ( e_{ik} , H_{ij} \cdot f_{lm} ) 
\overset{\eqref{eq:Xijekl}}{\underset{\eqref{eq:Xijfkl}}{=}} ( 1 + \delta_{ik} - \delta_{jk} -  \delta_{il} - \delta_{im} + \delta_{jl} + \delta_{jm} ) b ( e_{ik} , f_{lm} )~,  \label{eq:SWLWef} \\
0 &= b ( H_{ij} \cdot f_{ik} , f_{lm} ) + b ( f_{ik} , H_{ij} \cdot f_{lm} ) \underset{\eqref{eq:Xijfkl}}{=} - ( 1 + \delta_{ik} - \delta_{jk} +  \delta_{il} + \delta_{im} - \delta_{jl} - \delta_{jm} ) b ( f_{ik} , f_{lm} )~.  \label{eq:SWLWff}
\end{align}
If $N>4$ then it is always possible to choose $j \neq k,l,m$ in the expressions above. Doing this in \eqref{eq:SWLWee} and \eqref{eq:SWLWff} implies $b ( S^2 W , S^2 W ) =0$ and $b ( \Lambda^2 W^* , \Lambda^2 W^* ) = 0$. Doing this in \eqref{eq:SWLWef} implies $b ( e_{ik} , f_{lm} ) = 0$ unless $i \neq k$ and $i=l$ or $i=m$ (but not both since $f_{lm} = 0$ if $l=m$). Furthermore, for any $i \neq j$, invariance of $b$ under $E_{ij} \in \fsl (W)$ implies
\begin{equation}\label{eq:SWLWefE}
0 = b ( E_{ij} \cdot e_{jk} , f_{ji} ) + b ( e_{jk} , E_{ij} \cdot f_{ji} ) \overset{\eqref{eq:Eijekl}}{\underset{\eqref{eq:Eijfkl}}{=}}  b ( e_{ki} , f_{ji} ) + \delta_{jk} b ( e_{ji} , f_{ji} )~.
\end{equation}
If $j=k$ then \eqref{eq:SWLWefE} implies $b ( e_{ji} , f_{ji} ) = 0$. Substituting this back into \eqref{eq:SWLWefE} implies $b ( e_{ki} , f_{ji} ) = 0$. Therefore $b ( S^2 W , \Lambda^2 W^* ) =0$.

If $N=4$ and $i,k,l,m$ are not all different then it is always possible to choose $j \neq k,l,m$ in \eqref{eq:SWLWee}, \eqref{eq:SWLWef} and \eqref{eq:SWLWff}. Doing this in \eqref{eq:SWLWee} and \eqref{eq:SWLWff} implies $b ( e_{ik} , e_{lm} ) = 0$ and $b ( f_{ik} , f_{lm} ) = 0$. Doing this in \eqref{eq:SWLWef} implies $b ( e_{ik} , f_{lm} ) = 0$ as above (using \eqref{eq:SWLWefE}). If $i,k,l,m$ are all different then \eqref{eq:SWLWee} and \eqref{eq:SWLWff} are trivial while \eqref{eq:SWLWef} implies $b ( e_{ik} , f_{lm} ) = 0$ by choosing $j = l$. Therefore $b ( S^2 W , \Lambda^2 W^* ) =0$. Furthermore, $b ( S^2 W , S^2 W ) =0$ since
\begin{equation}\label{eq:SWLWeeE}
0 = b ( E_{ik} \cdot e_{kk} , e_{lm} ) + b ( e_{kk} , E_{ik} \cdot e_{lm} ) \underset{\eqref{eq:Eijekl}}{=} 2 b ( e_{ik} , e_{lm} )~.
\end{equation}
On the other hand, 
\begin{equation}\label{eq:SWLWffE}
0 = b ( E_{li} \cdot f_{lk} , f_{lm} ) + b ( f_{lk} , E_{li} \cdot f_{lm} ) \underset{\eqref{eq:Eijfkl}}{=}  - b ( f_{ik} , f_{lm} ) - b ( f_{lk} , f_{im} )
\end{equation}
implies $b ( f_{ik} , f_{lm} )$ is skewsymmetric in $il$, and therefore totally skewsymmetric in $iklm$. This means $b ( \Lambda^2 W^* , \Lambda^2 W^* )$ is defined by an element in $\Lambda^4 W$ which is unique (up to scaling) since $N=4$. Since the Lie group $\SL (W)$ preserves orientations on $W$, it follows that $\Lambda^N W$ is $\fsl (W)$-invariant (here for $N=4$).

If $N=3$ then $i,k,l,m$ cannot all be different. If more than one index is repeated or if an index is repeated more than once then it is always possible to choose $j \neq k,l,m$ in \eqref{eq:SWLWee}, \eqref{eq:SWLWef} and \eqref{eq:SWLWff}. Doing this in \eqref{eq:SWLWee} and \eqref{eq:SWLWff} implies $b ( e_{ik} , e_{lm} ) = 0$ and $b ( f_{ik} , f_{lm} ) = 0$. Doing this in \eqref{eq:SWLWef} implies $b ( e_{ik} , f_{lm} ) = 0$ as above (using \eqref{eq:SWLWefE}). The same conclusions also follow if one index is repeated just once (using only \eqref{eq:SWLWee}, \eqref{eq:SWLWef} and \eqref{eq:SWLWff}). Therefore $b ( S^2 W , S^2 W ) =0$, $b ( S^2 W , \Lambda^2 W^* ) =0$ and $b ( \Lambda^2 W^* , \Lambda^2 W^* ) = 0$.

To summarise, any $\fsl (W)$-invariant symmetric bilinear form $b$ on $S^2 W \oplus \Lambda^2 W^*$ is zero unless $N=4$. If $N=4$ then only the $b ( \Lambda^2 W^* , \Lambda^2 W^* )$ component of $b$ need not be zero and is defined by a unique (up to scaling) element in $\Lambda^4 W$. Clearly any such $b$ is not nondegenerate since $b ( S^2 W , - ) =0$.

%%%%%%%%%%%%%%%%%%%%%%

\bibliographystyle{utphys}
\bibliography{R3Alg}

\end{document}